\documentclass[11pt]{article}
\usepackage[margin=1in]{geometry}
\usepackage{amsthm}
\usepackage{amsmath}
\usepackage{amssymb}
\usepackage{enumerate}
\usepackage{graphicx}

\usepackage{tikz}
\usetikzlibrary{positioning, shapes.multipart}
\usepackage{mathtools}
\usepackage{algorithm}
\usepackage{algorithmic}
\usepackage{booktabs}
\usepackage{enumitem}
\usepackage{subcaption}
\newtheorem{theorem}{Theorem}
\newtheorem{lemma}[theorem]{Lemma}
\newtheorem{corollary}[theorem]{Corollary}

\newtheorem{observation}[theorem]{Observation}
\theoremstyle{definition}
\newtheorem{definition}[theorem]{Definition}
\newtheorem{problem}{Problem}
\newtheorem*{restatedDeterTheoremP}{Theorem \ref{thm:deterministic}}
\newtheorem*{restatedRandTheoremP}{Theorem \ref{thm:randomized-color-matching}}

\newcommand{\defn}[1]{\textbf{\emph{#1}}}
\newcommand{\dist}{\operatorname{dist}}
\newcommand{\miss}{\operatorname{miss}}
\newcommand{\cen}{\operatorname{cen}}

\title{More Efficient Parallel $(\Delta+1)$-Edge Coloring}
\author{Jeremy T. Fineman\\
    Georgetown University\\
    \texttt{jf474@georgetown.edu}
    \and 
    Seyed Ali Mohammadi\\
    Georgetown University\\
    \texttt{sm3933@georgetown.edu}}

\date{}

\begin{document}
\maketitle

\begin{abstract}
    This paper gives two parallel algorithms for $\Delta+1$ edge coloring, where $\Delta$ denotes the maximum degree of any vertex.  The first is a deterministic parallel algorithm with $\tilde{O}(\Delta^3)$ span and $\tilde{O}(m \Delta^3)$ work.  Our second algorithm and our main result is a more efficient randomized algorithm, achieving $\tilde{O}(\Delta^2)$ span and $\tilde{O}(m \Delta^2 )$ work both with high probability. These bounds substantially improve over the recent deterministic parallel algorithm of Elkin and Khuzman~\cite{elkin2026efficient}, which has $\tilde{O}(\Delta^4)$ span and $\tilde{O}(m \Delta^5)$ work. Our deterministic algorithm thus represents a $\tilde{O}(\Delta)$ improvement on span and $\tilde{O}(\Delta^2)$ on work compared to their algorithm, and our randomized algorithm improves the span and work by $\tilde{O}(\Delta^2)$ and $\tilde{O}(\Delta^3)$ factors, respectively. Moreover, our improvements do not come at the expense of larger logarithmic factors.
\end{abstract}

\section{Introduction}
Given a simple graph $G=(V,E)$, a \defn{proper edge coloring} is an assignment of colors to the edges such that no two edges incident to the same vertex receive the same color. A proper $k$-edge coloring uses colors from $[k]:=\{1,2,\ldots,k\}$, and the minimum such $k$ is called the \defn{edge chromatic number} of $G$. Vizing~\cite{vizing} proved constructively that every simple graph of maximum degree $\Delta$ admits a proper $(\Delta+1)$-edge coloring. In general, this bound is best possible: a complete graph with an odd number of vertices requires $\Delta+1$ colors. Moreover, deciding whether a graph of maximum degree $\Delta$ admits a $\Delta$ coloring or requires $\Delta+1$ colors is NP-complete~\cite{holyer1981np}. Naturally, allowing more colors makes the problem substantially easier. For example, the line graph $L(G)$ of the input graph $G$ has maximum degree at most $2\Delta-2$, and therefore a greedy vertex coloring of $L(G)$ with $2\Delta-1$ colors immediately gives a proper $(2\Delta-1)$-edge coloring of $G$. For smaller palettes, however, the vertex-coloring techniques do not simply translate to edge coloring. Our goal in this paper is to produce a parallel PRAM algorithm for $(\Delta+1)$-edge coloring that is efficient as possible.

\begin{problem}[$(\Delta+1)$-Edge Coloring]\label{prob:delta-plus-one}
Given a simple graph $G=(V,E)$ with maximum degree $\Delta$, find a proper edge coloring $\chi:E\to[\Delta+1]$.
\end{problem}

In the sequential setting, Vizing's theorem~\cite{vizing} gives an $O(mn)$-time algorithm for Problem~\ref{prob:delta-plus-one}, later improved to $\tilde{O}(m\sqrt n)$ time~ ~\cite{arjomandi1982efficient,gabow1985algorithms}.\footnote{The soft-O notation $\tilde{O}$ hides polylogarithmic factors; that is, for any constant $k$, $O(f(n)\log^{k} f(n)) = \tilde{O}(f(n))$.} Recently there has been significant progress~\cite{sinnamon2019fast,faster2024,even-faster2025,Near-Vizing2025}, culminating in a randomized algorithm running in $O(m\log\Delta)$-time with high probability~\cite{near-linear-time} and an $m^{1+o(1)}$-time deterministic algorithm~\cite{deter-almost-linear2026}. 

The parallel setting, however, has received comparatively less attention~\cite{karloff1987efficient,liang1996parallel,liang1997parallel}. Most recently, Elkin and Khuzman~\cite{elkin2026efficient} gave a deterministic parallel algorithm with $\tilde{O}(m\Delta^5)$ work and $\tilde{O}(\Delta^4)$ span. For large $\Delta$, this is far from work efficient. Elkin and Khuzman's algorithm builds on Vizing's algorithm~\cite{vizing}, which incrementally colors the graph. At a high level, Vizing's method chooses an uncolored edge (in a partially colored graph) and tries to color it; coloring an uncolored edge entails recoloring some already colored edges, while ensuring that the coloring remains proper.  A natural approach for parallelism is to color multiple uncolored edges simultaneously, but the recolorings that arise may interfere with each other.  Thus, a key challenge for parallelism is to find a large set of uncolored edges whose recolorings can proceed without interference.  Our algorithm also builds off of Vizing's approach, but we color more edges in parallel.

In more detail, most $(\Delta+1)$-coloring algorithms rely on the two fundamental structures from Vizing's proof~\cite{vizing}, fans and alternating paths (see Section~\ref{sec:prelim} for formal definitions). Suppose that the graph is partially colored with $\Delta+1$ colors and that $e=(u,v)$ is uncolored. Informally, a \defn{fan} for $e$ consists of the center $u$ together with a sequence of its neighbors such that colors \defn{missing} (that is, colors not used by any incident edge and hence available for use) at vertices of the fan can be shifted along the edges incident to $u$, i.e., a local recoloring within the neighborhood of $u$ called \defn{rotating} its fan. In the easier case, rotating the fan directly frees a color that can be assigned to $e$. In the harder case, a maximal \defn{alternating path} which is path that starts at $u$ consists of two alternating colors along its edge need to be found, \defn{flipping} the colors of the path, that is exchanging the two colors with each other on the path edges and the endpoints of the path, and then rotating the fan colors the uncolored edge.  Proving that this strategy always enables $e$ to be colored is the brilliance of Vizing's theorem~\cite{vizing}.

The most fundamental bottleneck for parallelizing a fan-based algorithm is that a fan for $e$ may intersect the fans from up to $\Theta(\Delta^2)$ other uncolored edges. (That is, a particular center vertex $u$ may share neighbors with $\Delta^2$ other center vertices.) Thus, when only coloring edges whose fans do not interfere, the best one could hope for in general is to color a $1/\Delta^2$ fraction of the uncolored edges in each parallel round of the algorithm. This limit from fan conflicts would necessitate a parallel algorithm with $\Omega(\Delta^2)$ rounds and hence $\Omega(\Delta^2)$ span and $\Omega(m\Delta^2)$ work. Our goal is to come as close to this fundamental bottleneck as possible. Indeed, ignoring logarithmic factors, our randomized algorithm essentially achieves this: in each round, every uncolored edge is colored with probability $\Omega(1/\Delta^2)$, yielding $\tilde{O}(\Delta^2)$ span and $\tilde{O}(m\Delta^2)$ work.  All prior parallel algorithms have significant overhead beyond this barrier. 

Fan neighborhoods are not the only source of interference; the alternating paths can also interfere with each other as well as with other fans. Each alternating path is associated with a pair of colors that should be flipped. Perhaps the most natural way to mitigate this source of interference is to consider only a single color pair in each round of the algorithm, which is what Elkin and Khuzman do~\cite{elkin2026efficient}. But there are $\Theta(\Delta^2)$ color pairs, so this strategy loses another $1/\Delta^2$ factor in the number of uncolored edges that could potentially be colored simultaneously in each parallel round. We address this interference and decrease the total number of parallel rounds. Elkin and Khuzman~\cite{elkin2026efficient} instead focus on the underlying parallel subroutines necessary to implement each round efficiently. We thus leverage much of their machinery when translating our high-level strategy to a PRAM realization. 

\paragraph{Color reduction.} 
As in much of the edge-coloring literature~\cite{liang1996parallel,elkin2026efficient,near-linear-time,sinnamon2019fast}, our algorithm first colors the graph with slightly more than $\Delta+1$ colors and then repeatedly decreases the number of colors in the palette by one, incorporating this color reduction subroutine into a divide-and-conquer algorithm for $(\Delta+1)$-coloring. (For an overview of the divide-and-conquer, see Section~\ref{subsec:solve_problem_1}.) Elkin and Khuzman~\cite{elkin2026efficient} use the same divide-and-conquer strategy at the top level, as do many prior parallel and sequential edge-coloring algorithms. The main contribution of this paper is an improved algorithm for the color reduction step (Problem~\ref{prob:color-reduction}), which improves the overall work and span of the resulting $(\Delta+1)$-coloring.

\begin{problem}[Color reduction]\label{prob:color-reduction}
Given a simple graph $G=(V,E)$ of maximum degree $\Delta$ and given a proper edge coloring $\chi:E\to[\Delta+k]$ using $\Delta + k$ colors, for constant integer $k > 1$, produce a coloring $\chi':E\to[\Delta+k-1]$ using one color fewer.
\end{problem}

One way to reduce the number of colors by one is to uncolor all edges of some color $c$ and recolor them using the remaining palette. Since the edges of every color in a proper edge coloring form a matching, the set of uncolored edges forms a matching as well. If we choose $c$ to be a color that appears on the fewest edges, then the resulting set of uncolored edges $U_\chi$ satisfies $|U_\chi|\leq m/(\Delta+k)$. Since $k$ is constant, we have $|U_\chi|=O(m/\Delta)$.\footnote{Choosing $c$ to be a color that appears on the fewest edges is needed only for our deterministic bound; our randomized algorithm can be applied to any color, since it remains valid for an uncolored matching of size $O(m)$.} Therefore, the main algorithmic task is the following: given a proper partial edge coloring whose set of uncolored edges $U_\chi$ forms a matching of size $O(m/\Delta)$, color all edges in $U_\chi$ using the current palette. Throughout, we focus on the case in which the palette consists of $\Delta+1$ colors and refer to this problem as \defn{color-matching}. An algorithm for color-matching that applies to any palette of size at least $\Delta+1$ yields an algorithm for Problem~\ref{prob:color-reduction} and consequently repeatedly applying this reduction yields an algorithm for Problem~\ref{prob:delta-plus-one}.

\subsection*{Results}
Our main results are the following: 
\begin{theorem}[Deterministic color reduction]\label{thm:deterministic} 
There exists a deterministic CRCW PRAM algorithm for Problem~\ref{prob:color-reduction} with $O(m\Delta^3\log^2 n)$ work and $O(\Delta^3\log^4 n)$ span.
\end{theorem}

\begin{theorem}[Randomized color reduction]\label{thm:randomized-color-matching} 
There exists a randomized CRCW PRAM algorithm for Problem~\ref{prob:color-reduction}  with $O(m\Delta^2\log^2 n\log\Delta)$ work and $O(\Delta^2\log^2 n\log\Delta)$ span, both with high probability.
\end{theorem}

\begin{corollary}[$(\Delta+1)$-Edge Coloring]\label{cor:main-delta-plus-one}
There exist deterministic and randomized algorithms for Problem~\ref{prob:delta-plus-one} on CRCW PRAM whose work and span match the bounds of Theorems~\ref{thm:deterministic} and~\ref{thm:randomized-color-matching}, respectively.
\end{corollary}

\subsection*{Overview of the Algorithms}

Throughout this section, fix a proper partial $(\Delta+1)$-edge coloring $\chi$ such that the set $U_\chi$ of uncolored edges forms a matching of size $O(m/\Delta)$, for each uncolored edge, we fix one of its endpoints as its \defn{center}, for each vertex in the graph we fix one of its missing colors as the \defn{designated missing color}. Our goal is to color all edges of $U_\chi$ using the same palette $[\Delta+1]$, without introducing any new colors while preserving properness. Vizing's proof shows how to extend such a coloring by one edge at a time; the difficulty in the parallel setting is to color many edges of $U_\chi$ simultaneously. After constructing a fan for an uncolored edge, the fan determines whether the edge can be colored by rotating the fan alone or whether a bichromatic alternating-path flip is needed first. In the latter case, the two colors $(\alpha,\beta)$ of the required path are determined by the fan itself: $\alpha$ is a color missing at the fan center (the designated missing color of the center), while $\beta$ is a color missing at its terminal leaf. Since the palette contains $\Delta+1$ colors, there are $\Theta(\Delta^2)$ possible color pairs. The case that most of the uncolored edges can be resolved by rotation alone is easier, so for the remainder of this overview we focus only on uncolored edges that require flipping an alternating path. 

Elkin and Khuzman~\cite{elkin2026efficient} process a single color pair in each iteration. They first build an auxiliary graph on the uncolored edges and take an independent set, ensuring that the fans of the selected uncolored edges are vertex-disjoint. Among these fans, they then select the color pair associated with the largest number of uncolored edges. For this group, the corresponding bichromatic paths are either vertex-disjoint or identical and can therefore be processed in parallel. They then build a second auxiliary conflict graph capturing the remaining interactions between these bichromatic paths and the selected fans, and extract a large independent set. The corresponding uncolored edges can then be colored simultaneously.

\paragraph{Processing many color-disjoint pairs at once.}
Our approach processes $\Theta(\Delta)$ different color pairs simultaneously. (This is a natural idea, and the resulting deterministic algorithm serves as a warmup for the remainder.) 
The potential concern here is that processing alternating paths for multiple color pairs could lead to more interference than before.  We argue that as long as the set of color pairs is ``color disjoint,'' then the total number of interferences remains similar --- the update for each uncolored edge conflicts with $O(\Delta^2)$ other uncolored edges. Thus, we can process $\Theta(\Delta)$ times more uncolored edges in each round.  Moreover, amortized across all rounds, the work per round of our deterministic algorithm is $\tilde{O}(m)$, yielding a $\Theta(\Delta)$-factor improvement over the $\tilde{O}(m\Delta)$ per-round work of Elkin and Khuzman's algorithm~\cite{elkin2026efficient} and a $\Theta(\Delta^2)$-factor improvement to total work.

In more detail, a set of color pairs is \defn{color-disjoint}  if each color occurs in at most one pair (i.e., the set of pairs is a matching of colors).  A set of color-disjoint pairs can contain as many as $\Theta(\Delta)$ color pairs. At a high level, our deterministic algorithm proceeds by repeating the following in each round: (1) fix any family of $\Theta(\Delta)$ color-disjoint sets that covers every color pair; (2) map uncolored edges to the color pair needed for their alternating path; (3) choose the color-disjoint set from the family that satisfies the most uncolored edges; (4) identify which uncolored edges interfere with each other by building a conflict graph; (5) find a large independent set in the conflict graph, and color the corresponding uncolored edges. The largest set from the family hits at least a $\Omega(1/\Delta)$ fraction of the uncolored edges, and at least a $\Omega(1/\Delta^2)$ fraction of those can be colored simultaneously, giving rise to $\tilde{O}(\Delta^3)$ rounds. 

To argue that the interference is $O(\Delta^2)$ per uncolored edge, we classify the conflicts into three categories, formalized in Section~\ref{subsec:conflicts} --- fan-fan conflicts, path-fan conflicts, and path-path conflicts. At a high level, a conflict occurs when two updates (rotating a fan or flipping a path) would make changes to the color of the same edge or change the designated missing color of the same vertex.  Notably, color-disjointness implies that all distinct maximal alternating paths are edge disjoint. Thus, there are effectively no path-path conflicts. Moroever, a fan conflicts with $O(\Delta)$ different paths on an edge and $O(\Delta)$ paths at a vertex end. The conflicts are thus dominated by the $O(\Delta^2)$ fan-fan conflicts for any particular uncolored edge.

\paragraph{Why our deterministic bound stops at $\Delta^3$ and how randomization helps.}
Each round of the deterministic algorithm colors an $\Omega(1/\Delta^3)$ fraction of the uncolored edges: a factor $1/\Delta$ is lost because only one of $\Theta(\Delta)$ color-pair classes is active at a time, and a further factor $1/\Delta^2$ is lost because the fan-fan conflict degree is as large as $\Theta(\Delta^2)$. Suppose we instead sample a color-disjoint set of $\Theta(\Delta)$ color pairs at random and consider only the uncolored edges whose required color pair (for their alternating paths) lies in this set. Since only about a $1/\Delta$ fraction of all color pairs are selected, one might hope that this sampling already thins the uncolored edges enough that the number of fan-fan conflicts drops from $O(\Delta^2)$ to $O(\Delta)$. This conclusion, however, is false: if nearby conflicting fans all share the same color pair, then whenever that color pair is selected we can still only allow a $1/\Delta^2$ fraction of those fans to proceed in parallel. On the other hand, in the lucky case that the required color pairs are independent and uniformly random, a random set would indeed contain only $O(\Delta)$ fan-fan conflicts. The key idea of our randomized algorithm is thus to somewhat randomize the required color pairs for each uncolored edges, which we do by providing an efficient  procedure that randomizes the designated missing colors at the fan centers (the first color of the required pairs).

\paragraph{Randomizing the missing colors.}  The goal of this step is to ensure that the designated missing color at each center is roughly uniformly random. More accurately, here we would like to ensure that if $u$ is the center of an uncolored edge and has the designated missing color $\alpha$, then at most $O(\Delta)$ center vertices in $u$'s 2-hop neighborhood also have $\alpha$ as their designated missing color. Thus, $u$ does not have too many fan-fan conflicts with nearby fans that involve the same color pair. (We later address the different-color fan-fan conflicts.) Interestingly, the efficient sequential algorithm of Assadi et al.~\cite{near-linear-time} benefits from doing the opposite; they instead concentrate the missing colors  through a process they call \emph{popularization}, whereby the goal is to ensure that a large fraction of vertices are missing the same color. Popularization seems to exacerbate the challenges in parallelism. Instead our parallel algorithm benefits from spreading out the missing colors out at random.

The main idea of our algorithm for randomizing colors is as follows. Consider a vertex $u$ that has missing color $\gamma$. Consider the binary representation of $\gamma$, choose any bit index $i$, and let $\gamma'$ be the color obtained by changing the $i$th bit of $\gamma$.  There is a unique (possibly empty) maximal $(\gamma,\gamma')$ alternating path $P$ that has $u$ as an endpoint. Flipping the colors on $P$ would result in $u$ missing $\gamma'$.  Instead, suppose we flip $P$ with probability $1/2$, then $u$ ends with missing color $\gamma$ or $\gamma'$, each with probability~$1/2$.  In other words, the $i$th bit of $u$'s missing color is uniformly random, and all other bits in its missing color are unchanged.  

Our algorithm builds off of this idea and proceeds roughly as follows in $O(\lg \Delta)$ rounds.  In the $i$th round, consider bit $i$.  Match each color with the color obtained by changing the $i$th bit, giving a collection of color pairs.  For each alternating path for each color pair, flip the path independently with probability $1/2$.  Importantly, these color pairs are color disjoint by construction, so we can perform flipping procedure for all such paths in parallel.  Note that although the paths are flipped independently, this recoloring process is not entirely independent --- the missing color for the other end of path $P$ changes if and only if $u$'s missing color changes. Nevertheless, we argue that it is independent enough to achieve the end goal. 

\paragraph{Random matching of the palette.} Randomizing the missing colors at the centers reduces the same-color fan-fan conflicts. To control the remaining different-color conflicts, we choose a uniformly random matching of the palette into color-disjoint pairs and order every pair independently at random. An uncolored edge that requires an alternating path is \defn{active} if the ordered pair required by its path appears in this  ordered matching. An uncolored edge is activated with probability $\Theta(1/\Delta)$. Consequently, we show that among the $\Theta(\Delta^2)$ different-color fans that may conflict with the fan of a fixed active edge, only $O(\Delta)$ remain active with constant probability. Together with randomizing the missing colors, the random matching reduces the number of relevant conflicts of a fixed active uncolored edge from $O(\Delta^2)$ to $O(\Delta)$ with constant probability.

\paragraph{Random sampling and final conflict removal.} After the last two steps, conditioned on a fixed uncolored edge being active, with constant probability only $O(\Delta)$ other active uncolored edges can block it. We therefore further sample every active edge independently with probability $\Theta(1/\Delta)$. For each remaining conflict between two sampled uncolored edges, we orient the conflict toward one and remove the other. Unlike in the deterministic algorithm, we do not need to build a conflict graph and extract a large independent set from it; where instead an additional conflict-removal step gives us the set of conflict-free uncolored edges. The required color pairs of the surviving uncolored edges still have pairwise color-disjoint required color pairs, so their alternating-path operations and coloring the corresponding edges can be done in parallel. Combining the $\Theta(1/\Delta)$ activation probability with the $\Theta(1/\Delta)$ sampling probability and the constant survival probability from the final conflict removal step, we obtain that each uncolored edge is colored with probability $\Omega(1/\Delta^2)$ in each round. Hence $O(\Delta^2\log n)$ rounds suffice to color all edges of $U_\chi$ with high probability.

\subsection*{Related Work}
Karloff and Shmoys~\cite{karloff1987efficient} presented several parallel algorithms for edge coloring. In particular, they gave a deterministic algorithm for $(\Delta+1)$-edge coloring with $O(m\Delta^6\log^4 n+m\Delta^8\log n)$ work and $O(\Delta^5\log^4 n+\Delta^7\log n)$ span.~\footnote{All previous parallel results are restated in terms of work and span.} They also gave a randomized algorithm for $\Delta+\tilde{O}(\sqrt{\Delta})$ edge coloring with $\tilde{O}(m)$ work and polylog span. Liang et al.~\cite{liang1996parallel} later presented a deterministic algorithm for $(\Delta+1)$-edge coloring and claimed a span bound of $\tilde{O}(\Delta^{3.5})$. However, Elkin and Khuzman~\cite{elkin2026efficient}  recently showed an error in their analysis and showed that the algorithm instead has $O\bigl((n\Delta^3+n^2)(\Delta^{4.5}\log^3\Delta\log n+\Delta^4\log^4 n)\bigr)$ work and $O(\Delta^{4.5}\log^3\Delta\log n+\Delta^4\log^4 n)$ span. Liang et al.~\cite{liang1997parallel} also gave another deterministic algorithm with $O((m+n)\Delta^9\log^2 n)$ work and $O(\Delta^9\log^2 n)$ span. Finally, Elkin and Khuzman~\cite{elkin2026efficient} developed several deterministic parallel algorithms for $(\Delta+1)$-edge coloring. Their algorithm with the smallest span in general graphs costs $O(m\Delta^5\log^4 n)$ work and $O(\Delta^4\log^4 n)$ span. We improve these bounds in both the deterministic and randomized settings. Our deterministic algorithm costs $O(m\Delta^3\log^2 n)$ work and $O(\Delta^3\log^4 n)$ span. Our randomized algorithm costs $O(m\Delta^2\log^2 n\log\Delta)$ work and $O(\Delta^2\log^2 n\log\Delta)$ span, both with high probability. Our algorithms improve both the $\Delta$-dependent factors and the logarithmic factors in the runtime compared with previous results. Table~\ref{tab:related-work} summarizes these results.

\begin{table}[t]
\centering
\small
\begin{tabular}{lccc}
\toprule
\textbf{Algorithm} & \textbf{Work} & \textbf{Span} & \textbf{ } \\
\midrule
Karloff and Shmoys*~\cite{karloff1987efficient}
& $\tilde{O}(m\Delta^8)$
& $\tilde{O}(\Delta^7)$
& Deterministic \\
Liang et al.*~\cite{liang1996parallel}
& $\tilde{O}(n\Delta^{7.5}+n^2\Delta^{4.5})$
& $\tilde{O}(\Delta^{4.5})$
& Deterministic \\
Liang et al.*~\cite{liang1997parallel}
& $\tilde{O}(m\Delta^9)$
& $\tilde{O}(\Delta^9)$
& Deterministic \\
Elkin and Khuzman*~\cite{elkin2026efficient}
& $\tilde{O}(m\Delta^5)$
& $\tilde{O}(\Delta^4)$
& Deterministic \\
Ours
& $\tilde{O}(m\Delta^3)$
& $\tilde{O}(\Delta^3)$
& Deterministic \\
Ours
& $\tilde{O}(m\Delta^2)$
& $\tilde{O}(\Delta^2)$
& Randomized \\
\bottomrule
\end{tabular}

\vspace{0.5em}
\caption{Summary of the parallel $(\Delta+1)$-edge coloring algorithms.
\normalfont\footnotesize* Restated in terms of work and span.}
\label{tab:related-work}
\end{table}

Edge coloring has also received substantial attention in several other computational settings. In the sequential setting, a long line of work has studied algorithms for $(\Delta+1)$-edge coloring~\cite{vizing,arjomandi1982efficient,gabow1985algorithms,faster2024,even-faster2025,near-linear-time,deter-almost-linear2026}, culminating in the randomized near-linear time algorithm of Assadi et al.~\cite{near-linear-time}. In the distributed setting, much of the progress on $(\Delta+1)$-edge coloring has focused on bounded- or low-degree graphs, since their algorithms incur a large $\Delta^{O(1)}$ dependence in their round complexity~\cite{bernshteyn2022fast,christiansen2023power,bernshteyn2025fast}. 
In both the parallel and distributed settings, substantially faster algorithms are known when additional colors or when using $(1+\epsilon)\Delta$ colors are allowed~\cite{elkin2026efficient,liang1995fast,furer1996parallel,davies2023improved,dubhashi1998near,panconesi2001some,elkin20142delta,fischer2017deterministic,ghaffari2018deterministic,balliu2022distributed,chang2017complexity}.
Additionally, edge coloring has also been studied extensively in other models, including the dynamic~\cite{christiansen2023power,christiansen2026deterministic,bhattacharya2024nibbling,bhattacharya2018dynamic,barenboim2017fully,duan2019dynamic} and online~\cite{dudeja2025randomized,blikstad2024online,blikstad2025online,bhattacharya2020online,saberi2021greedy,cohen2019tight} settings.

\subsection*{Roadmap}
The remainder of the paper is organized as follows. Section~\ref{sec:prelim} introduces the necessary preliminaries and notation, together with the lemmas used throughout the paper. In Section~\ref{sec:BB-deter-alg}, we show how multiple color-disjoint pairs can be processed in parallel and discuss their conflicts and the corresponding conflict graph, and then present our deterministic algorithm. Section~\ref{sec:randomized} presents our randomized algorithm and the procedures it invokes, including the randomization of missing colors. Additional technical details and proofs are deferred to Appendix~\ref{sec:appendix}.

\section{Preliminaries}\label{sec:prelim}

Let $G=(V,E)$ be a simple graph with $n=|V|$ vertices, $m=|E|$ edges, maximum degree $\Delta$, and $\chi:E\to[\Delta+1]\cup\{\bot\}$ a proper partial edge coloring, where an edge $e$ with $\chi(e)=\bot$ is uncolored. Let $U_\chi:=\{e\in E:\chi(e)=\bot\}$ be the set of uncolored edges and $\lambda_\chi:=|U_\chi|$. For a vertex $v$, $\miss_\chi(v):=[\Delta+1]\setminus\{\chi(e): e\text{ incident to }v,\ \chi(e)\neq\bot\}$ is the set of colors \defn{missing} (available) at $v$. Since $v$ has at most $\Delta$ incident edges, $\miss_\chi(v)\neq\varnothing$ at all times; for each vertex $v$, we fix one of its missing colors and call it the \defn{designated missing color}, denoted by $\varphi(v)$. Assume throughout that $U_\chi$ is a matching (as would be true when applying the color-reduction approach) and $\lambda_\chi = O(m/\Delta)$; then we assign to each uncolored edge $e=(u,v)\in U_\chi$ a distinguished endpoint $\cen(e)\in\{u,v\}$, called its \defn{center}, so that distinct uncolored edges have distinct centers. Throughout, we use the Vizing fan and alternating paths, which have been used extensively in the literature~\cite{vizing,gabow1985algorithms,arjomandi1982efficient,sinnamon2019fast,near-linear-time};  We define the fan based on the designated missing color of the vertices as follows.

\begin{definition}[Vizing fan or c-fan]\label{def:vizing-fan}
A \defn{Vizing fan} is a sequence
\[
F=(u,\alpha),(v_1,c_1),\dots,(v_k,c_k)
\]
where $u,v_1,\dots,v_k$ are distinct vertices and $\alpha,c_1,\dots,c_k\in[\Delta+1]$ are colors such that:
\begin{enumerate}[label=(\arabic*)]
    \item $\alpha = \varphi(u)$ is the designated missing color of $u$ and $c_i = \varphi(v_i)$ for $i\in[k]$;
    \item $v_1,\dots,v_k$ are distinct neighbors of $u$;
    \item $\chi(u,v_1)=\bot$ and $\chi(u,v_i)=c_{i-1}$ for all $i>1$;
    \item either $c_k\in\miss_\chi(u)$ or $c_k\in\{c_1,\dots,c_{k-1}\}$.
\end{enumerate}
\end{definition}

We say that $F$ is \defn{$\alpha$-primed} has \defn{center} $u$, and \defn{leaves} $v_1,\dots,v_k$. Its edges $(u,v_1),\dots,(u,v_k)$ are the \defn{fan edges}, and $V(F):=\{u,v_1,\dots,v_k\}$ is its vertex set. If $c_k\in\miss_\chi(u)$, we call the fan \defn{trivial}, Figure~\ref{fig:trivial-fan} shows a trivial fan where the missing color of the last leaf vertex $v_5$ is the same as the missing color of the center $u$ i.e., $c_5=\alpha$: \defn{rotating} the colors along the fan so that $(u,v_k)$ becomes uncolored (formally, setting $\chi(u,v_i)\leftarrow c_i$ for $i=1,\dots,k-1$) and then coloring $(u,v_k)$ by $c_k$, which is missing at both $u$ and $v_k$, colors the originally uncolored edge $(u,v_1)$ without using any alternating path. For instance, after rotating the trivial fan of Figure~\ref{fig:trivial-fan} the edge $(u,v_5)$ is colored with $\alpha$; the designated missing colors of the vertices of the fan are updated accordingly, resulting in Figure~\ref{fig:after-rotate}. 
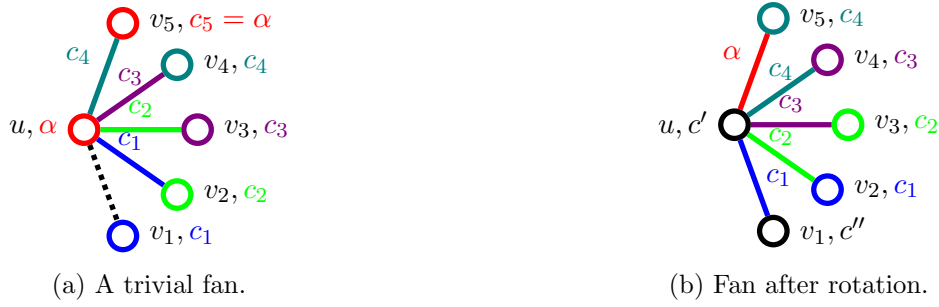
\begin{figure}[ht]
    \centering
    \begin{subfigure}[b]{0.48\columnwidth}
        \centering
%    \scalebox{.9}{
        \begin{tikzpicture}

            \node[circle, draw=red, fill=white, line width=0.7mm,  label=left:{$u,\textcolor{red}{\alpha}$}] at (0:0cm) (u) {};
            \node[circle, draw=blue, fill=white, line width=0.7mm, label=right:{$v_1,\textcolor{blue}{c_1}$}] at (290:1.5cm) (v1) {};
            \node[circle, draw=green, fill=white, line width=0.7mm, label=right:{$v_2,\textcolor{green}{c_2}$}] at (325:1.5cm) (v2) {};
            \node[circle, draw=violet, fill=white, line width=0.7mm, label=right:{$v_3,\textcolor{violet}{c_3}$}] at (0:1.5cm) (v3) {};
            \node[circle, draw=teal, fill=white, line width=0.7mm, label=right:{$v_4,\textcolor{teal}{c_4}$}] at (35:1.5cm) (v4) {};
            \node[circle, draw=red, fill=white, line width=0.7mm, label=right:{$v_5,\textcolor{red}{c_5=\alpha}$}] at (70:1.5cm) (v5) {};
            \draw[dotted, line width=0.7mm] (u) -- (v1);
            \draw[blue, line width=0.7mm] (u) -- node[above] {$\textcolor{blue}{c_1}$} (v2);
            \draw[green, line width=0.7mm] (u) -- node[above] {$\textcolor{green}{c_2}$} (v3);
            \draw[violet, line width=0.7mm] (u) -- node[above] {$\textcolor{violet}{c_3}$} (v4);
            \draw[teal, line width=0.7mm] (u) -- node[above left] {$\textcolor{teal}{c_4}$} (v5);
        \end{tikzpicture}
%}
        \caption{A trivial fan.}
        \label{fig:trivial-fan}
    \end{subfigure}
    \hfill
    \begin{subfigure}[b]{0.48\columnwidth}
        \centering
%\scalebox{.9}{
        \begin{tikzpicture}

            \node[circle, draw=black, fill=white, line width=0.7mm, label=left:{$u,c'$}] at (0:0cm) (ur) {};
            \node[circle, draw=black, fill=white, line width=0.7mm, label=right:{$v_1,c''$}] at (290:1.5cm) (v1r) {};
            \node[circle, draw=blue, fill=white, line width=0.7mm, label=right:{$v_2,\textcolor{blue}{c_1}$}] at (325:1.5cm) (v2r) {};
            \node[circle, draw=green, fill=white, line width=0.7mm,label=right:{$v_3,\textcolor{green}{c_2}$}] at (0:1.5cm) (v3r) {};
            \node[circle, draw=violet, fill=white, line width=0.7mm, label=right:{$v_4,\textcolor{violet}{c_3}$}] at (35:1.5cm) (v4r) {};
            \node[circle, draw=teal, fill=white, line width=0.7mm, label=right:{$v_5,\textcolor{teal}{c_4}$}] at (70:1.5cm) (v5r) {};
            \draw[blue, line width=0.7mm] (ur) -- node[right] {$\textcolor{blue}{c_1}$} (v1r);
            \draw[green, line width=0.7mm] (ur) -- node[above] {$\textcolor{green}{c_2}$} (v2r);
            \draw[violet, line width=0.7mm] (ur) -- node[above] {$\textcolor{violet}{c_3}$} (v3r);
            \draw[teal, line width=0.7mm] (ur) -- node[above] {$\textcolor{teal}{c_4}$} (v4r);
            \draw[red, line width=0.7mm] (ur) -- node[above left] {$\textcolor{red}{\alpha}$} (v5r);

        \end{tikzpicture}
%}
        \caption{Fan after rotation.}
        \label{fig:after-rotate}
    \end{subfigure}
    \caption{The left figure shows a trivial fan centered at $u$, with $\alpha$ as its missing color. For each $i$, $v_i$ denotes the $i$-th leaf and $c_i$ denotes the missing color at $v_i$. Each colored edge is assigned the corresponding color $c_i$, while the dashed edge is uncolored. The right figure shows the fan after the fan rotation, where $c'$ and $c''$ denote the updated designated missing colors of $u$ and $v_1$.}
    \label{fig:fan-and-rotate}
\end{figure}

\begin{definition}[$\{\alpha,\beta\}$-alternating path]\label{def:alt-path}
For distinct colors $\alpha,\beta$, an \defn{$\{\alpha,\beta\}$-alternating path} is a path whose edges are colored $\alpha$ and $\beta$ alternately. It starts at $u$ if $u$ is one endpoint and one of $\alpha,\beta$ is missing at $u$, and it is \defn{maximal} if it cannot be extended; a maximal path starting at $u$ ends at a vertex $w$ at which one of $\alpha,\beta$ is also missing. \defn{Flipping} such a path exchanges the two colors $\alpha,\beta$ on its edges and the endpoints.
\end{definition}

For a proper partial coloring $\chi$ and a color pair $p=\{\alpha,\beta\}$ of distinct colors, let $E_\chi(p)$ be the set of edges with colors $\alpha$ or $\beta$, and let $H_\chi(p)$ be the subgraph of $G$ induced by edges $E_\chi(p)$. Every vertex has at most one incident edge of each color; therefore vertices in $H_\chi(p)$ have maximum degree at most two. Its components are therefore paths or even cycles, with the two colors alternating along paths and cycles. Consider a vertex in $H_\chi(p)$. If $\alpha\in\miss_\chi(u)$ or $\beta\in\miss_\chi(u)$, then $u$ is an endpoint of its component. Let $K_\chi(p,u)$ be the component of $H_\chi(p)$ containing $u$.

\begin{theorem}[Vizing's theorem~\cite{vizing}]\label{thm:vizing}
Every simple graph $G$ of maximum degree $\Delta$ admits a proper $(\Delta+1)$-edge coloring.
\end{theorem}

We next briefly explain the intuition behind Vizing's theorem and the role of the fans and alternating paths. The proof posits roughly the following algorithm, which maintains the invariant that $\chi$ is a partial $(\Delta+1)$-edge coloring. To fully color the graph, repeat the following. Choose any uncolored edge $(u,v)$. Let $\alpha=\varphi(u)$ be the designated missing color of $u$ and construct an $\alpha$-primed Vizing fan with first leaf $v_1=v$. If the fan $F$ is trivial, then color $(u,v)$ immediately by rotating the colors along the fan, as explained above (see Figure~\ref{fig:fan-and-rotate} for an example).
Otherwise, the fan is nontrivial, and by condition~(4) of Definition~\ref{def:vizing-fan} its terminal color satisfies $c_k=c_j$ for some $j<k$.  Let $P$ be the maximal $\{\alpha,c_k\}$-alternating path starting at $u$. Flipping the colors of this path makes the fan trivial by changing the missing color of the center from $\alpha$ to $c_k$. Since there are two leaves $v_k$ and $v_j$ whose missing color equals the second color of the path, namely $c_k=c_j$, even if the path terminates at the fan itself, at either $v_k$ or $v_j$, the other leaf still has missing color $c_k=c_j$. Therefore, either way, the fan after the path flip becomes trivial and the uncolored edge can be colored by fan rotation. This process ends after every uncolored edge has been colored.  

Figure~\ref{fig:nontrivial-fan} shows a nontrivial fan with $c_5=c_3$ and $\alpha \neq c_i$ for $1 \leq i \leq 5$, its maximal $\{\alpha,c_3\}$-alternating path is $u-v_4-v_6-v_7-v_8-v_9$ that ends at $v_9$, Figure~\ref{fig:path-flip} shows the coloring after the maximal path has flipped so that $c_3$ is now missing at $u$ and is equal to the missing color of at least one leaf, then the (prefix) fan is rotated and the uncolored edge $(u,v_1)$ is colored, Figure~\ref{fig:rotate-and-color} shows the coloring after these operations. Note that the alternating path may terminate at one of the leaves of the fan. For example, in Figure~\ref{fig:nontrivial-fan}, suppose that the path terminates at $v_3$ instead of $v_9$ (with at least another vertex in between). Flipping the path changes the missing color at $v_3$ from $c_3$ to $\alpha$. However, by the definition of a nontrivial fan, the terminal missing color is repeated, so there exists another leaf at which $c_3$ is missing; here, this leaf is $v_5$. Since $c_5=c_3$ remains missing at $v_5$, the fan can still be rotated using $v_5$ to color the uncolored edge.

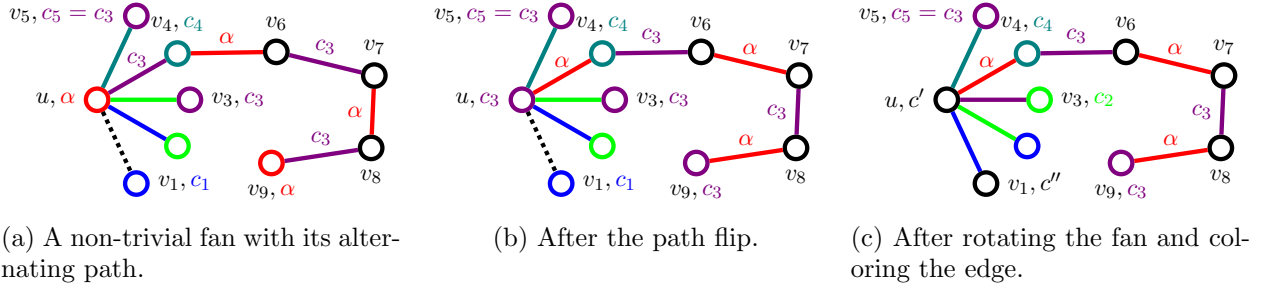
\begin{figure*}[t]
    \centering

    \begin{subfigure}[t]{0.32\textwidth}
        \centering
        \resizebox{\linewidth}{!}{%
            \begin{tikzpicture}

            \node[circle, draw=red, fill=white, line width=0.7mm, label=left:{$u,\textcolor{red}{\alpha}$}] at (0:0cm) (u) {};
            \node[circle, draw=blue, fill=white, line width=0.7mm, label=right:{$v_1,\textcolor{blue}{c_1}$}] at (295:1.5cm) (v1) {};
            \node[circle, draw=green, fill=white, line width=0.7mm, label=right:{}] at (330:1.5cm) (v2) {};
            \node[circle, draw=violet, fill=white, line width=0.7mm, label=right:{$v_3,\textcolor{violet}{c_3}$}] at (0:1.5cm) (v3) {};
            \node[circle, draw=teal, fill=white, line width=0.7mm, label=above:{$v_4,\textcolor{teal}{c_4}$}] at (30:1.5cm) (v4) {};
            \node[circle, draw=violet, fill=white, line width=0.7mm, label=left:{$v_5,\textcolor{violet}{c_5=c_3}$}] at (65:1.5cm) (v5) {};

            \node[circle, draw=black, fill=white, line width=0.7mm, label=above:{$v_6$}] at (15:3cm) (v6) {};
            \node[circle, draw=black, fill=white, line width=0.7mm, label=above:{$v_7$}] at (5:4.5cm) (v7) {};
            \node[circle, draw=black, fill=white, line width=0.7mm, label=below:{$v_8$}] at (350:4.5cm) (v8) {};
            \node[circle, draw=red, fill=white, line width=0.7mm, label=below:{$v_9,\textcolor{red}{\alpha}$}] at (340:3cm) (v9) {};

            \draw[dotted, line width=0.7mm] (u) -- (v1);
            \draw[blue, line width=0.7mm] (u) -- (v2);
            \draw[green, line width=0.7mm] (u) -- (v3);
            \draw[violet, line width=0.7mm] (u) --node[above] {$\textcolor{violet}{c_3}$} (v4);
            \draw[teal, line width=0.7mm] (u) -- (v5);
            \draw[red, line width=0.7mm] (v4) -- node[above] {$\textcolor{red}{\alpha}$} (v6);
            \draw[violet, line width=0.7mm] (v6) --node[above] {$\textcolor{violet}{c_3}$} (v7);
            \draw[red, line width=0.7mm] (v7) -- node[left] {$\textcolor{red}{\alpha}$} (v8);
            \draw[violet, line width=0.7mm] (v8) --node[above] {$\textcolor{violet}{c_3}$} (v9);

        \end{tikzpicture}
        }
        \caption{A non-trivial fan with its alternating path.}
        \label{fig:nontrivial-fan}
    \end{subfigure}
    \hfill
    \begin{subfigure}[t]{0.32\textwidth}
        \centering
        \resizebox{\linewidth}{!}{%
            \begin{tikzpicture}

            \node[circle, draw=violet, fill=white, line width=0.7mm, label=left:{$u,\textcolor{violet}{c_3}$}] at (0:0cm) (u) {};
            \node[circle, draw=blue, fill=white, line width=0.7mm, label=right:{$v_1,\textcolor{blue}{c_1}$}] at (295:1.5cm) (v1) {};
            \node[circle, draw=green, fill=white, line width=0.7mm, label=right:{}] at (330:1.5cm) (v2) {};
            \node[circle, draw=violet, fill=white, line width=0.7mm, label=right:{$v_3,\textcolor{violet}{c_3}$}] at (0:1.5cm) (v3) {};
            \node[circle, draw=teal, fill=white, line width=0.7mm, label=above:{$v_4,\textcolor{teal}{c_4}$}] at (30:1.5cm) (v4) {};
            \node[circle, draw=violet, fill=white, line width=0.7mm, label=left:{$v_5,\textcolor{violet}{c_5=c_3}$}] at (65:1.5cm) (v5) {};

            \node[circle, draw=black, fill=white, line width=0.7mm, label=above:{$v_6$}] at (15:3cm) (v6) {};
            \node[circle, draw=black, fill=white, line width=0.7mm, label=above:{$v_7$}] at (5:4.5cm) (v7) {};
            \node[circle, draw=black, fill=white, line width=0.7mm, label=below:{$v_8$}] at (350:4.5cm) (v8) {};
            \node[circle, draw=violet, fill=white, line width=0.7mm, label=below:{$v_9,\textcolor{violet}{c_3}$}] at (340:3cm) (v9) {};

            \draw[dotted, line width=0.7mm] (u) -- (v1);
            \draw[blue, line width=0.7mm] (u) -- (v2);
            \draw[green, line width=0.7mm] (u) -- (v3);
            \draw[red, line width=0.7mm] (u) -- node[above] {$\textcolor{red}{\alpha}$} (v4);
            \draw[teal, line width=0.7mm] (u) -- (v5);
            \draw[violet, line width=0.7mm] (v4) --node[above] {$\textcolor{violet}{c_3}$} (v6);
            \draw[red, line width=0.7mm] (v6) -- node[above] {$\textcolor{red}{\alpha}$} (v7);
            \draw[violet, line width=0.7mm] (v7) --node[left] {$\textcolor{violet}{c_3}$} (v8);
            \draw[red, line width=0.7mm] (v8) -- node[above] {$\textcolor{red}{\alpha}$} (v9);

        \end{tikzpicture}
        }
        \caption{After the path flip.}
        \label{fig:path-flip}
    \end{subfigure}
    \hfill
    \begin{subfigure}[t]{0.32\textwidth}
        \centering
        \resizebox{\linewidth}{!}{%
            \begin{tikzpicture}
            \node[circle, draw=black, fill=white, line width=0.7mm, label=left:{$u,c'$}] at (0:0cm) (u) {};
            \node[circle, draw=black, fill=white, line width=0.7mm, label=right:{$v_1,c''$}] at (295:1.5cm) (v1) {};
            \node[circle, draw=blue, fill=white, line width=0.7mm, label=right:{}] at (330:1.5cm) (v2) {};
            \node[circle, draw=green, fill=white, line width=0.7mm, label=right:{$v_3,\textcolor{green}{c_2}$}] at (0:1.5cm) (v3) {};
            \node[circle, draw=teal, fill=white, line width=0.7mm, label=above:{$v_4,\textcolor{teal}{c_4}$}] at (30:1.5cm) (v4) {};
            \node[circle, draw=violet, fill=white, line width=0.7mm, label=left:{$v_5,\textcolor{violet}{c_5=c_3}$}] at (65:1.5cm) (v5) {};

            \node[circle, draw=black, fill=white, line width=0.7mm, label=above:{$v_6$}] at (15:3cm) (v6) {};
            \node[circle, draw=black, fill=white, line width=0.7mm, label=above:{$v_7$}] at (5:4.5cm) (v7) {};
            \node[circle, draw=black, fill=white, line width=0.7mm, label=below:{$v_8$}] at (350:4.5cm) (v8) {};
            \node[circle, draw=violet, fill=white, line width=0.7mm, label=below:{$v_9,\textcolor{violet}{c_3}$}] at (340:3cm) (v9) {};

            \draw[blue, line width=0.7mm] (u) -- (v1);
            \draw[green, line width=0.7mm] (u) -- (v2);
            \draw[violet, line width=0.7mm] (u) -- (v3);
            \draw[red, line width=0.7mm] (u) -- node[above] {$\textcolor{red}{\alpha}$} (v4);
            \draw[teal, line width=0.7mm] (u) -- (v5);
            \draw[violet, line width=0.7mm] (v4) --node[above] {$\textcolor{violet}{c_3}$} (v6);
            \draw[red, line width=0.7mm] (v6) -- node[above] {$\textcolor{red}{\alpha}$} (v7);
            \draw[violet, line width=0.7mm] (v7) --node[left] {$\textcolor{violet}{c_3}$} (v8);
            \draw[red, line width=0.7mm] (v8) -- node[above] {$\textcolor{red}{\alpha}$} (v9);

        \end{tikzpicture}
        }
        \caption{After rotating the fan and coloring the edge.}
        \label{fig:rotate-and-color}
    \end{subfigure}

    \caption{In the figure, $u$ is the center of the fan with missing color $\alpha$, the vertices are labeled $v_i$, and $c_i$ denotes colors. Each colored edge is assigned either the color $\alpha$ or one of the colors $c_i$, while the dashed edge is uncolored. Figure~\ref{fig:nontrivial-fan} shows a non-trivial fan, since $\alpha \neq c_i$ for $1 \leq i \leq 5$ and $\miss(v_3)=\miss(v_5)=c_3$. For this fan, $(\alpha,v_3)$ is its \emph{required pair}, and the path $u-v_4-v_6-v_7-v_8-v_9$ is the maximal $\{\alpha,c_3\}$-alternating path starting at $u$ and ending at $v_9$, where $\alpha$ is missing. Figure~\ref{fig:path-flip} shows the coloring obtained after flipping this alternating path, after which $c_3$ is missing at $u$ and is the same as the missing color of at least one leaf (here both $v_3, v_5$). Finally, Figure~\ref{fig:rotate-and-color} shows the (prefix) fan rotation followed by coloring the previously uncolored edge $(u,v_1)$, where $c'$ and $c''$ denote the other updated missing colors.}
    \label{fig:nontrivial-fan-path}
\end{figure*}

\begin{definition}[Required pair]\label{def:good-pair}
Let $F=(u,\alpha),(v_1,c_1),\dots,(v_k,c_k)$ be an $\alpha$-primed Vizing fan with center $u$. The required pair of $F$ is the ordered pair of colors $\pi_F$ defined as follows.
\begin{itemize}
    \item If $F$ is nontrivial, i.e., $c_k\notin\miss_\chi(u)$, then by Definition~\ref{def:vizing-fan}, $c_k\in\{c_1,\ldots,c_{k-1}\}$. We call $\beta:=c_k$ the \defn{terminal color} of $F$ and define the ordered \defn{required pair} of $F$ to be $\pi_F:=(\alpha,\beta)$. Since $\alpha\in\miss_\chi(u)$ whereas $\beta\notin\miss_\chi(u)$, we have $\alpha\neq\beta$. The first color $\alpha$ is the designated color missing at the center, while $\beta$ is the second color used in the alternating path: flipping the maximal $\{\alpha,\beta\}$-alternating path starting at $u$ and then rotating the (prefix) fan colors the uncolored edge of $F$. For instance, the required pair of the nontrivial fan in Figure~\ref{fig:nontrivial-fan} is $(\alpha,c_3)$.
    \item If $F$ is trivial, i.e., $c_k\in\miss_\chi(u)$, no alternating path is needed, and we simply set $\pi_F:=(\alpha,\alpha)$.
\end{itemize}
For a uncolored edge $e$ with associated fan $F_e$, let $\pi_e:=\pi_{F_e}$.
\end{definition}

\paragraph{Parallel model.}
Because we use or augment some subroutines from Elkin and Khuzman~\cite{elkin2026efficient}, we adopt the same parallel model. That is, our parallel algorithms are stated for the ARBITRARY CRCW PRAM model. Processors operate synchronously on a shared memory; concurrent reads are allowed, and if multiple processors concurrently write to the same memory location, one of the written values is chosen arbitrarily. Each memory word contains $\Omega(\log n)$ bits, so vertex, edge, color, component, and record identifiers fit in $O(1)$ words. We measure the complexity of our algorithm using work and span: the \defn{work} is the total number of operations performed across all processors, and the \defn{span} (also called depth in much of the PRAM literature) is the longest chain of sequentially dependent operations, i.e., the parallel time with an unbounded number of processors.~\footnote{Some of the results referenced herein, bound the number of processors and parallel time, but we restate those results with respect to work, which is the product of both, and the span, which matches the parallel time.}

Throughout the paper, we use several standard deterministic parallel subroutines: comparison-based sorting, prefix sums, reductions, and compaction. (See~\cite{jaja1992parallel} for classical PRAM algorithms for many of these problems.) As sorting is the most expensive of these, and this paper does not focus on shaving the last logarithmic factor from the performance, we charge every such subroutine on $N$ items the sorting cost of $O(N\log N)$ work and $O(\log N)$ span, achievable deterministically by Cole's merge sort~\cite{cole1988parallel}.  For the randomized algorithm, we additionally assume that each processor can generate an independent uniformly random $O(\log n)$-bit word in $O(1)$ work and span.  We also make use of a random permutation, which can trivially be achieved by sorting random numbers. 

We additionally use a data representation similar to that of Elkin and Khuzman~\cite{elkin2026efficient}, deferring its discussion to the appendix. Additionally, the next three Lemmas~\ref{lem:parallel-paths}, \ref{lem:parallel-fans}, and \ref{lem:parallel-fan-rotations} roughly follow from their results and routines; likewise, we defer their proofs to the appendix as well. 

\begin{lemma}\label{lem:parallel-paths}
Let $H$ be a graph with $n$ vertices and maximum degree at most two, so that every connected component of $H$ is a path, a cycle, or an isolated vertex. In $O(n\log n)$ work and $O(\log n)$ span, one can identify and label all connected components of $H$ and the two endpoints of every path component containing at least one edge.
\end{lemma}

\begin{lemma}\label{lem:parallel-fans}
Given the designated center and designated missing color of every uncolored edge, maximal fans for all centers can be constructed, each fan can be classified as trivial or nontrivial, and the required pair of every nontrivial fan can be determined in $O(m\log\Delta)$ work and $O(\log\Delta)$ span.
\end{lemma}

As mentioned before, in the nontrivial case, coloring an uncolored edge requires two steps: first flipping the fan's required alternating path and then rotating the fan. After flipping the alternating path of a nontrivial fan, either the fan itself or an appropriate prefix of it becomes trivial. Thus, after all required path flips have been performed (whiteout any conflict), each originally nontrivial fan can be replaced by the resulting trivial fan or trivial prefix, while an originally trivial fan requires no path flip. The following lemma shows that, assuming these operations do not interfere with one another, a collection of pairwise vertex-disjoint trivial fans can be rotated in parallel to color their corresponding uncolored edges.

\begin{lemma}\label{lem:parallel-fan-rotations} Let $\mathcal{F}$ be a collection of pairwise vertex-disjoint trivial fans, each given as an indexed sequence of fan edges. Then all fans in $\mathcal{F}$ can be rotated and the uncolored edge corresponding to every fan is colored in $O(m)$ work and $O(1)$ span. \end{lemma}

Our deterministic algorithm, in each round, uses the following deterministic algorithm in order to find a large independent set of a graph with bounded average degree. 

\begin{lemma}[Large independent set~\cite{goldberg1993efficient}, restated]\label{lem:large-independent-set}
Let $G$ be a graph with $n$ vertices and $m$ edges. There is a deterministic parallel algorithm that finds an independent set of size at least $n^2/(2m+n)$ in $O\left((n+m)\alpha(m,n)\log n\right)$  work and $O(\log^3 n)$ span on CRCW PRAM, where $\alpha$ is the inverse of Ackermann's function.
\end{lemma}

To analyze our randomization algorithm, we use the following version of Azuma's inequality.

\begin{lemma}[Concentration Inequality~\cite{kuszmaul2021multiplicative}]\label{lem:conditional-chernoff}
Let $X_1,\ldots,X_n\in[0,c]$ be real-valued random variables with $c>0$. Suppose that $\mathbb{E}[X_i\mid X_1,\ldots,X_{i-1}]\le a_i$ for all $i$. Let $\mu=\sum_{i=1}^n a_i$. Then, for any $\delta>0$,
\[
 \Pr\!\left[ \sum_i X_i\ge(1+\delta)\mu \right]
 \le \exp\!\left( -\frac{\delta^2\mu}{(2+\delta)c} \right).
\]
\end{lemma}

\section{Building Blocks and  Deterministic Algorithm}\label{sec:BB-deter-alg}

This section describes the structural and parallel lemmas that form the backbone of our algorithms. We first show that bichromatic components corresponding to multiple color-disjoint pairs can be identified, processed, and flipped in parallel, rather than handling one color pair at a time. We then analyze the interactions that arise when many \defn{u-edges} (uncolored edges with fixed center and designated missing colors), whose required pairs come from multiple color-disjoint pairs, are colored simultaneously. We define the three relevant types of conflicts (fan-fan, path-path, and path-fan conflicts) and show that color-disjointness strongly controls the latter two, and that the total conflict graph remains sufficiently sparse. We also show how all of these conflicts can be detected and the conflict graph can be constructed efficiently in parallel. These ingredients already yield our first improvement: a deterministic algorithm that partitions all color pairs into $O(\Delta)$ color-disjoint sets, selects the largest set in each round, extracts a large conflict-free subset of the corresponding u-edges, and color them in parallel.

\subsection{Processing Color-disjoint Pairs in Parallel}
\label{subsec:parallel-pairs}

We begin by showing how multiple color-disjoint pairs can be processed in parallel; we postpone the conflicts that arise between the corresponding operations to Section~\ref{subsec:conflicts}. Recall that two color pairs are color-disjoint if the four colors involved are distinct. Let $\mathcal{P}$ be any set of pairwise color-disjoint pairs. The next two lemmas formalize how bichromatic components associated with all pairs in $\mathcal{P}$ can be identified and flipped simultaneously. To process these components efficiently, we construct an auxiliary graph and apply Lemma~\ref{lem:parallel-paths} to it.

Let $G'=(V',E')$ be an auxiliary graph consisting of one layer for each pair $p\in\mathcal{P}$. For every pair $p$ and every vertex $v$ incident to an edge of $H_\chi(p)$, we introduce a vertex $v_p$ in the layer corresponding to $p$. For every edge $(u,v)\in E_\chi(p)$, we introduce the corresponding edge $(u_p,v_p)$ in that layer. Thus, the layer corresponding to $p$ is a copy of $H_\chi(p)$ with its isolated vertices omitted. Since the pairs in $\mathcal{P}$ are color-disjoint, every colored edge of $G$ belongs to $H_\chi(p)$ for at most one $p\in\mathcal{P}$ and therefore introduces at most one edge in $G'$. Hence $|E(G')|\le m$. Moreover, a vertex $v_p$ is created only when $v$ is incident to an edge of $H_\chi(p)$, so $G'$ has no isolated vertices and consequently $ |V(G')|\le 2|E(G')|\le 2m$.

Thus, introducing one layer per color pair does not introduce a factor depending on $|\mathcal{P}|$ or $\Delta$ in the size of the auxiliary graph and $G'$ has $O(m)$ vertices and edges. The layers corresponding to distinct pairs are vertex-disjoint, and in the layer of $p$, the degree of $v_p$ is exactly the degree of $v$ in $H_\chi(p)$ and is therefore at most two. Hence $G'$ itself has maximum degree at most two, and every connected component of $G'$ corresponds to a bichromatic component of some $H_\chi(p)$. We can therefore apply Lemma~\ref{lem:parallel-paths} once to $G'$ to identify all bichromatic components simultaneously in $O(m\log n)$ work and $O(\log n)$ span. A formal construction of $G'$ is given in the proof of Lemma~\ref{lem:parallel-pairs}.

\begin{lemma}\label{lem:parallel-pairs}
Let $\mathcal{P}$ be a set of pairwise color-disjoint pairs of distinct colors. For every $p\in\mathcal{P}$, consider the components of $H_\chi(p)$ that contain at least one edge. All such components, over all pairs $p\in\mathcal{P}$, can be identified and labeled in parallel in $O(m\log n)$ work and $O(\log n)$ span.
\end{lemma}

\begin{proof}
For a pair $p=\{a,b\}\in\mathcal{P}$, let $E_\chi(p):=\{e\in E:\chi(e)\in\{a,b\}\}$ be the edge set of $H_\chi(p)$. Since $\chi$ is proper, every vertex has at most one incident edge of color $a$ and at most one incident edge of color $b$, and therefore $H_\chi(p)$ has maximum degree at most two. Every component containing at least one edge is a path or an even cycle whose edge colors alternate between $a$ and $b$. Moreover, since the pairs in $\mathcal{P}$ are pairwise color-disjoint, every colored edge belongs to $E_\chi(p)$ for at most one $p\in\mathcal{P}$. Thus
\[
    \sum_{p\in\mathcal{P}}|E_\chi(p)|\le m.
\]

We build an auxiliary graph $G'=(V',E')$ whose vertices are indexed by color pairs. For every $p\in\mathcal{P}$ and every vertex $v$ incident to an edge of $E_\chi(p)$, let $v_p$ denote the corresponding vertex of $G'$, and define
\[
\begin{aligned}
    V'&:=\bigl\{v_p:p\in\mathcal{P},\ v\text{ is incident to an edge of }E_\chi(p)\bigr\},\\
    E'&:=\bigl\{(u_p,v_p):p\in\mathcal{P},\ (u,v)\in E_\chi(p)\bigr\}.
\end{aligned}
\]
The vertices corresponding to different pairs are distinct, even when they arise from the same vertex of $G$. Thus the layers of $G'$ corresponding to different pairs are vertex-disjoint, and each component of $G'$ corresponds exactly to a component containing at least one edge of some $H_\chi(p)$. Furthermore, $\deg_{G'}(v_p)=\deg_{H_\chi(p)}(v)\le2$. Since $|E'|\le m$ and every vertex of $G'$ is incident to an edge, $|V'|\le2|E'|\le2m$.

We now show how to construct $G'$ in parallel. First, for each color, store the unique pair of $\mathcal{P}$ containing it, if such a pair exists. We then scan the edges of $G$ in parallel. An edge $g=(u,v)$ whose color belongs to a pair $p$ produces an auxiliary edge $g_p=(u_p,v_p)$ together with the two endpoint records $u_p$ and $v_p$, while an edge whose color belongs to no pair is ignored. Each endpoint record $v_p$ is identified by the pair consisting of the identifier of $v$ and the identifier of $p$. We sort all endpoint records by these identifiers. Thus all occurrences corresponding to the same auxiliary vertex $v_p$ become consecutive, while $v_p$ and $v_{p'}$ remain distinct whenever $p\neq p'$. Each maximal block of equal records is assigned a single vertex of $G'$, and this auxiliary vertex stores the identifier of its corresponding original vertex $v$. The assigned auxiliary-vertex identifiers are then propagated back to the endpoint records that generated them, so that every auxiliary edge $g_p$ obtains its two endpoints. At the same time, for every original edge $g$ that produces an auxiliary edge $g_p$, we store a pointer from $g$ to $g_p$.

For instance, two edges $(u_p,v_p)$ and $(v_p,w_p)$ initially produce the endpoint records $u_p,v_p,v_p,w_p$. Two occurrences of $v_p$ belong to the same block and hence correspond to the same vertex of $G'$, so the resulting edges form the path $(u_p-v_p-w_p)$. Since every edge of $G$ produces records for at most one pair, the construction costs $O(m\log n)$ work and $O(\log n)$ span. Applying Lemma~\ref{lem:parallel-paths} once to $G'$ then assigns component identifiers. Since $G'$ has $O(m)$ vertices and edges, this also costs $O(m\log n)$ work and $O(\log n)$ span. Thus all components and the endpoints of every path component are identified and labeled, and every original edge represented in $G'$ has a pointer to its corresponding auxiliary edge.
\end{proof}

It remains to show that the selected components can also be flipped simultaneously. In the two-color setting, once the required components are known, Elkin and Khuzman~\cite{elkin2026efficient} flip them by exchanging the two colors on their edges in parallel. Here $G'$ represents several color pairs at once, but every component belongs to a unique layer $p=\{a,b\}$, so flipping that component still simply exchanges $a$ and $b$ on its corresponding edges in $G$. Since the pairs in $\mathcal{P}$ are color-disjoint, every original edge is affected by at most one such flip. Thus all selected components can still be flipped with only linear work and constant span once their component labels are known; Lemma~\ref{lem:parallel-pair-flips} formalizes this statement.

\begin{lemma}\label{lem:parallel-pair-flips}
Let $\mathcal{P}$ be a set of pairwise color-disjoint pairs of distinct colors, and suppose that the components of $H_\chi(p)$ containing at least one edge have been identified for every $p\in\mathcal{P}$. Any collection of distinct such components can be flipped simultaneously in $O(m)$ work and $O(1)$ span, including the corresponding updates to the data structure. The resulting coloring is a proper partial coloring; its set of uncolored edges is unchanged, and the outcome is independent of the order in which the flips are applied.
\end{lemma}

\begin{proof}
Let $\mathcal{K}$ be the selected collection of components. For a component $K\in\mathcal{K}$ belonging to $H_\chi(p)$, where $p=\{a,b\}$, flipping $K$ exchanges colors $a$ and $b$ on all of its edges. Since the pairs in $\mathcal{P}$ are color-disjoint, every colored edge belongs to $H_\chi(p)$ for at most one $p\in\mathcal{P}$, and hence every edge is affected by at most one selected flip. Therefore, after marking the selected component identifiers, all edges can be scanned in parallel, and every edge belonging to a selected component can have its color exchanged independently.

We update the data structure within the same bounds. The color of edges and $\operatorname{Edge2Color}$ entries of every flipped edge are updated in parallel. For $\operatorname{Color2Edge}$, we first delete in parallel all entries corresponding to the old colors of flipped edges and then insert in parallel the entries corresponding to their new colors. Separating deletion from insertion avoids a concurrent-write conflict at an internal vertex, where one incident edge changes from $a$ to $b$ while the other changes from $b$ to $a$.  So it remains to update the stored missing colors of the vertices. At an internal vertex of a flipped path (or at a vertex of a flipped cycle), both colors of the corresponding pair are present before and after the flip, so no stored missing color needs to change. At an endpoint of a flipped path with pair $p=\{a,b\}$, exactly one of $a$ and $b$ is missing before the flip, and the other is missing afterwards. Therefore, for every stored missing color entry at that endpoint whose value belongs to $p$, we replace its value with the other color of $p$. Any stored missing color not belonging to $p$ remains missing and is left unchanged. Thus all stored missing color entries can be updated in $O(1)$ work per endpoint. Since the pairs in $\mathcal{P}$ are color-disjoint, each stored entry is affected by at most one flipped pair, so all such updates can be performed in parallel. Hence the entire data structure is updated in $O(m)$ work and $O(1)$ span.

Next, we show that the resulting coloring is proper. Fix a pair $p=\{a,b\}\in\mathcal{P}$ and a vertex $v$. All edges of $H_\chi(p)$ incident to $v$ belong to the same component, so either all of them are flipped, or none of them is. Before the flip, at most one incident edge has color $a$ and at most one has color $b$, and exchanging $a$ and $b$ on the entire component preserves this property. More explicitly, if $v$ is an internal vertex of a path or a vertex of a cycle, then it is incident to one edge of color $a$ and one of color $b$, and the flip simply exchanges these two colors. If $v$ is an endpoint of a path, then it is incident to exactly one edge of the component, which changes from $a$ to $b$ or from $b$ to $a$; thus the flip exchanges which of $a$ and $b$ is missing at $v$. Since different pairs in $\mathcal{P}$ use disjoint colors, the flips associated with different pairs cannot create a conflict in any color. Colors belonging to no pair are never changed. Hence every color still appears on at most one edge incident to each vertex, and the resulting coloring is proper. Moreover, an uncolored edge belongs to no $E_\chi(p)$ and is therefore never modified, while every flip only exchanges two colors on already colored edges. Hence the set of uncolored edges is unchanged.

Let us show that the order of the flips does not matter. Distinct components in the same layer have disjoint underlying edge sets. Components belonging to different layers may share vertices in $G$, but their underlying edge sets are also disjoint since the pairs in $\mathcal{P}$ are color-disjoint. Therefore, every edge of $G$ is affected by at most one selected flip. Let $K$ and $K'$ be two selected components. If an edge belongs to neither component, both flips leave it unchanged; if it belongs to exactly one, only the corresponding flip changes its color; and it cannot belong to both. Hence the two flips commute. Since every pair of selected flips commutes, applying all selected flips in any order produces the same final coloring.

It remains to bound the cost of performing the flips. Once the selected component identifiers are known, all edges can be scanned in parallel, and every edge belonging to a selected component can exchange its two colors independently. All data structure updates also take $O(1)$ work each and can be performed in parallel. Since at most $m$ edges and $O(m)$ endpoints are involved, all selected components can be flipped, including all required data-structure updates, in $O(m)$ work and $O(1)$ span.
\end{proof}

\subsection{Conflicts}\label{subsec:conflicts}

Throughout this subsection, fix one round of the algorithm and let $A$ be the set of active u-edges considered in this round. Let $F_e$ be the fan of a u-edge $e\in A$. If $F_e$ is nontrivial, let $\pi_e=(\alpha,\beta)$ be its required pair, define its underlying unordered color pair by $p_e:=\{\alpha,\beta\}$, and let $P_e:=K_\chi(p_e,u)$, where $u=\cen(e)$. Since $\alpha\in\miss_\chi(u)$ and $\beta\notin\miss_\chi(u)$, the vertex $u$ has degree exactly one in $H_\chi(p_e)$, and hence $P_e$ is the maximal $\{\alpha,\beta\}$-alternating path starting at $u$. If $F_e$ is trivial, we set $P_e:=\varnothing$. In a single round, the underlying color pairs of all active nontrivial fans belong to one color-disjoint set $\mathcal{P}$, where every color belongs to at most one pair. Consequently, for any two active nontrivial u-edges $e$ and $f$, either $p_e=p_f$ or $p_e$ and $p_f$ are color-disjoint.

In order to color two active u-edges simultaneously, their operations should not interfere with one another; in particular, executing the two operations in either order should give the same result. We define conflicts so that, after extracting a conflict-free set of u-edges, no two selected u-edges interact, and hence all of them can be colored in parallel without interference. In general, two active u-edges can interact in three ways: their fans overlap, their paths may intersect, or the alternating path of one modifies the fan of the other. We use a pessimistic notion of conflict: whenever two active operations do not conflict in any of these three senses, they may be executed simultaneously without affecting one another. Two distinct active fans $F_e,F_f$ have a \defn{fan-fan conflict} if $V(F_e)\cap V(F_f)\neq\varnothing$, i.e., they share at least a vertex. Since each fan contains at most $\Delta+1$ vertices, and for each such vertex $x$ there are at most $\Delta+1$ active fans containing $x$, the next observation follows.
\begin{observation}\label{obv:fan-fan-conflicts}
   Every active fan has fan-fan conflicts with $O(\Delta^2)$ other active fans.
\end{observation}

Consider two nonempty active paths $P_e$ and $P_f$. There are three possibilities. If $\pi_e$ and $\pi_f$ have the same underlying color pair ${\alpha,\beta}$, then both paths lie in the two-colored subgraph $H_\chi({\alpha,\beta})$; since distinct components of $H_\chi({\alpha,\beta})$ are vertex-disjoint, the two paths can interact only if they are in fact one path. On the other hand, if $p_e$ and $p_f$ are color-disjoint, then $P_e$ and $P_f$ use disjoint sets of colors and therefore have disjoint edge sets. They may share vertices, but flipping them simultaneously is harmless: the two flips exchange disjoint pairs of colors and hence do not interfere with one another. The remaining possibility is that the two pairs share exactly one color, say ${\alpha,\beta}$ and ${\alpha,\gamma}$, which never arises in our algorithms since our active pairs belong to the color-disjoint set $\mathcal{P}$. Thus, two paths with color-disjoint pairs do not conflict. 

Suppose that the underlying color pairs of the active nontrivial u-edges belong to a color-disjoint set $\mathcal{P}$. Then two distinct active u-edges $e,f$ with nonempty paths have a \defn{path-path conflict} if $p_e=p_f$ and $P_e=P_f$, that is, if they use the same alternating path. If $p_e$ and $p_f$ are color-disjoint, their paths do not interact. Therefore, each u-edge can share a path with at most one other u-edge. Note that this type of conflict is in fact harmless, since flipping the shared path $P_e=P_f$ frees the colors needed for both $e$ and $f$. Nevertheless, we pessimistically consider it a conflict so that no two u-edges in a conflict-free set interact with one another, even in a harmless way. Additionally, this type of interaction can be viewed as a conflict arising from the path of one u-edge overlapping with the fan of the other u-edge. We therefore count it as part of the next type of conflict (path-fan conflicts). Note that if $\mathcal{P}$ were not color-disjoint, separate path-path conflicts could in fact arise, since two alternating paths whose color pairs share a color may overlap and their flips may interfere.

For a fixed u-edge, its fan can be modified in two ways: either by overlapping with the fan of another u-edge or via the path of another u-edge. The former is accounted for by the fan-fan conflicts considered earlier, so it remains to consider the latter. A path can modify the fan $ F_e=(u,\alpha),(v_1,c_1),\dots,(v_k,c_k) $ of a u-edge $e$ by changing one of the colors ${\alpha,c_1,\ldots,c_k}$ associated with the fan. In particular, the path $P_f$ of one active u-edge may modify the fan $F_e$ of another when it is flipped in one of two ways. First, $P_f$ may end at a vertex $x\in V(F_e)$ and change whether a color that $F_e$ requires to be missing at $x$ remains missing. Second, $P_f$ may contain a colored fan edge of $F_e$ and recolor it. Note that both types of modification may occur simultaneously, but this does not affect the overall bound on the number of path-fan conflicts. For every $x\in V(F_e)$, recall that the designated missing color $\varphi(x)$ is the color that is required to be missing at $x$ for the fan and $\varphi(u)=\alpha$ is the designated missing color at the center and $\varphi(v_i)=c_i$ at each leaf $v_i$. Therefore, since $\varphi(x)\in\miss_\chi(x)$ for every $x\in V(F_e)$, the validity of the fan relies on this color remaining missing at $x$. A single long alternating path can interact with the fans of many other u-edges, so generally the path-fan conflict relation can have maximum degree as large as $\Omega(|A|)$, but we will show that the average degree is bounded. More formally, let $e,f$ be distinct active u-edges with $P_f\neq\varnothing$. We say $P_f$ has a \defn{path-fan conflict} with $F_e$, from $f$ to $e$, if flipping $P_f$ can modify $F_e$ in one of the following two ways: \defn{endpoint conflict}, $P_f$ has an endpoint at a vertex $x\in V(F_e)$ with $\varphi(x)$ one of the two colors of $\pi_f$; or \defn{fan-edge conflict}, $P_f$ contains at least one of the colored fan edges $(u,v_2),\dots,(u,v_k)$ of $F_e$. 

\begin{lemma}\label{lem:path-fan-conflicts}
Suppose that the underlying color pairs of the active nontrivial u-edges belong to a color-disjoint set $\mathcal{P}$. Then the fan of an active u-edge can conflict with paths of $O(\Delta)$ other u-edges.
\end{lemma} 

\begin{proof}
Fix an active fan $F_e$; we bound the number of active paths whose path-fan conflicts are oriented toward $e$, considering the two cases of path-fan separately.

\emph{Endpoint conflicts.} Fix a vertex $x\in V(F_e)$ and let $\gamma=\varphi(x)$ be the designated missing color of $x$ required by $F_e$. If an active path $P_f$ creates an endpoint conflict at $x$, then $x$ is an endpoint of $P_f$ and $\gamma\in p_f$. Since the underlying pairs of the active nontrivial u-edges belong to the color-disjoint set $\mathcal{P}$, there is at most one pair $p\in\mathcal{P}$ containing $\gamma$. For this fixed pair $p$, there is a unique component $K$ of $H_\chi(p)$ containing $x$. Any active path creating an endpoint conflict at $x$ must equal $K$. Since $x$ is an endpoint of that component, $K$ is a path component, and a path $K$ in general can be used by at most two active u-edges. Therefore, each vertex $x\in V(F_e)$ receives $O(1)$ endpoint conflicts. Since $|V(F_e)|\le\Delta+1$, the total number of endpoint conflicts oriented toward $e$ is $O(\Delta)$. Note that this case also includes the shared-path case from the path-path conflict: if $P_e=P_f$, then taking $x=\cen(e)$ gives an endpoint conflict from $f$ to $e$, since $\varphi(x)$ is the first color of $\pi_e$ and hence belongs to $p_e=p_f$.

\emph{Fan-edge conflicts.} Fix a colored fan edge $g\in\{(u,v_2),\dots,(u,v_k)\}$ of $F_e$ and let $\gamma=\chi(g)$. If an active path $P_f$ contains $g$, then $\gamma\in p_f$. Again, color-disjointness of $\mathcal{P}$ implies that there is at most one pair $p\in\mathcal{P}$ containing $\gamma$. For this pair, $g$ lies in a unique component $K$ of $H_\chi(p)$, and every active path containing $g$ must equal $K$. If such an active path exists, then $K$ is used by at most two active u-edges. Thus each colored fan edge receives $O(1)$ fan-edge conflicts. Since $F_e$ has at most $\Delta$ colored fan edges, the total number of fan-edge conflicts oriented toward $e$ is $O(\Delta)$.

Combining the two cases, the path-fan conflict in-degree of $e$ is $O(\Delta)$. The color-disjointness of $\mathcal{P}$ is essential here: for each relevant color $\gamma$, it ensures that at most one active color pair contains $\gamma$; without this property, up to $O(\Delta)$ different pairs containing $\gamma$ could contribute conflicts at a single fan vertex or fan edge, increasing the in-degree up to $O(\Delta^2)$ instead of $O(\Delta)$.
\end{proof}

\subsection{Conflict Graph}\label{subsec:conflict-graph}
So far we have established structural bounds on the number of fan-fan, path-path, and path-fan conflicts. To use these bounds algorithmically, we must also detect the conflicts efficiently and construct a conflict graph, in particular for the deterministic algorithm, from which a conflict-free set of u-edges can be extracted. Rather than building a separate graph for each conflict type, we combine all three into a single conflict graph. For a set of active u-edges $A$, let $G_{\mathrm{conf}}(A)$ be the \defn{conflict graph} with vertex set $A$, where two distinct u-edges $e,f\in A$ are adjacent whenever they have a conflict. We call a subset $I\subseteq A$ \defn{conflict-free} if it is an independent set of $G_{\mathrm{conf}}(A)$. By the definitions of the three conflict types, all u-edges in a conflict-free set can be colored simultaneously without interfering with one another. 

\begin{lemma}\label{lem:conflict-graph-edges}
The conflict graph $G_{\mathrm{conf}}(A)$ has $O(\Delta^2|A|)$ edges.
\end{lemma}

\begin{proof}
We bound the number of edges added by each conflict type. By Observation~\ref{obv:fan-fan-conflicts}, every active fan has fan-fan conflicts with $O(\Delta^2)$ other active fans. Summing over the $|A|$ active fans gives $O(\Delta^2|A|)$ fan-fan edges. By Lemma~\ref{lem:path-fan-conflicts}, every active u-edge has path-fan conflict in-degree $O(\Delta)$. Summing the in-degrees over all $|A|$ vertices gives $O(\Delta|A|)$ oriented path-fan conflicts, and hence at most $O(\Delta|A|)$ undirected path-fan edges in $G_{\mathrm{conf}}(A)$. Taking the union of all conflict types, we conclude that $G_{\mathrm{conf}}(A)$ has $O(\Delta^2|A|)$ edges.
\end{proof}
%%%
%%%
%%%
%%%
%%%.
We next show that the structural conflicts defined above can also be detected efficiently in parallel. The high-level idea is to represent potential conflicts by records and use parallel sorting and grouping to bring together records that witness the same interaction. We use the component labels of the auxiliary graph $G'$ to associate each alternating component with the active u-edges whose paths use it; we call this the \defn{ownership relation}. Fan-fan conflicts are detected by grouping fans that contain the same vertex. For path-fan conflicts, endpoint conflicts are detected by matching the missing colors required by fan vertices with the colors whose missing status is changed at path endpoints, while fan-edge conflicts are detected using the labeled components of $G'$ together with the ownership relation. These procedures may generate the same conflicting pair $\{e,f\}$ more than once; we call each generated pair a \defn{conflict record} and then remove duplicates at the end by sorting.

\begin{lemma}\label{lem:build-conf-graph}
Let $A$ be a set of uncolored edges such that the underlying required pairs of its nontrivial edges belong to a color-disjoint set. Then all conflicts among the edges of $A$ can be detected and the conflict graph $G_{\mathrm{conf}}(A)$ can be constructed in $O(m\log n+\Delta^2|A|\log n)$ work and $O(\log n)$ span.
\end{lemma}

\begin{proof}
By Lemma~\ref{lem:parallel-fans}, the fans of the edges in $A$ can be constructed and classified, and the required pair of every nontrivial fan can be determined in $O(m\log\Delta)$ work and $O(\log\Delta)$ span. Since the underlying required pairs of the nontrivial edges in $A$ belong to a color-disjoint set, Lemma~\ref{lem:parallel-pairs} constructs the corresponding pair-indexed graph $G'$ and identifies and labels all required bichromatic paths in $O(m\log n)$ work and $O(\log n)$ span. By the construction in Lemma~\ref{lem:parallel-pairs}, every original edge represented in $G'$ has a pointer to its corresponding auxiliary edge.

It remains to detect the conflicts and construct $G_{\mathrm{conf}}(A)$. We explain how to detect each type of conflict separately, generating a conflict record $\{e,f\}$ whenever a conflict between two edges $e,f\in A$ is found. The same conflict may be generated more than once; these duplicates are removed at the end. We first identify, for each alternating path, the active u-edges that use it; this will be used to find path-fan conflicts.

During the construction of $G'$, for every original edge $g\in E_\chi(p)$ represented by an auxiliary edge $g_p$ in the layer of $p$, we store a pointer from $g$ to $g_p$. After the components of $G'$ are labeled, this allows us to determine in $O(1)$ work the component containing $g$ by reading the component label of an endpoint of $g_p$. We also store the two endpoints of every path component identified by Lemma~\ref{lem:parallel-paths}.

\smallskip
\emph{Ownership relation:}
For every $e\in A$ with nontrivial fan, let $u=\cen(e)$, $\pi_e=(\alpha,\beta)$, and $p=\{\alpha,\beta\}$. Since $\alpha\in\miss_\chi(u)$ and $\beta\notin\miss_\chi(u)$, the vertex $u$ has no incident $\alpha$-colored edge and exactly one incident $\beta$-colored edge. Hence $\deg_{H_\chi(p)}(u)=1$, so $u$ is an endpoint of the alternating component $P_e=K_\chi(p,u)$. Using $\operatorname{Color2Edge}(u,\beta)$, we obtain the unique $\beta$-colored edge $g$ incident to $u$ in $O(1)$ work and $O(1)$ span. The pointer stored during the construction of $G'$ gives the corresponding auxiliary edge $g_p$. Since $g_p$ belongs to the component corresponding to $P_e$, the identifier of $P_e$ can be obtained in $O(1)$ work and $O(1)$ span by reading the component label of either endpoint of $g_p$.

For every such $e$, generate the record $(\operatorname{id}(P_e),e)$ and sort these, at most $|A|$ records, by component identifier. Every component appearing in this sorting is a path, and the center of each active u-edge associated with it is one of its two endpoints. Since the centers in $A$ are distinct, each endpoint is the center of at most one active u-edge, and hence every component is associated with at most two active u-edges. For each component identifier, we store these at most two u-edges in an ownership table. We call this association the component-to-u-edge ownership relation, or simply the \defn{ownership relation}, which can be constructed in $O(|A|\log n)$ work and $O(\log n)$ span.

\smallskip
\emph{Fan-fan conflicts:}
For every $e\in A$ (trivial or not) and every $x\in V(F_e)$, generate the record $(x,e)$ and sort all such records by $x$. The group with key $x$ consists exactly of the active fans containing $x$. Let its size be $s_x$. Any active fan containing $x$ is centered either at $x$ or at a neighbor of $x$, and the active centers are distinct, so $s_x\le \deg(x)+1\le \Delta+1$. We generate one conflict record $\{e,f\}$ for every pair of distinct active fans in the same group. Thus a group of size $s_x$ generates  $\binom{s_x}{2}$ fan-fan conflict records. Every fan-fan conflict is found in this way, since two fans have a fan-fan conflict exactly when they share a vertex. It remains to bound the total number of conflict records generated over all groups. Since $s_x\le\Delta+1$,
\[
\begin{aligned}
\sum_x \binom{s_x}{2}
&\le \frac{\Delta+1}{2}\sum_x s_x \\
&= \frac{\Delta+1}{2}\sum_{e\in A}|V(F_e)| \\
&\le \frac{(\Delta+1)^2}{2}|A| \\
&= O(\Delta^2|A|).
\end{aligned}
\]
Since we generate one record $(x,e)$ for each vertex $x\in V(F_e)$, the total number of such records is $\sum_{e\in A}|V(F_e)|\le(\Delta+1)|A|=O(\Delta|A|)$. Thus sorting them by $x$ requires $O(\Delta|A|\log n)$ work and $O(\log n)$ span. Together with the $O(\Delta^2|A|)$ work needed to generate the conflict records, this part costs $O(\Delta^2|A|\log n)$ work and $O(\log n)$ span.

\smallskip
\emph{Path-fan conflicts:}
We find endpoint conflicts and fan-edge conflicts separately. First consider endpoint conflicts. For every $e\in A$ and every $x\in V(F_e)$, the fan $F_e$ requires the designated color $\varphi(x)$ to be missing at $x$. We therefore generate the \defn{demand} record $(x,\varphi(x),e)$, representing that $F_e$ depends on $\varphi(x)$ remaining missing at $x$. There are $\sum_{e\in A}|V(F_e)|=O(\Delta|A|)$ demand records. Now let $f\in A$ be a nontrivial u-edge with $\pi_f=(\alpha,\beta)$. At either endpoint $y$ of $P_f$, exactly one of $\alpha$ and $\beta$ is missing before the flip, and flipping $P_f$ exchanges their colors; the flip changes the missing color at $y$. We therefore generate the two \defn{supply} records $(y,\alpha,f)$ and $(y,\beta,f)$ for each endpoint $y$ of $P_f$. Since $P_f$ has two endpoints, every nontrivial active u-edge generates four supply records, for a total of $O(|A|)$ supply records. Sort all demand and supply records lexicographically by their vertex and color fields. If a demand $(x,\gamma,e)$ and a supply $(x,\gamma,f)$ with $e\neq f$ occur in the same group, then $F_e$ requires $\gamma$ to be missing at $x$, while flipping $P_f$ changes whether $\gamma$ is missing at $x$. Hence $e$ and $f$ have an endpoint conflict, and we generate the conflict record $\{e,f\}$. Conversely, every endpoint conflict gives such a matching demand and supply record, so all endpoint conflicts are found.

To bound the number of records generated, fix a key $(x,\gamma)$. Any supply record with this key comes from an active path whose underlying color pair contains $\gamma$. Since the active color pairs belong to a color-disjoint set, at most one pair $p$ contains $\gamma$. For this pair, $x$ belongs to a unique component of $H_\chi(p)$, and by the ownership relation this component is associated with at most two active u-edges. Hence at most two supply records occur with key $(x,\gamma)$. Each demand record therefore generates at most two endpoint-conflict records. Since there are $O(\Delta|A|)$ demand records, the total number of endpoint-conflict records is $O(\Delta|A|)$. The sorting and grouping costs $O(\Delta|A|\log n)$ work and $O(\log n)$ span.

Now consider fan-edge conflicts. Let $g$ be a colored fan edge of $F_e$ and let $\gamma=\chi(g)$. An active alternating path can contain $g$ only if its underlying color pair contains $\gamma$. Since the active color pairs belong to a color-disjoint set, there is at most one such pair $p$. If no such pair exists, then $g$ creates no fan-edge conflict. Otherwise, $g\in E_\chi(p)$ and is represented by an auxiliary edge $g_p$ in the layer $p$ of $G'$. Using the pointer retained when $G'$ was constructed, we read the identifier of the component containing $g_p$ in $O(1)$ work. The ownership table then gives the at most two active u-edges whose alternating path is this component. For every such $f\neq e$, we generate the conflict record $\{e,f\}$. Hence every fan-edge conflict is found. Each fan has $O(\Delta)$ colored fan edges, so at most $O(\Delta|A|)$ such edges are examined, and each generates at most two conflict records. Hence the total number of fan-edge conflict records is $O(\Delta|A|)$. Once the component labels and ownership table are available, all required lookups and record generation take $O(\Delta|A|)$ work and $O(1)$ additional span.

\smallskip
\emph{Building the conflict graph:}
The preceding procedures generate $O(\Delta^2|A|)$ conflict records in total: $O(\Delta^2|A|)$ from fan-fan conflicts and $O(\Delta|A|)$ from path-fan conflicts (including the path-path conflicts). A conflicting pair may be generated several times, for example when two fans share more than one vertex. We therefore represent every conflict record $\{e,f\}$ canonically by $\bigl( \min\{\operatorname{id}(e),\operatorname{id}(f)\}, \max\{\operatorname{id}(e),\operatorname{id}(f)\} \bigr)$ and sort all conflict records by this pair, and keep one representative from each group of equal records. After removing duplicates, there is only one record for every pair of conflicting active u-edges, and these records form the edge set of $G_{\mathrm{conf}}(A)$. Moreover, the bounds obtained above for fan-fan and path-fan conflicts match the corresponding structural bounds of Observation~\ref{obv:fan-fan-conflicts} and Lemma~\ref{lem:path-fan-conflicts}.

Since there are $O(\Delta^2|A|)$ conflict records, sorting them costs $O(\Delta^2|A|\log n)$ work and $O(\log n)$ span. All preceding steps are bounded by the same work and span, and only a constant number of sorting, grouping, and lookup phases are performed. Therefore, all conflicts can be detected and the conflict graph $G_{\mathrm{conf}}(A)$ can be constructed in $O(m\log n+\Delta^2|A|\log n)$ work and $O(\log n)$ span.
\end{proof}

%%%
%%%
\subsection{The Deterministic Algorithm}
\label{subsec:deterministic}

In Section~\ref{subsec:parallel-pairs} we showed that components belonging to a color-disjoint set of pairs can be processed and flipped simultaneously. While previous parallel algorithms handle one color pair at a time out of $\Theta(\Delta^2)$ pairs of the palette, our deterministic algorithm takes advantage of the idea that conflict-free u-edges whose required pairs belong to a set of color-disjoint pairs can be colored in parallel. We partition the color pairs of the palette into $\ell=O(\Delta)$ sets of pairwise color-disjoint pairs, each of size $\Theta(\Delta)$. In each round, either at least half of the current fans are trivial, in which case we activate the trivial fans, or at least half are nontrivial. In the latter case, since the nontrivial required pairs are partitioned among $\ell=O(\Delta)$ sets, there exists a set that contains the required pairs of $\Omega(\lambda/\Delta)$ currently uncolored edges, where $\lambda$ is the number of uncolored edges. In each round of the algorithm, it first selects the set with the highest number of uncolored edges, then builds the conflict graph of the active u-edges, extracts a large conflict-free subset, and performs all of their Vizing operations in parallel. A conflict-free subset is guaranteed to contain an $\Omega(1/\Delta^2)$ fraction of the active set, so every round colors an $\Omega(1/\Delta^3)$ fraction of the remaining uncolored edges and therefore $O(\Delta^3\log n)$ rounds suffice to color all uncolored edges.

We first show how to extract such a conflict-free set; the sparsity bound for the conflict graph from Lemma~\ref{lem:conflict-graph-edges} is sufficient to guarantee a large independent set.

\begin{lemma}\label{lem:large-is}
Given the conflict graph $G_{\mathrm{conf}}(A)$, a conflict-free subset $I\subseteq A$ of size $\Omega(|A|/\Delta^2)$ can be computed deterministically in $O(\Delta^2|A|\,\alpha(n)\log n)$ work and $O(\log^3 n)$ span, where $\alpha$ denotes the inverse Ackermann function.
\end{lemma}

\begin{proof}
Let $N:=|A|$ and $M:=|E(G_{\mathrm{conf}}(A))|$. By Lemma~\ref{lem:conflict-graph-edges}, $M=O(\Delta^2N)$. Applying Lemma~\ref{lem:large-independent-set} to $G_{\mathrm{conf}}(A)$ yields an independent set of size at least
\[
\frac{N^2}{2M+N}
=\Omega\left(\frac{N}{\Delta^2}\right)
=\Omega\left(\frac{|A|}{\Delta^2}\right).
\]
Since two vertices are adjacent in $G_{\mathrm{conf}}(A)$ whenever the corresponding active u-edges conflict, this independent set is conflict-free. The span is $O(\log^3 N)=O(\log^3 n)$ and the work is
\[
\begin{aligned}
O\left((N+M) \alpha(M,N)\log N\right)
&=O\left( \Delta^2|A|\, \alpha(\Delta^2|A|,|A|) \log n\right)\\
&=O\left( \Delta^2|A|\, \alpha(n^3) \log n\right)\\
&=O\left( \Delta^2|A|\, \alpha(n) \log n\right).
\end{aligned}
\]
\end{proof}

\paragraph{Partitioning the palette into color-disjoint sets.}
Before starting the first round of the algorithm, we partition the unordered pairs of colors into color-disjoint sets. This can be viewed as an edge-coloring problem on the complete graph whose $\Delta+1$ vertices correspond to the colors in the palette and whose edges correspond to unordered pairs of distinct colors. A proper edge coloring of this complete graph partitions its edges into matchings, and hence partitions the unordered color pairs into color-disjoint sets.

\begin{lemma}\label{lem:complete-graph-edge-coloring}
The complete graph $K_{\Delta+1}$ admits a proper $\ell$-edge coloring with
\[
    \ell:=
    \begin{cases}
        \Delta, & \text{if $\Delta+1$ is even},\\
        \Delta+1, & \text{if $\Delta+1$ is odd}.
    \end{cases}
\]
If $\Delta+1$ is even, every color class is a perfect matching containing $(\Delta+1)/2$ edges. If $\Delta+1$ is odd, every color class is a matching containing $\Delta/2$ edges. Moreover, such an edge coloring can be obtained in $O(\Delta^2)$ work and $O(1)$ span.
\end{lemma}
\begin{proof}
Every complete graph with an even number of vertices admits an explicit 1-factorization in which the factor containing any given edge can be determined directly from its endpoints~\cite{harary1969}~(Theorem~9.1). Suppose first that $\Delta+1$ is even. Applying the construction in Theorem~9.1 of~\cite{harary1969} to $K_{\Delta+1}$ partitions its edges into $\Delta$ perfect matchings, each containing $(\Delta+1)/2$ edges. Taking these matchings as the color classes gives the desired proper edge coloring with $\ell=\Delta$. Suppose now that $\Delta+1$ is odd. Add one dummy vertex together with all of its incident edges, obtaining $K_{\Delta+2}$. Since $\Delta+2$ is even, applying the same 1-factorization gives $\Delta+1$ perfect matchings. Delete the dummy vertex and all edges incident to it. Each perfect matching loses exactly one edge, and hence each remaining matching contains $\Delta /2$ edges. These $\Delta+1$ matchings partition $E(K_{\Delta+1})$, giving the desired proper edge coloring with $\ell=\Delta+1$.

It remains to bound the work and span. By the explicit construction in Theorem~9.1 of~\cite{harary1969}, the factor containing each edge can be determined in $O(1)$ work and span. In the odd case, adding the dummy vertex and its incident edges, and subsequently discarding them, also requires $O(1)$ span and $O(\Delta)$ work. Thus all $O(\Delta^2)$ edges can be assigned to their color classes independently in parallel, costing $O(\Delta^2)$ total work and $O(1)$ span.
\end{proof}

\begin{lemma}\label{lem:color-pair-partition}
The unordered pairs of distinct colors in $[\Delta+1]$ can be represented as sets $\mathcal{P}_1,\ldots,\mathcal{P}_\ell$ such that every $\mathcal{P}_j$ consists of pairwise color-disjoint pairs, where $\ell=\Delta$ if $\Delta+1$ is even and $\ell=\Delta+1$ otherwise. The sets and a lookup table that, for every unordered pair $\{\alpha,\beta\}$, returns the unique set containing it can be constructed in $O(\Delta^2\log\Delta)$ work and $O(\log\Delta)$ span. Afterwards, the set containing any given unordered pair can be determined in $O(1)$ work and span.
\end{lemma}

\begin{proof}
Apply Lemma~\ref{lem:complete-graph-edge-coloring} and let the sets $\mathcal{P}_1,\ldots,\mathcal{P}_\ell$ be the edge-color classes of the resulting proper edge coloring of $K_{\Delta+1}$. Since every edge-color class is a matching, every $\mathcal{P}_j$ consists of pairwise color-disjoint pairs, and every unordered pair of distinct palette colors belongs to exactly one set. For every unordered pair $\{\alpha,\beta\}$, create a record containing the pair and the color $j$ assigned to the corresponding edge of $K_{\Delta+1}$, and store $j$ in a table entry indexed by $\{\alpha,\beta\}$. The table therefore answers the set query in $O(1)$ work and span. To explicitly store the sets, sort the $\Theta(\Delta^2)$ records by $j$ and group equal keys, which takes $O(\Delta^2\log\Delta)$ work and $O(\log\Delta)$ span. Since the construction of the complete graph and its edge coloring is dominated by these bounds, in total it costs $O(\Delta^2\log\Delta)$ work and $O(\log\Delta)$ span.
\end{proof}

Partitioning the pairs is a one-time preprocessing step (before the start of the first round) since the partition depends only on the fixed palette $[\Delta+1]$ and remains unchanged throughout the algorithm. In each round, once a nontrivial u-edge $e$ obtains its required pair $\pi_e=(\alpha,\beta)$, we can determine in $O(1)$ work the unique set $\mathcal{P}_j$ containing $\{\alpha,\beta\}$.

\paragraph{The algorithm.}
Recall that the goal is to color every currently uncolored edge, and that the uncolored edges form a matching of size $O(m/\Delta)$. Algorithm~\ref{alg:det-color-matching} presents the full procedure. In line~3, before the first round, the color-pair partition $\mathcal{P}_1,\ldots,\mathcal{P}_\ell$ is constructed using Lemma~\ref{lem:color-pair-partition}; this is done only once since the partition depends only on the palette. The algorithm then repeats the following until all uncolored edges are colored (lines~4--26). First, for every currently uncolored edge $e$, its fan $F_e$ is constructed in parallel (line~5), the fan is classified as trivial or nontrivial (line~6), and, if it is nontrivial, its required pair $\pi_e=(\alpha,\beta)$ is determined (line~7). Using this required pair, the unique set $\mathcal{P}_j$ containing the unordered pair $\{\alpha,\beta\}$ is found. After all fans have been constructed, the number of nontrivial u-edges belonging to each set is counted, and the set containing the largest number of them is found (lines~9--13). More concretely, let $T$ be the set of u-edges whose fans are trivial, and for each $j\in[\ell]$ let $A_j$ be the set of nontrivial u-edges $e$ whose required pair $\pi_e=(\alpha,\beta)$ satisfies $\{\alpha,\beta\}\in\mathcal{P}_j$, so let $j^\star\in[\ell]$ maximize $|A_j|$. Then, in lines~14--17, the larger set between the set of trivial u-edges and the largest set of nontrivial u-edges is activated: if $|T|\ge |A_{j^\star}|$, activate the trivial u-edges by setting $A:=T$; otherwise, activate the nontrivial u-edges of the largest set by setting $A:=A_{j^\star}$. In the latter case, all distinct unordered required pairs used by $A$ belong to the same color-disjoint set $\mathcal{P}_{j^\star}$. For these nontrivial active edges, the pair-indexed graph $G'$ is constructed, and the corresponding bichromatic components are labeled using Lemma~\ref{lem:parallel-pairs} (line~18); trivial active edges require no bichromatic components. Then, in line~20, the conflict graph $G_{\mathrm{conf}}(A)$ is constructed using Lemma~\ref{lem:build-conf-graph}, and in line~21, Lemma~\ref{lem:large-is} is applied to obtain a conflict-free subset $I\subseteq A$. Finally, the Vizing operations corresponding to all edges of $I$ are performed in parallel (lines~22--24). Since $I$ is conflict-free, these operations do not interfere; the resulting partial coloring remains proper, and every edge of $I$ becomes colored. The next round is applied to the remaining uncolored edges, with the designated missing colors, fans, required pairs, set sizes, bichromatic components, and conflicts recomputed again from the current coloring. 

\begin{lemma}\label{lem:parallel-vizing-operations}
Let $I$ be a conflict-free set of u-edges such that the underlying color pairs of its nontrivial u-edges belong to a color-disjoint set $\mathcal{P}$, and suppose that all required bichromatic components, including their endpoints, have already been identified. Then all Vizing operations corresponding to $I$ can be performed in parallel in $O(m)$ work and $O(1)$ span, coloring every edge of $I$ while preserving properness.
\end{lemma}

\begin{proof}
For every nontrivial u-edge $e\in I$, let $ F_e=(u,\alpha),(v_1,c_1), \ldots,(v_k,c_k)$, $\beta:=c_k$, and $P_e$ be its required $\{\alpha,\beta\}$-alternating path. If $F_e$ is nontrivial, there is an index $j<k$ with $c_j=\beta$, which is stored with the indexed representation of $F_e$.  Since $I$ is conflict-free, two u-edges with the same underlying color pair cannot use the same bichromatic path. Thus their required paths are distinct, while paths corresponding to different underlying pairs belong to color-disjoint pairs. Hence, by Lemma~\ref{lem:parallel-pair-flips}, all required paths can be flipped simultaneously, including the corresponding data-structure updates, in $O(m)$ work and $O(1)$ span. The resulting coloring remains proper, and the set of uncolored edges is unchanged.

We next determine the trivial fan to be rotated for each $e\in I$. If $F_e$ was already trivial, no path was flipped and we keep $F'_e :=F_e$ as is. Now consider a nontrivial fan $F_e$. Before the flip, $\alpha$ is missing at $u$ and $\beta$ is not, so $P_e$ starts at $u$ with the unique incident $\beta$-colored edge. After flipping $P_e$, the color $\beta$ becomes missing at $u$.

Let $z_e$ be the other endpoint of $P_e$. If $z_e\neq v_j$ (and even if $z_e = v_k$), then $\beta$ remains missing at $v_j$, and none of the fan edges $(u,v_1),\ldots,(u,v_j)$ is changed by the flip. Hence $ F'_e=(u,\beta),(v_1,c_1), \ldots,(v_j,c_j) $ is a trivial fan. If instead $z_e=v_j$, then flipping $P_e$ changes the missing color at $v_j$ from $\beta$ to $\alpha$, while $\beta$ remains missing at $v_k$. Therefore, after replacing the stored color $c_j$ by $\alpha$, the full sequence $ F'_e=(u,\beta),(v_1,c_1),\ldots,(v_j,\alpha),\ldots,(v_k,c_k) $ is a trivial fan. Thus, in either case, the path flip produces an explicitly represented trivial fan $F'_e$ whose rotation colors the original u-edge.

Since the endpoint $z_e$, the index $j$, and all fan entries are already stored, deciding which of the two cases applies and setting the resulting fan length and, when necessary, updating the entry at index $j$, takes for all $e\in I$, $O(|I|)=O(m)$ work and $O(1)$ span.

It remains to rotate these trivial fans. Since $I$ has no path-fan conflicts, flipping $P_f$ for any distinct $e,f\in I$ does not modify the fan information needed for $F_e$. Furthermore, since $I$ has no fan-fan conflicts, the original fans $\{F_e:e\in I\}$ are pairwise vertex-disjoint. Each $F'_e$ is either $F_e$ itself or a prefix of $F_e$, so the resulting trivial fans $\{F'_e:e\in I\}$ are also pairwise vertex-disjoint. Therefore, by Lemma~\ref{lem:parallel-fan-rotations}, all of them can be rotated simultaneously, including the corresponding data-structure updates, in $O(m)$ work and $O(1)$ span, while preserving properness and coloring every u-edge in $I$. Thus the claim follows.
\end{proof}

We are now ready to state the final theorem for the deterministic algorithm. Here first we give the intuition behind the number of rounds. Once the number of rounds is bounded, the work and span bounds follow from the per-round costs of the subroutines, together with an amortization over the sizes of the active sets throughout the execution. To bound the number of rounds, consider the worst case in which most of the remaining fans are nontrivial. Their required pairs are partitioned among $O(\Delta)$ color-disjoint sets, so the largest set contains the required pairs of an $\Omega(1/\Delta)$ fraction of the currently uncolored edges. We activate the edges associated with this set and construct their conflict graph. By the sparsity of this graph, we can extract a conflict-free subset containing an $\Omega(1/\Delta^2)$ fraction of the active edges. Thus, each round colors an $\Omega(1/\Delta^3)$ fraction of the remaining uncolored edges, implying that $\tilde{O}(\Delta^3)$ rounds suffice.

\begin{algorithm}[H]
\caption{Deterministic color-matching algorithm}\label{alg:det-color-matching}
\begin{algorithmic}[1]
\STATE \textbf{Input:} a partial $(\Delta+1)$-edge coloring $\chi$ with uncolored edges $U_\chi$ that form a matching of size $O(m/\Delta)$
\STATE \textbf{Output:} a partial $(\Delta+1)$-edge coloring $\chi$ with the edges of $U_\chi$ colored
\STATE construct the color-pair partition $\mathcal{P}_1,\ldots,\mathcal{P}_\ell$ (Lemma~\ref{lem:color-pair-partition})
\WHILE{$U_\chi\neq\varnothing$}
    \STATE construct the fan $F_e$ for every $e\in U_\chi$ in parallel (Lemma~\ref{lem:parallel-fans})
    \STATE classify every $F_e$ as trivial or nontrivial in parallel 
    \STATE compute the required pair $\pi_e$ of every nontrivial fan in parallel
    \STATE let $T:=\{e\in U_\chi:F_e\text{ is trivial}\}$
    \FORALL{nontrivial $e\in U_\chi$ \textbf{in parallel}}
        \STATE find the unique $j(e)\in[\ell]$ such that the unordered pair underlying $\pi_e$ belongs to $\mathcal{P}_{j(e)}$
    \ENDFOR
    \STATE let $A_j:=\{e\in U_\chi:F_e\text{ is nontrivial and }j(e)=j\}$ for every $j\in[\ell]$
    \STATE find $j^\star\in\arg\max_{j\in[\ell]}|A_j|$
    \IF{$|T|\ge |A_{j^\star}|$}
        \STATE $A:=T$
    \ELSE
        \STATE $A:=A_{j^\star}$
        \STATE construct the pair-indexed graph $G'$ for the distinct required pairs used by $A$ and label their bichromatic components (Lemma~\ref{lem:parallel-pairs})
    \ENDIF
    \STATE construct the conflict graph $G_{\mathrm{conf}}(A)$ (Lemma~\ref{lem:build-conf-graph})
    \STATE compute a large independent set $I$ of $G_{\mathrm{conf}}(A)$ (Lemma~\ref{lem:large-is})
    \FORALL{$e\in I$ \textbf{in parallel}}
        \STATE perform the Vizing operation of $F_e$ and color $e$
    \ENDFOR
    \STATE update $U_\chi$ to the remaining uncolored edges
\ENDWHILE
\STATE \textbf{return} $\chi$
\end{algorithmic}
\end{algorithm}

\begin{restatedDeterTheoremP}
Given a partial $(\Delta+1)$-edge coloring whose uncolored edges form a matching of size $O(m/\Delta)$, Algorithm~\ref{alg:det-color-matching} colors all of them deterministically in $O(m\Delta^3\log^2 n)$ work and $O(\Delta^3\log^4 n)$ span.
\end{restatedDeterTheoremP}

\begin{proof}
Let $\lambda_t$ denote the number of uncolored edges at the beginning of round $t$, let $A_t$ denote the active set selected in that round, and let $I_t\subseteq A_t$ denote the conflict-free set to be colored. We first bound the progress made in one round and hence the total number of rounds.

Let $T_t$ be the set of u-edges whose fans are trivial in round $t$, and for each $j\in[\ell]$ let $A_{j,t}$ be the set of nontrivial u-edges whose unordered required pair belongs to $\mathcal{P}_j$ (in round $t$). By Lemma~\ref{lem:color-pair-partition}, every unordered pair of distinct colors belongs to exactly one set, so the sets $A_{1,t},\ldots,A_{\ell,t}$ partition the nontrivial u-edges. Therefore
\[
\sum_{j=1}^{\ell}|A_{j,t}|=\lambda_t-|T_t|.
\]
Let $j^\star$ maximize $|A_{j,t}|$. Then
\[
|A_{j^\star,t}|\ge\frac{\lambda_t-|T_t|}{\ell}.
\]
Recall from Lemma~\ref{lem:color-pair-partition} that $\ell=O(\Delta)$. Since the algorithm chooses the larger of $T_t$ and $A_{j^\star,t}$, if $|T_t|\ge\lambda_t/2$ then $|A_t| = |T_t|\ge\lambda_t/2$. Otherwise, (in the worst case) $\lambda_t-|T_t|>\lambda_t/2$, and hence
\[
|A_t|
\ge |A_{j^\star,t}|
\ge \frac{\lambda_t-|T_t|}{\ell}
> \frac{\lambda_t}{2\ell}
=\Omega\!\left(\frac{\lambda_t}{\Delta}\right).
\]
Thus either way $|A_t|=\Omega(\lambda_t/\Delta)$. By Lemma~\ref{lem:large-is}, we compute a conflict-free set of size
\[
|I_t|
=\Omega\!\left(\frac{|A_t|}{\Delta^2}\right)
=\Omega\!\left(\frac{\lambda_t}{\Delta^3}\right).
\]
Hence there is an absolute constant $c>0$ such that $|I_t|\ge c\lambda_t/\Delta^3$. By Lemma~\ref{lem:parallel-vizing-operations}, all edges of $I_t$ are colored simultaneously, and no new uncolored edge is created. Consequently,
\[
\begin{aligned}
\lambda_{t+1}
&=\lambda_t-|I_t|\\
&\le\left(1-\frac{c}{\Delta^3}\right)\lambda_t.
\end{aligned}
\]
Iterating gives
\[
\begin{aligned}
\lambda_t
&\le \lambda_0\left(1-\frac{c}{\Delta^3}\right)^t\\
&\le \lambda_0\exp\!\left(-\frac{ct}{\Delta^3}\right)\\
&\le m\exp\!\left(-\frac{ct}{\Delta^3}\right),
\end{aligned}
\]
where we used $1-x\le e^{-x}$ and $\lambda_0 = O(m/\Delta)$ which certainly is at most $m$. Since $G$ is simple, $m < n^2$, so for $t=c'\Delta^3\log n$ with a sufficiently large absolute constant $c'$, the right-hand side is less than one. Since $\lambda_t$ is a nonnegative integer, this implies $\lambda_t=0$. Therefore, the algorithm terminates after $O(\Delta^3\log n)$ rounds.

We next analyze the cost of one round. The color-pair partition in line~3 is constructed only once. By Lemma~\ref{lem:color-pair-partition}, this takes $O(\Delta^2\log\Delta)$ work and $O(\log\Delta)$ span. Consider a round $t$. By Lemma~\ref{lem:parallel-fans}, constructing all fans, classifying them as trivial or nontrivial, and determining the required pair of every nontrivial fan takes $O(m\log\Delta)$ work and $O(\log\Delta)$ span.

For every nontrivial u-edge $e$, Lemma~\ref{lem:color-pair-partition} determines in $O(1)$ work and span the unique set index $j(e)$ containing its unordered required pair. We sort the nontrivial u-edges by $j(e)$ and group equal set indices, thereby forming all sets $A_{j,t}$ and determining their sizes. An index $j^\star$ maximizing $|A_{j,t}|$ is then found by a parallel reduction. Lines~8--13 require $O(\lambda_t\log\lambda_t)$ work and $O(\log n)$ span, where $\lambda_t \le \lambda_0 = O(m/\Delta)$. Choosing between $T_t$ and $A_{j^\star,t}$ then requires at most linear additional work and $O(1)$ additional span.

If $A_t=A_{j^\star,t}$ consists of nontrivial u-edges, then all distinct unordered required pairs used by $A_t$ belong to the color-disjoint set $\mathcal{P}_{j^\star}$. Therefore, by Lemma~\ref{lem:parallel-pairs}, the pair-indexed graph $G'$ can be constructed and all required bichromatic components can be identified and labeled in $O(m\log n)$ work and $O(\log n)$ span. If $A_t=T_t$, all active fans are trivial and this step is unnecessary. Lemma~\ref{lem:build-conf-graph} then constructs $G_{\mathrm{conf}}(A_t)$ in $O(m \log n+ \Delta^2|A_t|\log n)$ work and $O(\log n)$ span, and Lemma~\ref{lem:large-is} computes $I_t$ in $O(\Delta^2|A_t|\,\alpha(n)\log n)$ work and $O(\log^3 n)$ span.

Finally, since $I_t$ is conflict-free, and in the nontrivial case its underlying color pairs belong to the same color-disjoint set and its required bichromatic components have already been identified, Lemma~\ref{lem:parallel-vizing-operations} performs all Vizing operations corresponding to $I_t$ in $O(m)$ work and $O(1)$ span. Updating $U_\chi$ to the remaining uncolored edges requires $O(\lambda_t)$ work and $O(\log n)$ span. Thus every round has $O(\log^3 n)$ span. Since there are $O(\Delta^3\log n)$ rounds, the total span is $O(\Delta^3\log^4 n)$.

We now bound the total work. The $O(\Delta^2|A_t|\log n)$ part from the work for constructing the conflict graph is dominated by the $O(\Delta^2|A_t|\,\alpha(n)\log n)$ work for computing the independent set. All remaining work in a round (including the $O(m\log n)$ term of the conflict graph) in total costs $O(m\log n)$ work. Hence the work of round $t$ is $O(m\log n+\Delta^2|A_t|\,\alpha(n)\log n)$. The first term costs in total $O(m\Delta^3\log^2 n)$ over all rounds.

For the active-set-dependent term, Lemma~\ref{lem:large-is} gives an absolute constant $c''>0$ such that $|I_t|\ge (c''|A_t|)/\Delta^2$, and therefore $|A_t|=O(\Delta^2|I_t|)$. Every edge of $I_t$ is colored in round $t$ and is never uncolored again, so the sets $I_t$ are pairwise disjoint and $\sum_t|I_t|=\lambda_0 = O(m/\Delta)$. Consequently,
\[
\begin{aligned}
\sum_t|A_t|
&=O(\Delta^2)\sum_t|I_t|\\
&=O(m\Delta).
\end{aligned}
\]
It follows that the total active set dependent work is
\[
\begin{aligned}
\sum_t O\!\left(\Delta^2|A_t|\,\alpha(n)\log n\right)
&=O\!\left(\Delta^2\alpha(n)\log n\sum_t|A_t|\right)\\
&=O\!\left(m\Delta^3\alpha(n)\log n\right).
\end{aligned}
\]
Adding the two contributions, and observing that the one-time $O(\Delta^2\log\Delta)$ preprocessing work is dominated, the total work is $O(m\Delta^3\log^2 n+m\Delta^3\alpha(n)\log n)$, since $\alpha(n)=O(\log n)$, the claim follows.

Finally, we verify correctness. If $A_t$ consists of nontrivial u-edges, all of its underlying color pairs belong to the single color-disjoint set $\mathcal{P}_{j^\star}$, and all required bichromatic components are identified before the conflict graph is constructed. If $A_t$ consists of trivial u-edges, no bichromatic components are required. In either case, $I_t$ is an independent set of $G_{\mathrm{conf}}(A_t)$ and hence contains no conflict. Therefore, the requirements of Lemma~\ref{lem:parallel-vizing-operations} are satisfied, so all Vizing operations corresponding to $I_t$ can be performed simultaneously, preserving properness, coloring every edge of $I_t$, and creating no new uncolored edge. Hence the remaining uncolored edges form a subset of the previous uncolored matching and therefore remain a matching.
\end{proof}

\section{The Randomized Algorithm}\label{sec:randomized}
We now describe our main randomized algorithm, which achieves $\tilde{O}(\Delta^2)$ span. We first present the high-level ideas of one round and then describe each step in detail; the complete procedure is given in Algorithm~\ref{alg:rand-color-match} in Section~\ref{subsec:complete-rand}. Recall that $\chi$ is a partial $(\Delta+1)$-edge coloring whose uncolored edges form a matching, and let $W_0$ denote the set of their distinct centers, so $|W_0|=\lambda$.~\footnote{We note that, unlike the deterministic algorithm, the randomized algorithm and its bounds do not require the number of uncolored edges to be $O(m/\Delta)$. In fact, the bounds in the randomized algorithm remain valid for any partial coloring whose uncolored edges form a matching, even when the number of uncolored edges is $O(m)$. We therefore state and analyze the randomized algorithm in this more general setting.} Each round consists of four main steps.

\emph{\textbf{Step 1:} randomize the missing colors.} The algorithm first selects a power-of-two subpalette $\Gamma\subseteq[\Delta+1]$ and a subset $W\subseteq W_0$ containing more than half of the current centers on that subpalette, and then randomizes the designated missing color of every center in $W$ so that it is marginally uniform over $\Gamma$ (Section~\ref{subsec:rand-miss-color-par}). The crucial consequence is that, for any fixed center and any set of at most $\Delta^2$ other centers, the expected number of u-edges with the same designated missing color is only $O(\Delta)$, instead of the worst-case $\Theta(\Delta^2)$ (Theorem~\ref{thm:expectation-same-color}). In particular, this applies to the centers in its two-hop neighborhood. In fact, the corresponding number is $O(\Delta+\log n)$ with high probability (Theorem~\ref{thm:conc}). 

\emph{\textbf{Step 2:} build and classify the fans.} The algorithm next constructs the fans using the randomized designated missing colors, classifies each fan as trivial or nontrivial, and determines the required pair of every nontrivial fan. This step must be performed after Step~(1), since the previous step may change the corresponding fans, their classification, and their required pairs.

\emph{\textbf{Step 3:} randomly match the palette and activate fans.} The algorithm samples a uniformly random matching of the full palette $[\Delta+1]$ and independently orders each of its pairs (Section~\ref{subsec:random-matching}). A nontrivial fan is \defn{active} if its ordered required pair belongs to the resulting ordered matching, while each trivial fan is activated independently with probability $1/(2\Delta)$. Thus every center is activated with probability $\Theta(1/\Delta)$, and conditioned on a fixed center being active, with constant probability it has only $O(\Delta)$ conflicts that prevent it from being colored in parallel. (Lemma~\ref{lem:active-conflict-degree}).

\emph{\textbf{Step 4:} random sampling and conflict removal.} Finally, the algorithm independently samples each active center with probability $\Theta(1/\Delta)$ and identifies the bichromatic components required by the sampled nontrivial centers. It then detects all conflicts among the sampled centers, orients every conflict toward one of its endpoints, and removes every sampled center toward which some conflict is oriented (Section~\ref{subsec:random-sampling}). The surviving set $\mathcal{S}$ is therefore conflict-free. They can be colored simultaneously by Lemma~\ref{lem:parallel-vizing-operations}. Moreover, every fixed center $u\in W$ belongs to $\mathcal{S}$ with probability $\Omega(1/\Delta^2)$ (Theorem~\ref{thm:conflict-free-set}), and hence a round colors $\Omega(|W|/\Delta^2)=\Omega(\lambda/\Delta^2)$ edges in expectation.
Next, we describe the procedures for Steps 1, 3, and 4 in detail.

\subsection{Randomizing Missing Colors in Parallel}\label{subsec:rand-miss-color-par}
Let $W_0\subseteq V$ be the set of centers, let $\varphi_0(u)\in\miss_{\chi_0}(u)$ denote the initial designated missing color of $u\in W_0$, and $\chi_0:=\chi$. Our randomizing missing colors procedure is defined on a palette whose size is a power of two. We therefore first restrict to a power-of-two subpalette, chosen so that it contains the designated missing colors of at least a constant fraction of the centers, and then randomize the missing colors of those centers. 

Let $b:= \left\lfloor\log_2(\Delta+1)\right\rfloor$ and $k:=2^b$. Then $(\Delta+1)/2 < k \le \Delta+1$. For each color $c\in[\Delta+1]$, let $ W_c:=\{u\in W_0:\varphi_0(u)=c\}$. Choose a set $\Gamma\subseteq[\Delta+1]$ consisting of the $k$ colors with the largest values of $|W_c|$, breaking ties deterministically and treating colors that do not occur as designated colors as having $|W_c|=0$. Since $\Gamma$ contains the $k$ largest color classes,
\[
    |W|= \left| \{ u \in W_0: \varphi_0(u) \in \Gamma\} \right|
    = \sum_{c\in \Gamma}|W_c|
    \ge \frac{k}{\Delta+1}|W_0|
    > \frac{|W_0|}{2}.
\]
Therefore, after discarding only a constant fraction of the centers, we continue the current invocation of the randomizing missing color procedure on those centers whose designated missing color lies in the power-of-two subpalette $\Gamma$. From now on, we work on this subset of centers denoted by $W$, so that $\varphi_0(u)\in \Gamma$ for every $u\in W$. The discarded centers are not used in the current round of the algorithm. 

The subpalette $\Gamma$ can be constructed in parallel in $O((|W_0|+\Delta)\log n)$ work and $O(\log n)$ span: we group the centers by their designated colors, compute the sizes of the resulting color classes, and select the $k$ largest classes.

\begin{lemma}\label{lem:power-two-subpalette}
Let $W_0$ be a set of centers with designated colors $\varphi_0(u)\in[\Delta+1]$, and let $k:=2^{\lfloor\log_2(\Delta+1)\rfloor}$. In $O((|W_0|+\Delta)\log n)$ work and $O(\log n)$ span one can compute a set $\Gamma\subseteq[\Delta+1]$ of $k$ colors such that $W := \{u\in W_0:\varphi_0(u)\in \Gamma\}$ and $|W| >|W_0|/2$.
\end{lemma}

\begin{proof}
For each color $c\in[\Delta+1]$ let $W_c:=\{u\in W_0:\varphi_0(u)=c\}$. Compute all sizes $|W_c|$ by sorting the centers by $\varphi_0(u)$ and grouping equal colors, assigning zero to colors that do not occur. Then sort the $\Delta+1$ colors by decreasing $|W_c|$, breaking ties by color identifier, and let $\Gamma$ consist of the first $k$ of them. As the $k$ largest classes, those in $\Gamma$ hold at least a $k/(\Delta+1)$ fraction of the total, since their average size is at least the average over all $\Delta+1$ classes; hence
\[
    \sum_{c\in \Gamma}|W_c|
    \ge \frac{k}{\Delta+1}\sum_{c\in[\Delta+1]}|W_c|
    = \frac{k}{\Delta+1}\,|W_0|
    > \frac{|W_0|}{2},
\]
where the last inequality uses $k>(\Delta+1)/2$. The first sort is on $|W_0|$ records and the second handles $\Delta+1$ records, so the total cost is $O\bigl((|W_0|+\Delta)\log n\bigr)$ work and $O(\log n)$ span.
\end{proof}

This preprocessing is performed once, before the randomization phases. Since $|W_0|+\Delta=O(m)$, its cost is $O(m\log n)$ work and $O(\log n)$ span, dominated by the total cost of the randomization phases. Note that since more than half of the original centers are kept, any later step that processes a constant fraction of these remaining centers still processes a constant fraction of the original centers. Thus the restriction to $\Gamma$ does not change the asymptotic number of rounds, work, or span of the algorithm.

Now since $|\Gamma|=k=2^b$, we identify the colors of $\Gamma$ with the $b$-bit strings in $\{0,1\}^b$, one per color. The goal of the procedure below is to randomize the designated missing colors of the remaining centers in $b=O(\log\Delta)$ phases, while keeping every designated color inside $\Gamma$. After the procedure, every center has, in expectation, only $O(\Delta)$ nearby centers with the same final designated color. Throughout the procedure, the coloring remains proper, and the set of uncolored edges does not change. For each color $c\in \Gamma$ and each $i\in[b]$, let $c^{(i)}$ denote the $i$-th bit of $c$, and let $\operatorname{flip}_i(c)$ denote the color of $\Gamma$ obtained by flipping the $i$-th bit of $c$. For each $i\in[b]$, pair the colors of $\Gamma$ according to their $i$-th bit:
\[
    \mathcal{M}_i :=
    \bigl \{ \{c,\operatorname{flip}_i(c)\} : c \in \Gamma \bigr\},
\]
where each unordered pair is included only once.

\begin{observation}\label{obs:matching}
For every $i\in[b]$, $\mathcal{M}_i$ is a perfect matching of  $\Gamma$: it consists of $2^{b-1}$ pairwise color-disjoint pairs, and every color $c\in \Gamma$ lies in exactly one pair of $\mathcal{M}_i$, namely $\{c,\operatorname{flip}_i(c)\}$.
\end{observation}

This is what makes each phase parallelizable. In phase $i$, any two centers use either the same color pair or two disjoint color pairs, since the pairs of $\mathcal{M}_i$ are pairwise color-disjoint. Each component $K_\chi(p,u)$ of a two-colored subgraph $H_\chi(p)$ is a path or an even cycle, with the two colors alternating. For a fixed phase $i$, every colored edge whose color belongs to $\Gamma$ belongs to exactly one subgraph $H_\chi(p)$ with $p\in\mathcal{M}_i$. Thus $\mathcal{M}_i$ satisfies the color-disjointness hypothesis of Lemma~\ref{lem:parallel-pairs}, and all phase-$i$ flips can be carried out simultaneously. The underlying coloring continues to use the full palette $[\Delta+1]$; the randomization procedure only flips bichromatic components whose two colors belong to $\Gamma$, and hence never recolors an edge whose color lies outside $\Gamma$.

\begin{observation}\label{obs:endpoint}
Let $c \in \miss_\chi(u)$ and let $p=\{c,c'\}$. Then $\deg_{H_\chi(p)}(u) \le 1$, so $u$ is an endpoint of a path component $K_\chi(p,u)$ of $H_\chi(p)$ (or an isolated vertex). Moreover, $u$ shares the component $K_\chi(p,u)$ with at most one other center $v\neq u$ that has a color of $p$ missing, namely the opposite endpoint of the path.
\end{observation}

\paragraph{Procedure Randomize Missing Colors.}
For $i=1,\ldots,b$, perform the following three steps. For every $u\in W$, let $p_i(u) := \{ \varphi_{i-1}(u), \operatorname{flip}_i( \varphi_{i-1}(u))\} \in \mathcal{M}_i$ be its phase-$i$ color pair, and let $K_i(u) := K_{\chi_{i-1}}(p_i(u),u)$ be the component of $H_{\chi_{i-1}}(p_i(u))$ containing $u$. Since $\varphi_{i-1}(u)$ is missing at $u$, by Observation~\ref{obs:endpoint} the component $K_i(u)$ is a path with $u$ as an endpoint or the single vertex $u$.
\begin{enumerate}[label=(\arabic*)]
    \item Toss an independent fair coin $\xi_{i,p,K}\in\{0,1\}$ for each distinct component $K=K_i(u)$, $u\in W$. If two centers are the two endpoints of the same component, they read the same coin.
    \item Simultaneously flip every such component that contains at least one edge and whose coin is $\xi_{i,p,K}=1$. Let $\chi_i$ denote the resulting coloring.
    \item For every $u\in W$, set
    \[
        \varphi_i(u) =
        \begin{cases}
        \operatorname{flip}_i(\varphi_{i-1}(u)), & \text{if }\xi_{i,p_i(u),K_i(u)}=1,\\[2pt]
        \varphi_{i-1}(u), & \text{otherwise.}
        \end{cases}
    \]
    If $K_i(u)$ is the single vertex $u$, then both colors of $p_i(u)$ are missing at $u$; no edge is flipped, and the coin only determines which of them becomes $\varphi_i(u)$.
\end{enumerate}

For a center considered on its own, phase $i$ randomizes bit $i$ of its designated color using a fair coin. Hence, after all $b$ phases, its designated color is uniform on the subpalette $\Gamma$. The designated colors of different centers are not necessarily independent, because two centers that are endpoints of the same alternating path use the same coin in a phase. We will show that this is the only source of dependence and that it is sufficiently limited for our purposes. 

\begin{lemma}\label{lem:rand-parallel-bit}
All flips of phase $i$ may be performed in parallel, in $O(m\log n)$ work and $O(\log n)$ span, and the result is independent of their order. Moreover, $\chi_i$ is a proper partial coloring and $U_{\chi_i}=U_{\chi_{i-1}}$.
\end{lemma}

\begin{proof}
The pairs of $\mathcal{M}_i$ are pairwise color-disjoint, and the flipped components are distinct components of the subgraphs $H_{\chi_{i-1}}(p)$ for $p\in\mathcal{M}_i$. Applying Lemma~\ref{lem:parallel-pairs} to $\chi_{i-1}$ and the set $\mathcal{M}_i$, all these components can be identified and any chosen collection of them can be flipped simultaneously in $O(m\log n)$ work and $O(\log n)$ span. The outcome is independent of the order of the flips, the resulting coloring $\chi_i$ is proper, and the set of uncolored edges is unchanged.
\end{proof}

\begin{lemma}\label{lem:invariant-miss-color}
For every $i\in\{0,\ldots,b\}$ and every $u\in W$, we have $\varphi_i(u)\in\miss_{\chi_i}(u)$.
\end{lemma}

\begin{proof}
We prove the claim by induction on $i$. For $i=0$ it holds by definition. Suppose that $\varphi_{i-1}(u)\in\miss_{\chi_{i-1}}(u)$, and let $p:=p_i(u) = \{\varphi_{i-1}(u),\operatorname{flip}_i(\varphi_{i-1}(u))\}$.

Since $\varphi_{i-1}(u)$ is missing at $u$, by Observation~\ref{obs:endpoint} the vertex $u$ is an endpoint of the component $K_i(u)$ or an isolated vertex of $H_{\chi_{i-1}}(p)$. Thus at most one edge incident to $u$ has a color in $p$, and if such an edge exists its color is $\operatorname{flip}_i(\varphi_{i-1}(u))$.

A flip belonging to a different pair $p'\in\mathcal{M}_i$ uses two colors disjoint from $p$ and therefore cannot change the presence of either color of $p$ at $u$. Hence only the component $K_i(u)$ can affect which color of $p$ is missing at $u$. If $\xi_{i,p,K_i(u)}=0$, then $K_i(u)$ is not flipped and $\varphi_i(u)=\varphi_{i-1}(u)$ remains missing. If $\xi_{i,p,K_i(u)}=1$ and $K_i(u)$ contains an edge, flipping it exchanges the two colors at the endpoint $u$, so $\operatorname{flip}_i(\varphi_{i-1}(u))=\varphi_i(u)$ becomes missing. If $K_i(u)$ is an isolated vertex, both colors of $p$ are missing and remain missing. Thus $\varphi_i(u)\in\miss_{\chi_i}(u)$ in every case.
\end{proof}

For $0\le i\le b$, let $\mathcal{F}_i:=\sigma(\xi_{1,\cdot,\cdot},\ldots,\xi_{i,\cdot,\cdot})$ denote the $\sigma$-field generated by the coins of the first $i$ phases. The subpalette $\Gamma$, its identification with $\{0,1\}^b$, and the initial coloring and designated colors are fixed before these coins are sampled.

\begin{lemma}\label{lem:agree-and-measurable}
For every $u\in W$, the colors $\varphi_i(u)$ and $\varphi_{i-1}(u)$ agree on every bit except possibly bit $i$. Hence, for all $j<i$, $\bigl(\varphi_i(u)\bigr)^{(j)}=\bigl(\varphi_j(u)\bigr)^{(j)}$. In particular, bit $i$ of $\varphi_b(u)$ is decided in phase $i$ and never changes again. Moreover, $\chi_i$ and $\varphi_i$ are determined by the coins of the first $i$ phases, and hence are $\mathcal{F}_i$-measurable.
\end{lemma}

\begin{proof}
By step~(3), $\varphi_i(u)$ equals to either $\varphi_{i-1}(u)$ or $\operatorname{flip}_i(\varphi_{i-1}(u))$, and therefore differs from $\varphi_{i-1}(u)$ only possibly in bit $i$. Since a later phase $\ell$ changes only bit $\ell$, bit $j$ is never changed after phase $j$.

For measurability, proceed by induction on $i$. The initial coloring $\chi_0$ and designated colors $\varphi_0$ are fixed. Assuming that $\chi_{i-1}$ and $\varphi_{i-1}$ are $\mathcal{F}_{i-1}$-measurable, the pairs $p_i(u)$ and components $K_i(u)$ are determined before the phase-$i$ coins are sampled. The coloring $\chi_i$ is obtained by flipping exactly those components whose phase-$i$ coins equal one, while $\varphi_i(u)$ is determined from $\varphi_{i-1}(u)$ and the coin $\xi_{i,p_i(u),K_i(u)}$. Hence both $\chi_i$ and $\varphi_i$ are $\mathcal{F}_i$-measurable.
\end{proof}

\begin{theorem}\label{thm:uniform}
For every $u\in W$, the final designated color $\varphi_b(u)$ is uniform on $\Gamma$.
\end{theorem}

\begin{proof}
Let $B^{(j)}$ denote bit $j$ of $\varphi_b(u)$, so that $\varphi_b(u)=(B^{(1)},\ldots, B^{(b)})$. Fix $i\in[b]$ and condition on $\mathcal{F}_{i-1}$. By Lemma~\ref{lem:agree-and-measurable}, the values $\varphi_{i-1}(u)$, $p_i(u)$, and $K_i(u)$ are then determined, whereas $\xi_{i,p_i(u),K_i(u)}$ is a fresh fair coin independent of $\mathcal{F}_{i-1}$. Therefore $ B^{(i)} = \bigl(\varphi_{i-1}(u)\bigr)^{(i)} \oplus \xi_{i,p_i(u),K_i(u)}$ is uniform on $\{0,1\}$ conditioned on $\mathcal{F}_{i-1}$. Since $B^{(1)},\ldots, B^{(i-1)}$ are already determined by $\mathcal{F}_{i-1}$, for every $x^{(1)},\ldots,x^{(i)}\in\{0,1\}$, $\Pr\left[ B^{(i)}=x^{(i)} \,\middle|\, B^{(1)}=x^{(1)},\ldots, B^{(i-1)}=x^{(i-1)} \right] = 1/2$. Consequently, for every color $x \in \{0,1\}^b$,

\[ 
\begin{aligned}
    \Pr[\varphi_b(u)=x] 
    &= \prod_{i=1}^b \Pr\left[ B^{(i)}=x^{(i)} \,\middle|\, B^{(1)}=x^{(1)},\ldots, B^{(i-1)}=x^{(i-1)} \right]  \\
    &= 2^{-b} = \frac{1}{k}
\end{aligned}
\]

Since $|\Gamma|=k$, $\varphi_b(u)$ is uniform on $\Gamma$.
\end{proof}

As noted before, the designated missing colors of different centers are not necessarily independent after the procedure, because two centers may share the same component coin in some phase. Therefore, to control this dependence, we fix an arbitrary center $u$ and track the centers whose designated colors still agree with that of $u$ on all bits exposed so far.

\begin{definition}[Tied]
Two distinct centers $u,w\in W$ are \defn{tied} in phase $i$ if $p_i(u)=p_i(w)$ and $K_i(u)=K_i(w)$; equivalently, they read the same coin $\xi_{i,p_i(u),K_i(u)}$ in phase $i$.
\end{definition}

In every phase, a center is tied to at most one other center. If $w$ is tied to $u$, then both are endpoints of the same path component $K_i(u)$, and by Observation~\ref{obs:endpoint} a path has at most one endpoint other than $u$. Thus ``tied in phase $i$'' defines a partial matching on $W$. Consequently, among any set of centers currently tracked with $u$, at most one uses the same phase-$i$ coin as $u$, while every other center uses a distinct fair coin. This limited dependence is enough to obtain the following expectation bound.

\begin{theorem}\label{thm:expectation-same-color}
Let $u\in W$ and let $C\subseteq W\setminus\{u\}$ with $|C|\le\Delta^2$. Then $\mathbb{E}\,\left[ \left|\{w\in C:\varphi_b(w)=\varphi_b(u)\}\right| \right] < 2 \Delta$. In other words, among any set of at most $\Delta^2$ other centers, the expected number of centers whose final designated missing color is the same as that of $u$ is $O(\Delta)$.
\end{theorem}

\begin{proof}
Fix $u$. For $0\le i\le b$, let $S_i:=\{w\in C:\varphi_i(w)\text{ and }\varphi_i(u) \text{ agree on bits }1,\ldots,i\}$, and  $m_i:=|S_i|$. Then $S_0=C$, and by Lemma~\ref{lem:agree-and-measurable}, $S_b=\{w\in C:\varphi_b(w)=\varphi_b(u)\}$. Fix $i\in[b]$ and condition on $\mathcal{F}_{i-1}$. By Lemma~\ref{lem:agree-and-measurable}, the set $S_{i-1}$ and its size $m_{i-1}$ are then determined. Since bits $1,\ldots,i-1$ never change again, $S_i\subseteq S_{i-1}$. For $w\in S_{i-1}$, we have $w\in S_i$ if and only if
\[
    \bigl(\varphi_{i-1}(w)\bigr)^{(i)} \oplus \xi_{i,p_i(w),K_i(w)}
 = \bigl(\varphi_{i-1}(u)\bigr)^{(i)} \oplus \xi_{i,p_i(u),K_i(u)}.
\]
At most one center of $S_{i-1}$ is tied to $u$ in phase $i$. If such a center exists, its contribution to $m_i$ is at most one. Now consider a center $w\in S_{i-1}$ that is not tied to $u$. Then $w$ and $u$ use distinct phase-$i$ coins, and $w$'s coin is fair and independent of both $\mathcal{F}_{i-1}$ and $u$'s phase-$i$ coin. Condition further on the value of $u$'s phase-$i$ coin. Once this value is fixed, exactly one of the two possible values of $w$'s coin makes the $i$-th bits of $\varphi_i(w)$ and $\varphi_i(u)$ agree. Therefore $\Pr[w\in S_i\mid\mathcal{F}_{i-1}] = 1/2$. If a tied center exists, we pessimistically assume that it always remains in $S_i$. By linearity of expectation, the remaining $m_{i-1}-1$ centers contribute an expected $(m_{i-1}-1)/2$, and hence
\[
    \mathbb{E}[m_i\mid\mathcal{F}_{i-1}]
    \le 1+\frac{1}{2}(m_{i-1}-1)
    = \frac{1}{2}m_{i-1}+\frac{1}{2}.
\]
If no tied center exists, then every center in $S_{i-1}$ remains in $S_i$ with probability $1/2$, and therefore  $\mathbb{E}[m_i\mid\mathcal{F}_{i-1}] = {(1/2)}m_{i-1}.$ Thus, in either case, $\mathbb{E}[m_i\mid\mathcal{F}_{i-1}] \le (1/2)m_{i-1}+(1/2)$. Taking expectations of both sides and using the tower property and linearity of expectation,
\[
    \mathbb{E}[m_i] \le \frac{1}{2}\mathbb{E}[m_{i-1}]+\frac{1}{2}.
\]
Unrolling this recurrence and using the fact that $m_0=|C|$, we obtain
\[
\begin{aligned}
    \mathbb{E}[m_b]
    &\le 2^{-b}|C|+1-2^{-b}\\
    &=\frac{|C|}{k}+1-\frac{1}{k}.
\end{aligned}
\]
Since $|C|\le\Delta^2$ and $k>(\Delta+1)/2$,
\[
\begin{aligned}
    \mathbb{E}[m_b]
    &\le \frac{\Delta^2-1}{k}+1\\
    &\le \frac{2(\Delta^2-1)}{\Delta+1}+1\\
    & =2\Delta-1\\
    & <2\Delta.
\end{aligned}
\]
\end{proof}

For our algorithm, a constant-probability guarantee for this step is sufficient, and hence we do not use the stronger high-probability bound in the remainder of the paper. Nevertheless, the high-probability statement below may be of independent interest and follows from a more refined analysis of the same randomization procedure. We therefore state it here and defer its proof to Appendix~\ref{sec:appendix}.

\begin{theorem}\label{thm:conc}
Let $u\in W$, let $C\subseteq W\setminus\{u\}$ with $|C|\le\Delta^2$, and let $D := |\{w\in C:\varphi_b(w)=\varphi_b(u)\}|$. There exists a constant $\kappa>0$ such that, for every $\lambda\ge1$, $\Pr[D\ge\kappa(\Delta+\lambda)] \le b e^{-\lambda}$. In particular, for $\lambda=\Theta(\log n)$ we have $D=O(\Delta+\log n)$ with probability $1-1/\operatorname{poly}(n)$, and $D=O(\Delta)$ whenever $\Delta=\Omega(\log n)$.
\end{theorem}

For a center $u\in W$, let $N_u$ denote the number of centers whose final designated color equals that of $u$ and whose fan conflicts with $F_u$:
\[
    N_u := \Bigl| \bigl\{ w\in W\setminus\{u\}: \varphi_b(w)=\varphi_b(u) \text{ and }  V(F_u)\cap V(F_w)\neq\emptyset \bigr\} \Bigr|.
\]

\begin{lemma}\label{lem:const-prob-same}
For every fixed center $u\in W$ and every $a\ge1$, $ \Pr[N_u\ge a\Delta] < 2/a$. In particular, $\Pr[N_u<20\Delta] > 9/10$.
\end{lemma}

\begin{proof}
Let $C(u) := \{w\in W\setminus\{u\}:\dist(u,w)\le2\}$. The center of every fan conflicting with $F_u$ belongs to $C(u)$, and $|C(u)|\le\Delta^2$. Hence $N_u \le \left| \{w\in C(u):\varphi_b(w)=\varphi_b(u)\} \right|$. By Theorem~\ref{thm:expectation-same-color}, $\mathbb{E}[N_u]<2\Delta$. Markov's inequality therefore gives, for every $a\ge1$,
\[
    \Pr[N_u\ge a\Delta]
    \le \frac{\mathbb{E}[N_u]}{a\Delta}
    < \frac{2\Delta}{a\Delta}  =  \frac{2}{a}.
\]
Setting $a=20$ yields $\Pr[N_u\ge20\Delta] < 1/10$, and consequently $\Pr[N_u<20\Delta] > 9/10$.
\end{proof}

\begin{lemma}\label{lem:randomize-missing-colors-runtime}
The Randomize Missing Colors procedure, including the preprocessing that constructs the subpalette $\Gamma$ and the subset of centers $W$, can be implemented in $O(m\log n\log\Delta)$ work and $O(\log n\log\Delta)$ span.
\end{lemma}

\begin{proof}
By Lemma~\ref{lem:power-two-subpalette}, the preprocessing constructs $\Gamma$ and the subset $W$ in $O((|W_0|+\Delta)\log n)=O(m\log n)$ work and $O(\log n)$ span.

It remains to bound the $b$ randomization phases. Fix a phase $i\in[b]$. By Observation~\ref{obs:matching}, the pairs of $\mathcal{M}_i$ are pairwise color-disjoint. Hence, by Lemma~\ref{lem:parallel-pairs}, all components of the corresponding bichromatic subgraphs can be identified and labeled in $O(m\log n)$ work and $O(\log n)$ span. For each center $u\in W$, we then determine whether $K_i(u)$ is isolated in $O(1)$ work; if it is, we assign it a unique singleton component identifier, and otherwise we use the component identifier returned by Lemma~\ref{lem:parallel-pairs}. Thus every center obtains an identifier for $K_i(u)$. We sort the centers by their component identifiers and assign one independent fair coin to each distinct component. This step takes $O(m\log n)$ work and $O(\log n)$ span. The nonisolated components whose coin is $1$ are then flipped simultaneously using Lemma~\ref{lem:parallel-pair-flips}, in $O(m)$ work and $O(1)$ span, including the required data-structure updates. Isolated components require no flip. Finally, updating $\varphi_i(u)$ for every $u\in W$ takes $O(|W|)$ work and $O(1)$ span. Hence each phase costs $O(m\log n)$ work and $O(\log n)$ span. Since there are $b=O(\log\Delta)$ phases and they run sequentially, all phases together require $O(m\log n\log\Delta)$ work and $O(\log n\log\Delta)$ span. The preprocessing is dominated by these bounds, thus the claim follows.
\end{proof}

\subsection{Random Matching of the Palette}\label{subsec:random-matching}
The Randomize Missing Colors procedure, in a single round of the algorithm, randomizes for each center $u\in W$ a designated missing color $\varphi_b(u)$ over the subpalette $\Gamma$. By Theorem~\ref{thm:expectation-same-color}, within any two-hop neighborhood only $O(\Delta)$ centers in expectation share a fixed center's designated color. Next, in Step~(2), the algorithm builds the fans under these new randomized missing colors (Lemma~\ref{lem:parallel-fans}) and classifies them as trivial or nontrivial, in $O(m\log\Delta)$ work and $O(\log\Delta)$ span; this also determines, for each nontrivial fan, its ordered required pair $\pi_{e_u}$.

We now describe Step~(3), in which the algorithm samples a random matching of the palette and keeps only the centers whose required pair belongs to it. Although $\varphi_b(u)\in\Gamma$ for every center $u$, the second color $\psi(u)$ of the required pair $\pi_{e_u}=(\varphi_b(u),\psi(u))$ is determined by the fan construction and might not belong to $\Gamma$. Therefore, in this step we return to the full palette $[\Delta+1]$ and sample the random matching from all $\Delta+1$ colors. The required pairs of the resulting active centers are color-disjoint, so their alternating paths can be flipped in parallel (Lemma~\ref{lem:parallel-pairs}). We first show that each center is activated with probability $\Theta(1/\Delta)$; we then show how the random matching bounds the number of \emph{different-color} conflicting active fans. Combined with the \emph{same-color} bound $O(\Delta)$ from step~(1) via Theorem~\ref{thm:expectation-same-color}, this gives total conflicts of $O(\Delta)$ with constant probability for every active center $u$.

\paragraph{Procedure Random Matching.}
\begin{enumerate}[label=(\arabic*)]
    \item \emph{Color matching.} Let $c_1,\ldots,c_{\Delta+1}$ be a uniformly random permutation of the palette $[\Delta+1]$. For each $j=1,\ldots,\lfloor(\Delta+1)/2\rfloor$, include the unordered pair $\{c_{2j-1},c_{2j}\}$ in $\mathcal{M}$. If $\Delta+1$ is even, $\mathcal{M}$ is a perfect matching of the palette; if $\Delta+1$ is odd, the last color $c_{\Delta+1}$ is left unmatched. In either case $\mathcal{M}$ consists of pairwise color-disjoint pairs.

    \item \emph{Ordering of pairs.} Independently for every pair $\{c,c'\}\in\mathcal{M}$, order the pair uniformly at random, choosing $(c,c')$ or $(c',c)$ with probability $1/2$ each. Let $M$ denote the resulting set of ordered pairs.
\end{enumerate}

For every center $u\in W$, the ordered required pair is $\pi_{e_u}=(\varphi_b(u),\psi(u))$, where $\varphi_b(u)$ is the designated missing color at $u$ and $\psi(u)$ is the color obtained by the fan. If $\varphi_b(u)=\psi(u)$ the fan is trivial and is handled separately; otherwise $\pi_{e_u}$ is an ordered pair of two distinct colors. We call a center $u$ with a nontrivial fan \defn{$M$-active} if $\pi_{e_u}\in M$. 

\begin{lemma}\label{lem:random-matching-runtime}
Given the ordered required pair $\pi_{e_u}$ of every nontrivial center $u\in W$, the Random Matching procedure and the identification of all $M$-active centers can be implemented in $O(m\log n)$ work and $O(\log n)$ span.
\end{lemma}

\begin{proof}
A uniformly random permutation of the $\Delta+1$ colors can be generated in $O(\Delta\log n)$ work and $O(\log n)$ span. Once the permutation is available, pairing consecutive colors and independently ordering every pair takes $O(\Delta)$ work and $O(1)$ span.

To identify the $M$-active centers, construct an array indexed by the colors in $[\Delta+1]$. For every ordered pair $(c,c')\in M$, store $c'$ as the outgoing partner of $c$; every color that is not the first color of an ordered pair stores a null value. Since every color belongs to at most one pair of $\mathcal{M}$, these entries can be written independently in $O(\Delta)$ work and $O(1)$ span. For every nontrivial center $u\in W$, test whether the outgoing partner of $\varphi_b(u)$ is $\psi(u)$. By definition, this holds if and only if $\pi_{e_u}=(\varphi_b(u),\psi(u))\in M$. Thus all $M$-active centers can be marked independently in $O(|W|)$ work and $O(1)$ span.  Since $|W|\le m$ and $\Delta\le m$, the total cost is $O(m\log n)$ work and $O(\log n)$ span.
\end{proof}

For an $M$-active center $u$, let $d_u$ denote the number of $M$-active centers $w\neq u$ such that $F_w$ conflicts with $F_u$ and $\varphi_b(w)=\varphi_b(u)$, and let $\bar{d}_u$ denote the number of such centers with $\varphi_b(w)\neq\varphi_b(u)$. Let $D_u:=d_u+\bar{d}_u$ be the total number of $M$-active centers whose fans conflict with $F_u$. The two terms are controlled by different sources of randomness: the missing-color randomization of the previous subsection bounds $d_u$, while the random matching of the palette bounds $\bar{d}_u$. We will show that, with constant probability, $D_u=O(\Delta)$.

\begin{lemma}\label{lem:M-active-prob}
For every fixed center $u\in W$ with a nontrivial fan,
\[
    \frac{1}{2(\Delta+1)}
    \le \Pr[u\text{ is $M$-active}]
    \le \frac{1}{2\Delta},
\]
so $\Pr[u\text{ is $M$-active}]=\Theta(1/\Delta)$.
\end{lemma}

\begin{proof}
First consider whether an unordered pair $\{\alpha,\beta\}$ is chosen by $\mathcal{M}$. When $\Delta+1$ is even, $\mathcal{M}$ is a perfect matching, so the match of any fixed color is uniform among the other $\Delta$ colors. Thus $ \Pr[\{\alpha,\beta\}\in\mathcal{M}] = 1/\Delta$.

When $\Delta+1$ is odd, $\mathcal{M}$ leaves one color unmatched. The color $\alpha$ is matched with probability $\Delta/(\Delta+1)$, and conditioned on this its match is uniform among the other $\Delta$ colors, so
\[
    \Pr[\{\alpha,\beta\}\in\mathcal{M}]
    = \frac{\Delta}{\Delta+1}\cdot\frac{1}{\Delta}
    = \frac{1}{\Delta+1}.
\]

In both cases, conditioned on $\{\alpha,\beta\}\in\mathcal{M}$, the ordering is uniform and independent, so the ordered pair equals $(\alpha,\beta)$ with probability $1/2$. Therefore
\[
    \Pr[(\alpha,\beta)\in M] =
    \begin{cases}
        \displaystyle \frac{1}{2\Delta}, & \text{if $\Delta+1$ is even},\\[6pt]
        \displaystyle \frac{1}{2(\Delta+1)}, & \text{if $\Delta+1$ is odd}.
    \end{cases}
\]

Fix $u$ and let $\pi_{e_u}=(\varphi_b(u),\psi(u))$ be its required pair with $\varphi_b(u)\neq\psi(u)$. Recall that $u$ is $M$-active if $\pi_{e_u}\in M$. Using $(\alpha,\beta)=(\varphi_b(u),\psi(u))$ and the above probabilities, we obtain
\[
    \frac{1}{2(\Delta+1)}
    \le \Pr[u\text{ is $M$-active}]
    \le \frac{1}{2\Delta}.
\]
Therefore, a center is $M$-active with probability $\Theta(1/\Delta)$.
\end{proof}

\begin{lemma}\label{lem:diff-color-conflicts}
Let $u\in W$ be a fixed center with a nontrivial fan. For every $a \ge 1$ and every $\Delta \ge3$,
\[
    \Pr\left[ \bar{d}_u\ge\frac{3a}{2}\Delta \,\middle|\, \pi_{e_u}\in M \right]
    \le \frac{1}{a}.
\]
In particular, setting $a=2$, $\Pr\left[ \bar{d}_u<3\Delta \,\middle|\, \pi_{e_u}\in M \right] \ge 1/2$.
\end{lemma}

\begin{proof}
Condition on the entire outcome of the previous randomization process and of the fan construction; all designated colors, fans, and ordered required pairs are then fixed, and the only remaining randomness is $M$. By Lemma~\ref{lem:M-active-prob},
\[
    \Pr[\pi_{e_u}\in M] \ge \frac{1}{2(\Delta+1)} > 0,
\]
so we may condition on this event. Let a center $w\neq u$ be \defn{bad} if its fan $F_w$ conflicts with $F_u$ and $\varphi_b(w)\neq\varphi_b(u)$. For every bad center $w$, define the indicator random variable $X_w:=1$ if $\pi_{e_w}\in M$ and $0$ otherwise, so that $X_w=1$ when $w$ is $M$-active. Then
\[
    \bar{d}_u=\sum_{w\text{ bad}}X_w.
\]
Recall that if $F_w$ conflicts with $F_u$, then $\dist(u,w)\le 2$. Therefore, the number of bad centers is at most $\Delta^2$. Fix a bad center $w$. Since $\varphi_b(w)\neq\varphi_b(u)$, we have $\pi_{e_w}\neq\pi_{e_u}$. We first show that if the two required pairs share a color, then $w$ cannot contribute; since the pairs in $M$ are color-disjoint, two distinct required pairs that share a color cannot both belong to $M$. In all such cases, $\Pr[X_w=1\mid\pi_{e_u}\in M] = 0$.

Suppose now that the two required pairs are color-disjoint. Conditioned on $\pi_{e_u}\in M$, the two colors of $\pi_{e_u}$ are removed and the remaining colors are matched by a uniformly random maximum matching. If $\varphi_b(w)$ is matched, its partner is uniform among the $\Delta-2$ other remaining colors. Therefore
\[
    \Pr[ \{\varphi_b(w),\psi(w)\}\in\mathcal{M} \mid \pi_{e_u}\in M ]
    \le \frac{1}{\Delta-2},
\]
with the independent uniform ordering,
\[
    \Pr[X_w=1\mid\pi_{e_u}\in M] \le \frac{1}{2(\Delta-2)}.
\]

Summing over the at most $\Delta^2$ bad centers and using linearity of expectation, we have
\[
\begin{aligned}
    \mathbb{E}\left[ \bar{d}_u \,\middle|\, \pi_{e_u}\in M \right]
    &= \sum_{w\text{ bad}} \Pr[X_w=1\mid\pi_{e_u}\in M]\\
    &\le \frac{\Delta^2}{2(\Delta-2)}\\
    &\le \frac{3\Delta}{2},
\end{aligned}
\]
where the last inequality holds for $\Delta \ge3$. By Markov's inequality,
\[
\begin{aligned}
    \Pr\left[ \bar{d}_u\ge\frac{3a}{2}\Delta \,\middle|\, \pi_{e_u}\in M \right]
    &\le \frac{ \mathbb{E}[ \bar{d}_u \mid \pi_{e_u}\in M ] }{(3a/2)\Delta}\\
    &\le \frac{(3/2)\Delta}{(3a/2)\Delta}\\
    & = \frac{1}{a}.
\end{aligned}
\]
Finally, setting $a=2$ and taking the complementary event yields
\[
    \Pr[ \bar{d}_u<3\Delta  \mid\pi_{e_u}\in M ] \ge \frac{1}{2}.
\]
\end{proof}

\begin{theorem}\label{thm:active-conf-bound}
For every fixed center $u\in W$ with a nontrivial fan and every $\Delta\ge3$,
\[
\Pr\left[ D_u<23\Delta \,\middle| \,  \pi_{e_u}\in M \right] \ge 2/5.
\]
Moreover,
\[
    \Pr\left[ \pi_{e_u}\in M \text{ and } D_u<23\Delta \right]
    \ge \frac{1}{5(\Delta+1)}.
\]
\end{theorem}

\begin{proof}
By Lemma~\ref{lem:const-prob-same}, $\Pr[N_u<20\Delta]\ge 9/10 $ and  $\Pr[N_u\ge20\Delta]\le 1/10$. Let $\mathcal{R}$ denote the entire outcome of the randomization process and the fan construction. Conditional on any fixed outcome of $\mathcal{R}$, the ordered pair $\pi_{e_u}$ is fixed. Since the fan of $u$ is nontrivial, Lemma~\ref{lem:M-active-prob} gives
\[
    \Pr[\pi_{e_u}\in M\mid\mathcal{R}] =
    \begin{cases}
        \displaystyle\frac{1}{2\Delta}, & \text{if $\Delta+1$ is even},\\[6pt]
        \displaystyle\frac{1}{2(\Delta+1)}, & \text{if $\Delta+1$ is odd}.
    \end{cases}
\]
In particular, this probability depends only on $\Delta$ and not on the outcome of $\mathcal{R}$. Hence the event $[\pi_{e_u}\in M]$ is independent of $\mathcal{R}$, and therefore is independent of the event $[N_u \ge 20\Delta]$. Consequently,
\begin{equation}\label{eq:Nu-cond}
    \Pr\left[ N_u\ge20\Delta \,\middle|\, \pi_{e_u}\in M \right]
    = \Pr[N_u\ge20\Delta]
    \le \frac{1}{10}.
\end{equation}
Recall from Lemma~\ref{lem:diff-color-conflicts} that
\begin{equation}\label{eq:dbar-cond}
    \Pr\left[ \bar{d}_u\ge3\Delta  \,\middle|\, \pi_{e_u}\in M \right]
    \le  \frac{1}{2}.
\end{equation}

We now combine \eqref{eq:Nu-cond} and \eqref{eq:dbar-cond}:
\[
\begin{aligned}
    &\Pr\left[  N_u<20\Delta  \text{ and }  \bar{d}_u<3\Delta \,\middle|\, \pi_{e_u}\in M \right]\\
    &\qquad=  1- \Pr\left[  N_u\ge20\Delta  \text{ or }  \bar{d}_u\ge3\Delta  \,\middle|\,  \pi_{e_u}\in M  \right]\\
    &\qquad\ge  1- \Pr\left[  N_u\ge20\Delta \,\middle|\,  \pi_{e_u}\in M \right] - \Pr\left[ \bar{d}_u\ge3\Delta  \,\middle|\, \pi_{e_u}\in M  \right]\\
    &\qquad\ge 1-\frac{1}{10}-\frac{1}{2}
    = \frac{2}{5}.
\end{aligned}
\]
Consider the event  $[N_u<20\Delta  \text{ and } \bar{d}_u<3\Delta]$. Since $d_u$ counts only those centers counted by $N_u$ that are additionally $M$-active, we have $d_u\le N_u$. Therefore
\[
    D_u = d_u+\bar{d}_u
    \le N_u+\bar{d}_u
    < 20\Delta+3\Delta
    = 23\Delta.
\]
Thus the event $[N_u<20\Delta\text{ and }\bar{d}_u<3\Delta]$ is contained in $[D_u<23\Delta]$, and hence
\[
\begin{aligned}
    \Pr\left[  D_u<23\Delta \,\middle|\,  \pi_{e_u}\in M \right]
    &\ge \Pr\left[  N_u<20\Delta  \text{ and } \bar{d}_u<3\Delta  \,\middle|\,  \pi_{e_u}\in M  \right]\\
    &\ge \frac{2}{5}.
\end{aligned}
\]

Finally, by Lemma~\ref{lem:M-active-prob}, $\Pr[\pi_{e_u}\in M] \ge 1/({2(\Delta+1)})$. Therefore
\[
\begin{aligned}
    \Pr\left[ \pi_{e_u}\in M \text{ and } D_u<23\Delta \right]
    &=  \Pr[\pi_{e_u}\in M]\, \Pr\left[  D_u<23\Delta \,\middle|\,  \pi_{e_u}\in M \right]\\
    &\ge \frac{1}{2(\Delta+1)} \cdot \frac{2}{5}\\
    &=
    \frac{1}{5(\Delta+1)}.
\end{aligned}
\]
\end{proof}

We have shown that a fixed center is $M$-active with probability $\Theta(1/\Delta)$ (by Lemma~\ref{lem:M-active-prob}), and conditioned on being $M$-active, it has fewer than $23\Delta$ conflicting active fans with probability at least $2/5$ (by Theorem~\ref{thm:active-conf-bound}). Recall that the remaining conflicts are bounded by Lemma~\ref{lem:path-fan-conflicts}. Fix a constant $c_0>0$ such that, for every active center $u$, the number of active centers $w$ for which $P_w$ has a path-fan conflict with $F_u$ is at most $c_0\Delta$. Thus, whenever $D_u<23\Delta$, the total number of active centers that can block $u$ is less than $(23+c_0)\Delta$. A conflict degree of $O(\Delta)$ is still too large to process all active centers in parallel. Later on, the algorithm therefore performs an independent subsampling with probability $\Theta(1/\Delta)$, followed by a deterministic conflict-removal step.

\subsection{Trivial Fans and Activated Conflicts}\label{subsec:trivial}

The previous step of the algorithm activates only \emph{nontrivial} fans, through the random matching $M$. Trivial fans need separate treatment, since the two colors in the required pair of a trivial fan are the same, so its required pair cannot belong to $M$. As in the deterministic algorithm, one could instead separate the trivial and nontrivial fans and process only the larger one. In the randomized algorithm, however, it is more convenient to activate trivial fans alongside the nontrivial fans and analyze both within a common framework, which yields a uniform success probability bound for every fixed center regardless of its fan type. Moreover, randomizing the missing color of Section~\ref{subsec:rand-miss-color-par} may cause the fan subsequently constructed at a center to be trivial. A trivial fan with center $u$ is colored by a rotation alone, so it has no alternating path of its own; its only possible conflicts are fan-fan conflicts and incoming path-fan conflicts. It may still have fan-fan conflicts with as many as $\Delta^2$ other fans, and our analysis of the previous steps does not give a corresponding reduction for trivial fans. To thin these conflicts, we mimic the effect of the random matching and independently activate each trivial center $u$ by a coin $\zeta_u$ with $\Pr[\zeta_u=1]=1/({2\Delta})$, which is on the same $\Theta(1/\Delta)$ scale as the activation probability of a nontrivial center in Lemma~\ref{lem:M-active-prob}. The coins $\{\zeta_u\}$ are mutually independent and are independent of $M$ and of all previous randomness.

Let $W_{\mathrm{triv}} := \{u\in W:F_u\text{ is trivial}\}$. Let $A_{\mathrm{nontriv}}$ denote the set of $M$-active (nontrivial) centers obtained in the previous subsection, and let $A_{\mathrm{triv}} := \{u\in W_{\mathrm{triv}}:\zeta_u=1\}$ be the activated trivial fans, and $A := A_{\mathrm{nontriv}}\cup A_{\mathrm{triv}}$ be all activated fans. Thus both trivial and nontrivial centers are activated with probability $\Theta(1/\Delta)$, and the sampling and conflict-removal step of the next subsection will be applied to the entire set $A$. However, we first need to show that this independent activation reduces the fan-fan conflicts of a trivial center from $O(\Delta^2)$ to $O(\Delta)$ with constant probability.

For every center $u\in W$ let  $J_u := \bigl| \{w\in A\setminus\{u\}: V(F_w)\cap V(F_u)\neq\varnothing\} \bigr|$ denote the total number of active centers whose fans conflict with $F_u$, and let $T_u := \bigl| \{w\in A_{\mathrm{triv}}\setminus\{u\}:  V(F_w)\cap V(F_u)\neq\varnothing\} \bigr|$ denote the number of such centers that are trivial.

\begin{lemma}\label{lem:trivial-fan-conflicts}
Let $u\in W_{\mathrm{triv}}$ be a fixed \emph{trivial} center. Then  $\Pr[  J_u<5\Delta \mid  u\in A_{\mathrm{triv}}  ]  \ge  9/10$. In other words, conditioned on the trivial center $u$ being active, with probability at least $9/10$, less than $5\Delta$ other active centers have fans that conflict with $F_u$.
\end{lemma}

\begin{proof}
Condition on the entire outcome of randomization of the missing color step and the fan construction, so that all fans and their conflicts are fixed. If $F_w$ conflicts with $F_u$, then $\dist(u,w)\le2$. Therefore, there are at most $\Delta^2$ centers $w\neq u$ whose fans can conflict with $F_u$. Fix one such center $w$. If $F_w$ is trivial, then  $\Pr[w\in A_{\mathrm{triv}}]  = 1/{(2\Delta)}$ by the definition of the coin $\zeta_w$. On the other hand, if $F_w$ is nontrivial, then by Lemma~\ref{lem:M-active-prob}, $\Pr[w\in A_{\mathrm{nontriv}}] \le 1/{(2\Delta)}$. Thus, in either case,  $\Pr[w\in A] \le 1/{(2\Delta)}$. Moreover, conditioning on $u\in A_{\mathrm{triv}}$ does not change any of these probabilities. The event $u\in A_{\mathrm{triv}}$ depends only on $\zeta_u$, whereas $\zeta_u$ is independent of $M$ and of every coin $\zeta_w$ with $w\neq u$. Hence, for every $w\neq u$ whose fan conflicts with $F_u$,  $\Pr[  w\in A \mid u\in A_{\mathrm{triv}} ] \le 1/{(2\Delta)}$. Therefore,
\[
 \mathbb{E}[   J_u  \mid u\in A_{\mathrm{triv}} ] \le \Delta^2\cdot\frac{1}{2\Delta} = \frac{\Delta}{2}.
\]

Applying Markov's inequality,
\[
\begin{aligned}
    \Pr[ J_u\ge 5 \Delta  \mid   u\in A_{\mathrm{triv}}  ]
    &\le \frac{ \mathbb{E}[ J_u \mid  u\in A_{\mathrm{triv}} ]}{5\Delta}\\
    &\le \frac{\Delta/2}{5\Delta}\\
    &= \frac{1}{10}.
\end{aligned}
\]
Taking the complementary event then gives $\Pr[  J_u<5\Delta  \mid u\in A_{\mathrm{triv}} ] \ge 1-1/10 = 9/10$.
\end{proof}

The activation of trivial fans can also introduce additional fan-fan conflicts for an active nontrivial center. The next lemma shows that this changes only the constant in the conflict bound proved in the previous section.

\begin{lemma}\label{lem:nontrivial-trivial-conflicts}
Let $u$ be a fixed \emph{nontrivial} center. Then  $\Pr[  T_u<5\Delta \mid u\in A_{\mathrm{nontriv}} ] \ge 9/10$. Informally, conditioned on the nontrivial center $u$ being active, with probability at least $9/10$, less than $5\Delta$ active trivial centers have fans that conflict with $F_u$.

Consequently, $\Pr[D_u+T_u<28\Delta \mid u\in A_{\mathrm{nontriv}} ] \ge 3/10$, i.e., conditioned on the nontrivial center $u$ being active, with probability at least $3/10$, less than $28\Delta$ active centers have fans that conflict with $F_u$.
\end{lemma}

\begin{proof}
Condition on the entire outcome of the missing-color randomization and the fan construction, so that all fans and their conflicts are fixed. As before, every center whose fan conflicts with $F_u$ lies within distance at most two of $u$. Hence at most $\Delta^2$ trivial centers can have fans that conflict with $F_u$. Every such trivial center $w$ is activated independently with probability $\Pr[w\in A_{\mathrm{triv}}] = 1/{(2\Delta)}$. Furthermore, the coins $\{\zeta_w\}$ used to activate the trivial centers are independent of $M$. Since the event $u\in A_{\mathrm{nontriv}}$ is determined by $M$, conditioning on $u\in A_{\mathrm{nontriv}}$ does not change the activation probability of any trivial center $w$. Therefore, $\Pr[ w\in A_{\mathrm{triv}} \mid  u\in A_{\mathrm{nontriv}} ]  =  1/{(2\Delta)}$. Similarly,
\[
\mathbb{E}[  T_u \mid u\in A_{\mathrm{nontriv}} ] \le \Delta^2\cdot\frac{1}{2\Delta}  =  \frac{\Delta}{2}.
\]

Applying Markov's inequality,
\[
\begin{aligned}
    \Pr[  T_u\ge5\Delta  \mid  u\in A_{\mathrm{nontriv}}  ]
    &\le  \frac{  \mathbb{E}[  T_u \mid  u\in A_{\mathrm{nontriv}}  ] }{5\Delta}\\
    &\le \frac{\Delta/2}{5\Delta}\\
    &=  \frac{1}{10}.
\end{aligned}
\]
Hence $\Pr[ T_u<5\Delta \mid  u\in A_{\mathrm{nontriv}} ] \ge 1-1/10 = 9/10$, which proves the first claim.

\smallskip
We now combine this bound with the conflict bound for active nontrivial centers from the previous subsection. By Theorem~\ref{thm:active-conf-bound}, $\Pr[ D_u<23\Delta \mid  u\in A_{\mathrm{nontriv}}  ] \ge  2/5$, or equivalently, $\Pr[  D_u\ge23\Delta  \mid  u\in A_{\mathrm{nontriv}} ] \le 3/5$. Then also from the first part of this lemma,
\[
\begin{aligned}
    &\Pr[  D_u<23\Delta \text{ and } T_u<5\Delta  \mid  u\in A_{\mathrm{nontriv}} ]\\
    &\qquad=   1-  \Pr[  D_u\ge23\Delta  \text{ or }  T_u\ge5\Delta  \mid  u\in A_{\mathrm{nontriv}} ]\\
    &\qquad \ge 1-\frac{3}{5}-\frac{1}{10}\\
    &\qquad = \frac{3}{10}.
\end{aligned}
\]

Whenever both $D_u<23\Delta$ and $T_u<5\Delta$ hold, $D_u$ counts less than $23\Delta$ active nontrivial centers whose fans conflict with $F_u$, while $T_u$ counts less than $5\Delta$ active trivial centers whose fans conflict with $F_u$. Since $A_{\mathrm{nontriv}}$ and $A_{\mathrm{triv}}$ are disjoint, together these account for all active centers whose fans conflict with $F_u$. Therefore, $D_u+T_u < 23\Delta+5\Delta = 28\Delta$. Since the simultaneous event $D_u<23\Delta$ and $T_u<5\Delta$ occurs with probability at least $3/10$ conditioned on $u\in A_{\mathrm{nontriv}}$, we conclude that $\Pr[ D_u+T_u<28\Delta \mid u\in A_{\mathrm{nontriv}}  ]  \ge 3/10$. This proves the second claim of the lemma.
\end{proof}

We can now combine the trivial and nontrivial cases into a common bound as needed for the next step of the analysis. For an active center $u\in A$, whether its fan is trivial or nontrivial, let $B_u$ denote the number of active centers that can block $u$. Namely, $B_u$ counts the active centers $w\neq u$ for which $F_w$ has a fan-fan conflict with $F_u$ or $P_w$ has a path-fan conflict with $F_u$.

\begin{lemma}\label{lem:active-conflict-degree}
There is a constant $c_1>0$ such that, for every fixed center $u\in W$, whether $F_u$ is trivial or nontrivial,
\[
\Pr[B_u<c_1\Delta\mid u\in A]\ge \frac{3}{10}.
\]
In other words, conditioned on $u$ being active, with probability at least $3/10$, it has fewer than $c_1\Delta$ conflicts that can prevent it from being colored in parallel.
\end{lemma}

\begin{proof}
By Lemma~\ref{lem:path-fan-conflicts}, there is a constant $c_0>0$ such that every center $u$ has at most $c_0\Delta$ active blockers caused by path-fan conflicts oriented from $w$ to $u$. It therefore remains to bound the fan-fan conflicts. Suppose first that $F_u$ is nontrivial. By Lemma~\ref{lem:nontrivial-trivial-conflicts}, conditioned on $u\in A$, the total number of active fan-fan conflicts of $u$ is less than $28\Delta$ with probability at least $3/10$. Hence, in this case, $ B_u < (28+c_0)\Delta$.

Suppose instead that $F_u$ is trivial. By Lemma~\ref{lem:trivial-fan-conflicts}, conditioned on $u\in A$, the number of active fan-fan conflicts of $u$ is less than $5\Delta$ with probability at least $9/10$. Since a trivial fan has no alternating path of its own, its only path-related blockers are incoming path-fan conflicts, which are already included in the $c_0\Delta$ bound. Thus, with probability at least $9/10$, $B_u < (5+c_0)\Delta \le (28+c_0)\Delta$. Setting $ c_1:=28+c_0$ completes the proof.
\end{proof}

Therefore both trivial and nontrivial centers are activated with probability $\Theta(1/\Delta)$, and conditioned on being active, every center has only $O(\Delta)$ conflicts (or blockers) with constant probability. The next step of the algorithm applies a common sampling and conflict-removal procedure to the full active set $A$, to remove the remaining conflicts.

\begin{lemma}\label{lem:trivial-activation-runtime}
Given the classification of the fans and the set $A_{\mathrm{nontriv}}$, the activation of the trivial centers and the construction of $A=A_{\mathrm{nontriv}}\cup A_{\mathrm{triv}}$ can be implemented in $O(m)$ work and $O(\log n)$ span.
\end{lemma}

\begin{proof}
For every $u\in W_{\mathrm{triv}}$, generate the independent coin $\zeta_u$ and mark $u$ active if $\zeta_u=1$. This takes $O(|W_{\mathrm{triv}}|)$ work and $O(1)$ span. The active trivial centers can then be filtered in $O(|W_{\mathrm{triv}}|)$ work and $O(\log n)$ span. Concatenating this list with the already computed set $A_{\mathrm{nontriv}}$ gives $A$ within the same bounds. Since $|W_{\mathrm{triv}}|\le |W|\le m$, the total cost is $O(m)$ work and $O(\log n)$ span.
\end{proof}

\subsection{Random Sampling and Conflict Removal}\label{subsec:random-sampling}

After the previous steps, the active set $A$ contains both trivial and nontrivial centers. Recall, for every active center $u\in A$, $B_u$ denotes the number of active centers whose conflicts can block $u$. By Lemma~\ref{lem:active-conflict-degree}, conditioned on $u$ being active, with probability at least $3/10$ we have $B_u<c_1\Delta$. Thus, with constant probability, an active center has only $O(\Delta)$ conflicts that can prevent it from being processed. Then the algorithm performs one additional independent sampling in which every active center is selected with probability $\Theta(1/\Delta)$. If $B_u<c_1\Delta$, then the expected number of centers whose conflicts can block $u$ and that are also selected is less than $1/2$.

\paragraph{Procedure Random Sampling.}
Independently sample every active center $u\in A$ with probability $p := 1/{(2c_1\Delta)}$, where all sampling decisions are mutually independent and independent of all previous randomness. Let $\mathcal{S}_0 := \{u\in A:u\text{ is sampled}\}$ be the set of sampled active centers.

After the sampling, a final deterministic conflict-removal step extracts a conflict-free subset of the sampled centers. It is worth contrasting this with the deterministic algorithm of Section~\ref{subsec:deterministic}. There, we need a worst-case guarantee on the size of the conflict-free subset, and merely orienting each conflict and deleting one endpoint gives no lower bound on how many centers remain: in the worst case, almost every center is the target of some oriented conflict and is deleted. We therefore used the large-independent-set subroutine (Lemma~\ref{lem:large-independent-set}) in that setting. However, here we need something weaker, namely a lower bound on the probability that each fixed center survives. The independent random sampling procedure provides exactly what is needed here: whenever $B_u<c_1\Delta$, with constant probability none of the centers whose conflicts are oriented toward $u$ is sampled, so $u$ survives (Lemma~\ref{lem:sampled-survival}). A simple deterministic orient-and-delete step therefore suffices.

\paragraph{Procedure Conflict Removal.}
Construct all conflicts among the centers of $\mathcal{S}_0$. For a fan-fan conflict between two sampled centers $u$ and $w$, orient the conflict according to a fixed deterministic rule (e.g. from the center with smaller identifier to the center with larger identifier). For a path-fan conflict, keep the orientation defined in Lemma~\ref{lem:path-fan-conflicts}. Thus every conflict among sampled centers is oriented toward one of its endpoints. For every sampled center $u\in\mathcal{S}_0$, check whether there exists another sampled center $w\in\mathcal{S}_0$ such that the conflict between $w$ and $u$ is oriented from $w$ to $u$. If such a center $w$ exists, remove $u$; otherwise, $u$ survives. Let $\mathcal{S}$ be the set of surviving centers. As shown later, $\mathcal{S}$ is conflict-free. 
The following lemma shows that a sampled center satisfying the bound $B_u<c_1\Delta$ survives the Conflict Removal procedure with constant probability.

\begin{lemma}\label{lem:sampled-survival}
For every fixed center $u\in W$, $\Pr[ u\in\mathcal{S} \mid  u\in\mathcal{S}_0,\,  B_u<c_1\Delta ] \ge 1/2$.
\end{lemma}
\begin{proof}
Condition on the entire outcome of all randomness preceding the Random Sampling procedure, so that the active set $A$, all conflicts, and the values $B_u$ are fixed. Suppose that $B_u<c_1\Delta$. Among the active centers $w\neq u$, only those whose conflict with $u$ is oriented toward $u$ can cause $u$ to be removed by the Conflict Removal procedure, and every such center is counted by $B_u$. Thus fewer than $c_1\Delta$ active centers can cause $u$ to be removed. Now condition further on $u\in\mathcal{S}_0$. Since the sampling decisions are mutually independent, conditioning on $u$ being sampled does not affect the sampling decision of any other active center. Hence every center that can cause $u$ to be removed is still sampled with probability $p=1/(2c_1\Delta)$. Conditioned on $u\in\mathcal{S}_0$, the event that $u$ is removed is exactly the event that at least one blocker of $u$ is also sampled. Therefore, by the union bound,
\[
\begin{aligned}
\Pr[u\text{ is removed}\mid u\in\mathcal{S}_0,\,B_u<c_1\Delta]
&\le B_u \,p\\
&<c_1\Delta\cdot\frac{1}{2c_1\Delta}=\frac{1}{2}.
\end{aligned}
\]
Therefore, $\Pr[u\in\mathcal{S}\mid u\in\mathcal{S}_0,\,B_u<c_1\Delta]\ge\frac{1}{2}$.
\end{proof}

We next show that both procedures can be implemented efficiently in parallel. The condition $B_u<c_1\Delta$ is used only in the probability analysis; neither procedure computes $B_u$ or determines which centers satisfy this condition.

\begin{lemma}\label{lem:random-sampling-runtime}
For every fixed realization of $\mathcal{S}_0$, the required bichromatic components can be identified and the Conflict Removal procedure can be implemented in $O(m \log n+\Delta^2|\mathcal{S}_0|\log n)$ work and $O(\log n)$ span. Consequently, Random Sampling, identification of the required bichromatic components, and Conflict Removal can together be implemented in $O(m\log n)$ expected work and $O(\log n)$ span.
\end{lemma}

\begin{proof}
We first consider Random Sampling. Sampling every center of $A$ independently takes $O(|A|)$ work and $O(1)$ span, and compacting the sampled centers into the array $\mathcal{S}_0$ takes $O(|A|)$ work and $O(\log n)$ span. Since $|A|\le |W|\le m$, this requires $O(m)$ work and $O(\log n)$ span. Fix a realization of $\mathcal{S}_0$. Every nontrivial center in $\mathcal{S}_0$ is $M$-active, so its underlying unordered required pair belongs to $\mathcal{M}$. Since the pairs of $\mathcal{M}$ are pairwise color-disjoint, Lemma~\ref{lem:parallel-pairs} identifies and labels all required bichromatic components in $O(m\log n)$ work and $O(\log n)$ span (trivial centers have no alternating path and therefore require no bichromatic component). Once the required bichromatic components are labeled, Lemma~\ref{lem:build-conf-graph} finds all conflicts among the centers of $\mathcal{S}_0$ and produces $O(\Delta^2|\mathcal{S}_0|)$ conflict records in $O(m \log n+ \Delta^2|\mathcal{S}_0|\log n)$ work and $O(\log n)$ span. From every conflict record, we generate in $O(1)$ work a removal record containing the center toward which the conflict is oriented. This produces $O(\Delta^2|\mathcal{S}_0|)$ removal records. Sorting these records by center identifier, marking every center that occurs, and filtering the unmarked centers of $\mathcal{S}_0$ to obtain $\mathcal{S}$ takes $O(\Delta^2|\mathcal{S}_0|\log n)$ work and $O(\log n)$ span. Hence, for every fixed realization of $\mathcal{S}_0$, once the required bichromatic components have been identified and labeled, Conflict Removal takes $O(\Delta^2|\mathcal{S}_0|\log n)$ work and $O(\log n)$ span.

It remains to bound the expected work dependent on the sampled set. Fix a center $u\in W$ and condition on the outcome of Randomize Missing Colors and the subsequent fan construction and classification, so that $F_u$ and whether it is trivial or nontrivial are fixed. Then $u$ is active with probability at most $1/(2\Delta)$. Thus, under every such conditioning, $\Pr[u\in A] \le 1/(2\Delta)$, and therefore the same bound holds unconditionally. Random Sampling selects every active center independently with probability $p=1/(2c_1\Delta)$, independently of all previous randomness. Hence
\[
\begin{aligned}
   \Pr[u\in\mathcal{S}_0]
   &=\Pr[u\in A]\cdot\frac{1}{2c_1\Delta}\\
   &\le\frac{1}{2\Delta}\cdot\frac{1}{2c_1\Delta}\\
   &=\frac{1}{4c_1\Delta^2}.
\end{aligned}
\]
By linearity of expectation,
\[
\begin{aligned}
   \mathbb{E}[|\mathcal{S}_0|]
   &=\sum_{u\in W}\Pr[u\in\mathcal{S}_0]\\
   &\le\frac{|W|}{4c_1\Delta^2}\\
   &\le\frac{m}{4c_1\Delta^2}.
\end{aligned}
\]

Therefore, the expected work dependent on the sampled set of Conflict Removal for some constants $\hat{c}$ and $\hat{c}'$ is
\[
\begin{aligned}
  \mathbb{E}\!\left[ \hat{c}\, m \log n+ \hat{c}' \, \Delta^2|\mathcal{S}_0|\log n\right]
   &=\hat{c}\,m \log n+ \hat{c}'\,\Delta^2\log n\cdot\mathbb{E}[|\mathcal{S}_0|]\\
  &=O(m\log n).
\end{aligned}
\]
Adding the $O(m)$ work of Random Sampling and the $O(m\log n)$ work for bichromatic-component identification gives $O(m\log n)$ expected work in total. Each stage has $O(\log n)$ span and there are only constant many stages, so the total span is $O(\log n)$.
\end{proof}

We now show that $\mathcal{S}$ has the two properties required for the round: it is conflict-free, so its Vizing operations can run in parallel, and every fixed center is in it with probability $\Theta(1/\Delta^2)$. Note that both hold uniformly for trivial and nontrivial fans.

\begin{theorem}\label{thm:conflict-free-set}
The set $\mathcal{S}$ is conflict-free: no two distinct centers in $\mathcal{S}$ have a fan-fan or path-fan conflict. Moreover, for every fixed center $u\in W$,
\[
    \Pr[u \in \mathcal{S}] = \Omega \! \left(\frac{1}{\Delta^2}\right).
\]
Consequently,
\[
    \mathbb{E}[|\mathcal{S}|] = \Omega \! \left(\frac{|W|}{\Delta^2} \right).
\]
\end{theorem}
\begin{proof}
By construction, $\mathcal{S}$ is conflict-free; every conflict between two centers of $\mathcal{S}_0$ is oriented toward one of its endpoints, and Conflict Removal deletes that endpoint. Hence at least one endpoint of every fan-fan or path-fan conflict is removed from $\mathcal{S}_0$. Therefore no two distinct centers that remain in $\mathcal{S}$ conflict with each other.

We next bound the probability that a fixed center $u\in W$ belongs to $\mathcal{S}$. Regardless of whether $F_u$ is trivial or nontrivial, we have
\[
\begin{aligned}
    \Pr[u\in A]
    &\ge \frac{1}{2(\Delta+1)}
    && \text{(Definition of $A_{\mathrm{triv}}$ and}\\[-2pt]
    &&& \text{\phantom{(}Lemma~\ref{lem:M-active-prob})},\\
    \Pr[B_u<c_1\Delta \mid u\in A]
    &\ge \frac{3}{10}
    && \text{(Lemma~\ref{lem:active-conflict-degree})},\\
    \Pr[u\in\mathcal{S}_0 \mid u\in A,\, B_u<c_1\Delta]
    &= \frac{1}{2c_1\Delta}
    && \text{(Random Sampling)},\\
    \Pr[u\in\mathcal{S} \mid u\in\mathcal{S}_0,\, B_u<c_1\Delta]
    &\ge \frac{1}{2}
    && \text{(Lemma~\ref{lem:sampled-survival})}.
\end{aligned}
\]

To (lower) bound the probability that $u$ belongs to $\mathcal{S}$, we restrict attention to the case in which $u$ has less than $c_1\Delta$ active blockers ($B_u$). Moreover, every center that survives Conflict Removal must first have been sampled, and every sampled center must first have been active. Therefore, we can apply the chain rule to these successive events and obtain,
\[
\begin{aligned}
    \Pr[u\in\mathcal{S}]
    &\ge \Pr[  u\in A,\,  B_u<c_1\Delta,\, u\in\mathcal{S}_0,\,  u\in\mathcal{S} ]\\
    &= \Pr[u\in A]\, \Pr[B_u<c_1\Delta\mid u\in A]\\
    &\qquad\cdot \Pr[  u\in\mathcal{S}_0  \mid  u\in A,\,  B_u<c_1\Delta ]\\
    &\qquad\cdot \Pr[  u\in\mathcal{S}  \mid  u\in\mathcal{S}_0,\,  B_u<c_1\Delta ]\\
    &\ge \frac{1}{2(\Delta+1)} \cdot \frac{3}{10} \cdot \frac{1}{2c_1\Delta} \cdot \frac{1}{2}\\
    &= \frac{3}{80c_1\Delta(\Delta+1)} \,\ge\, \frac{3}{160c_1\Delta^2} = \Omega\!\left(\frac{1}{\Delta^2}\right),
\end{aligned}
\]
where the last inequality simply uses $\Delta+1\le2\Delta$. Finally, by linearity of expectation,
\[
\begin{aligned}
    \mathbb{E}[|\mathcal{S}|] &= \sum_{u\in W}\Pr[u\in\mathcal{S}]\\
    &\ge |W|\cdot\frac{3}{160c_1\Delta^2}\\
    &= \Omega\!\left(\frac{|W|}{\Delta^2}\right).
\end{aligned}
\]
\end{proof}

\subsection{Complete Algorithm and its Analysis}\label{subsec:complete-rand}
We now describe the complete randomized algorithm, given in Algorithm~\ref{alg:rand-color-match}. First, each uncolored edge is assigned a center and a designated missing color at its center. The algorithm then repeats the following until all uncolored edges are colored (lines~4--16).

At the beginning of a round, let $W_0$ be the set of centers of the currently uncolored edges. In line~5, the algorithm constructs the power-of-two subpalette $\Gamma$ and the corresponding subset $W\subseteq W_0$ by Lemma~\ref{lem:power-two-subpalette}, $|W|>|W_0|/2$ where the designated missing color of every center in $W$ belongs to $\Gamma$. The centers in $W_0\setminus W$ are not processed in the current round which are only constant fraction of the centers. The Randomize Missing Colors procedure is then applied to the centers in $W$ (line~6), so that the final designated missing color of every $u\in W$ is uniform on $\Gamma$ while the coloring remains proper and the set of uncolored edges is unchanged.

Next, the fan of every center in $W$ is constructed and classified as trivial or nontrivial, and the required pair of every nontrivial fan is determined (lines~7--8). The algorithm then samples a random ordered matching $M$ of the full palette $[\Delta+1]$ using the Random Matching procedure (line~9). A nontrivial center is activated if its ordered required pair belongs to $M$, while each trivial center is activated independently with probability $1/(2\Delta)$ (line~10). Let $A$ denote the resulting set of active centers.

In line~11, every center in $A$ is independently sampled with probability $p=1/(2c_1\Delta)$, where $c_1$ is the constant from Lemma~\ref{lem:active-conflict-degree}, and let $\mathcal{S}_0$ be the set of sampled centers. For the nontrivial centers in $\mathcal{S}_0$, the required bichromatic components are identified and labeled in parallel using Lemma~\ref{lem:parallel-pairs} (line~12). Since all their required pairs belong to the random matching $M$, their distinct required pairs are pairwise color-disjoint. The Conflict Removal procedure is then applied to $\mathcal{S}_0$ to obtain a conflict-free subset $\mathcal{S}$ (line~13). Finally, all Vizing operations corresponding to the centers in $\mathcal{S}$ are performed in parallel using Lemma~\ref{lem:parallel-vizing-operations} (line~14), so every corresponding u-edge becomes colored while properness is preserved. The next round is applied to the remaining uncolored edges.

\begin{algorithm}[t]
\caption{Randomized color-matching algorithm}\label{alg:rand-color-match}
\begin{algorithmic}[1]
\STATE \textbf{Input:} a proper partial $(\Delta+1)$-edge coloring $\chi$ with uncolored edges $U_\chi$ that form a matching
\STATE \textbf{Output:} a proper $(\Delta+1)$-edge coloring $\chi$ in which all edges of $U_\chi$ are colored
\STATE assign each $e\in U_\chi$ a center $\cen(e)$ and a designated missing color $\varphi(\cen(e))$
\WHILE{$U_\chi\neq\varnothing$}
    \STATE let $W_0:=\{\cen(e):e\in U_\chi\}$, and construct a power-of-two subpalette $\Gamma$ and $W\subseteq W_0$ with $|W|>|W_0|/2$ (Lemma~\ref{lem:power-two-subpalette})
    \STATE run \textsc{Randomize Missing Colors} on the centers of $W$ and the subpalette $\Gamma$ (Section~\ref{subsec:rand-miss-color-par})
    \STATE construct the fan $F_u$ for every $u\in W$ in parallel (Lemma~\ref{lem:parallel-fans})
    \STATE classify every $F_u$ as trivial or nontrivial, and compute the required pair $\pi_{e_u}$ of every nontrivial fan
    \STATE sample a random ordered matching $M$ of the full palette $[\Delta+1]$ (Section~\ref{subsec:random-matching})
    \STATE activate every nontrivial $u$ with $\pi_{e_u}\in M$, and independently activate every trivial $u$ with probability $1/(2\Delta)$; let $A$ be the set of active centers
    \STATE independently sample every $u\in A$ with probability $p:=1/(2c_1\Delta)$, and let $\mathcal{S}_0$ be the set of sampled centers
    \STATE identify and label the required bichromatic components of the nontrivial centers in $\mathcal{S}_0$ in parallel (Lemma~\ref{lem:parallel-pairs})
    \STATE run \textsc{Conflict Removal} on $\mathcal{S}_0$ to obtain the conflict-free set $\mathcal{S}$ (Section~\ref{subsec:random-sampling})
    \STATE perform the Vizing operations corresponding to all centers in $\mathcal{S}$ and color them  in parallel (Lemma~\ref{lem:parallel-vizing-operations})
    \STATE update $U_\chi$ to the remaining uncolored edges
\ENDWHILE
\STATE \textbf{return} $\chi$
\end{algorithmic}
\end{algorithm}

\begin{restatedRandTheoremP}
Given a proper partial $(\Delta+1)$-edge coloring whose uncolored edges form a matching, Algorithm~\ref{alg:rand-color-match} colors all of them in $O(m\Delta^2\log^2 n\log\Delta)$ work and $O(\Delta^2\log^2 n\log\Delta)$ span, both with high probability.
\end{restatedRandTheoremP}

\begin{proof}
We wish to prove correctness, bound the number of rounds with high probability, and then bound the total work and span. Assume $\Delta\ge3$ (other wise. Let $\lambda_t$ denote the number of uncolored edges at the beginning of round $t$, and let $W_{0,t}$ be the set of their centers. Since every uncolored edge has exactly one center, $|W_{0,t}|=\lambda_t$. Let $\mathcal{H}_t$ be the $\sigma$-field generated by all random choices made before round $t$. Thus the current partial coloring, the current set of uncolored edges, their centers, their designated missing colors, and all other data maintained by the algorithm at the beginning of round $t$ are $\mathcal{H}_t$-measurable. Moreover, $\mathcal{H}_{t-1}\subseteq\mathcal{H}_t$ for every $t \ge 1$.

\paragraph{Bounding the number of rounds.}We first bound the progress made in one round. Let $W_t\subseteq W_{0,t}$ be the subset kept by the power-of-two subpalette construction in round $t$. Conditioned on $\mathcal{H}_t$, the state at the beginning of round $t$ is fixed. Since the power-of-two subpalette construction is deterministic given this state, the resulting subset $W_t$ is also fixed. By Lemma~\ref{lem:power-two-subpalette}, $|W_t|>|W_{0,t}|/2 =\lambda_t/2$. Let $\mathcal{S}_t$ be the conflict-free set produced at the end of round $t$. By Theorem~\ref{thm:conflict-free-set}, there is an absolute constant $\eta>0$ such that, for every fixed center $u\in W_t$,
\[
    \Pr[u \in \mathcal{S}_t \mid \mathcal{H}_t] \ge \frac{\eta}{\Delta^2}.
\]
Conditioning on $\mathcal{H}_t$, the state at the beginning of round $t$ is fixed, while all random choices used during that round are fresh, so the analysis of Theorem~\ref{thm:conflict-free-set} carries over. With the constants used there, $\eta=3/(160c_1)$. By linearity of expectation,
\[
\begin{aligned}
    \mathbb{E} \bigl[ |\mathcal{S}_t| \mid \mathcal{H}_t \bigr]
    & = \sum_{u\in W_t} \Pr[u\in \mathcal{S}_t \mid \mathcal{H}_t]\\
    &\ge |W_t| \frac{\eta}{\Delta^2}\\
    & > \frac{\eta\, \lambda_t}{2\Delta^2}.
\end{aligned}
\]
Every center in $\mathcal{S}_t$ corresponds to a distinct currently uncolored edge, and Lemma~\ref{lem:parallel-vizing-operations} guarantees that all these edges become colored while no new uncolored edge is created. Therefore, $\lambda_{t+1} =\lambda_t-|\mathcal{S}_t|$. Taking conditional expectations,
\[
\begin{aligned}
    \mathbb{E}[\lambda_{t+1}\mid\mathcal{H}_t]
    &=\mathbb{E}[\lambda_{t}\mid\mathcal{H}_t]-\mathbb{E}\bigl[|\mathcal{S}_t|\mid\mathcal{H}_t\bigr]\\
    &=\lambda_t-\mathbb{E}\bigl[|\mathcal{S}_t|\mid\mathcal{H}_t\bigr]\\
    &\le \lambda_t-\frac{\eta}{2\Delta^2}\lambda_t\\
    &=\left(1-\frac{\eta}{2\Delta^2}\right)\lambda_t.
\end{aligned}
\]
Taking expectations of both sides again,
\[
\begin{aligned}
    \mathbb{E}[\lambda_{t+1}]
    &=\mathbb{E}\!\left[ \mathbb{E}[\lambda_{t+1}\mid\mathcal{H}_t] \right]\\
    &\le \mathbb{E}\!\left[ \left(1-\frac{\eta}{2\Delta^2}\right)\lambda_t \right]\\
    &= \left(1-\frac{\eta}{2\Delta^2}\right) \mathbb{E}[\lambda_t].
\end{aligned}
\]
Iterating this inequality and using $1-x\le e^{-x}$ and $\lambda_0 \le m$ gives
\[
\begin{aligned}
    \mathbb{E}[\lambda_t]
    &\le \lambda_0 \left(1-\frac{\eta}{2\Delta^2}\right)^t\\
    &\le \lambda_0 \exp\!\left(-\frac{\eta t}{2\Delta^2}\right)\\
    &\le m \exp\!\left(-\frac{\eta t}{2\Delta^2}\right).
\end{aligned}
\]
Since $\lambda_t$ is a nonnegative integer, applying Markov's inequality gives
\[
\begin{aligned}
    \Pr[\lambda_t>0]
    &=\Pr[\lambda_t\ge1]\\
    &\le \mathbb{E}[\lambda_t]\\
    &\le m \exp\!\left(-\frac{\eta t}{2\Delta^2}\right).
\end{aligned}
\]

Fix an arbitrary constant $a>0$. We show that the algorithm terminates within $O(\Delta^2\log n)$ rounds with probability at least $1-n^{-(a+1)}$. Let  $r:=\left\lceil c\Delta^2\log n\right\rceil$, where $c>0$ is a sufficiently large constant depending only on $a$. Using the bound above and the fact that $m < n^2$ for a simple graph, we obtain
\[
\begin{aligned}
    \Pr[\lambda_r>0]
    &\le m \exp \!\left(-\frac{\eta \,r}{2\Delta^2}\right)\\
    &\le n^2 \exp \!\left(-\frac{\eta\, c\log n}{2}\right)\\
    &= n^{\,2-\eta c/2}.
\end{aligned}
\]
Choose $c$ such that the exponent is at most $-(a+1)$, i.e.,  for a $c\ge 2(a+3)/\eta$, we have $\Pr[\lambda_r>0]\le n^{-(a+1)}$. Therefore, with probability at least $1-n^{-(a+1)}$, no uncolored edge remains after $r=O(\Delta^2\log n)$ rounds.

\paragraph{Bounding the work.} We next obtain a high-probability bound on the work of Conflict Removal. Lemma~\ref{lem:random-sampling-runtime} gives only an expected-work bound, while all other work in a round is deterministically bounded. It therefore remains to bound the total work dependent on the sampled-set $\mathcal{S}_{0,t}$ of Conflict Removal over all rounds. If the algorithm terminates before round $r$, we define all sets associated with subsequent rounds to be empty. Let $A_t$ be the set of active centers in round $t$, and let $\mathcal{S}_{0,t}$ be the set obtained from $A_t$ by Random Sampling. Fix $u\in W_t$. Condition temporarily on the entire outcome of the Randomize Missing Colors procedure and the subsequent fan construction and classification in round $t$. Under every such conditioning, the fan of $u$ and its classification are fixed. A fan is activated with probability at most $1/(2\Delta)$. Therefore,  $\Pr[u\in A_t\mid\mathcal{H}_t] \le 1/{(2\Delta)}$. By linearity of expectation,
\[
\begin{aligned}
   \mathbb{E}[|A_t|\mid\mathcal{H}_t]
   &=\sum_{u\in W_t}\Pr[u\in A_t\mid\mathcal{H}_t]\\
   &\le\frac{|W_t|}{2\Delta}\\
   &\le\frac{m}{2\Delta}.
\end{aligned}
\]

After $A_t$ is determined, Random Sampling selects every center of $A_t$ independently with probability $p=1/(2c_1\Delta)$, using fresh randomness independent of all previous choices. Hence, conditioning first on $A_t$,
\[
\begin{aligned}
   \mathbb{E}[|\mathcal{S}_{0,t}|\mid\mathcal{H}_t]
   &=\mathbb{E}\!\left[ \mathbb{E}[|\mathcal{S}_{0,t}|\mid A_t,\mathcal{H}_t] \mid\mathcal{H}_t \right]\\
   &=\mathbb{E}[p|A_t|\mid\mathcal{H}_t]\\
   &=p\,\mathbb{E}[|A_t|\mid\mathcal{H}_t]\\
   &\le\frac{m}{4c_1\Delta^2}.
\end{aligned}
\]

Let $X_t:= |\mathcal{S}_{0,t}| / m$ and $q:= 1/(4c_1\Delta^2)$. Then $0\le X_t\le1$ and $\mathbb{E}[X_t\mid\mathcal{H}_t]\le q$. Moreover, $X_0,\ldots,X_{t-1}$ are $\mathcal{H}_t$-measurable, since they are completely determined by the outcomes of the preceding rounds. Therefore, 
\[
\begin{aligned}
  \mathbb{E}[X_t\mid X_0,\ldots,X_{t-1}]
  &=\mathbb{E}\!\left[ \mathbb{E}[X_t\mid\mathcal{H}_t] \mid X_0,\ldots,X_{t-1} \right]\\
  &\le q.
\end{aligned}
\]

We can therefore apply Lemma~\ref{lem:conditional-chernoff} to $X_0,\ldots,X_{r-1}$ (after shifting the indices by one)  with range bound $1$, $a_i=q$ for every $i$, and $\mu=qr$. Setting $\delta=1$ gives
\[
  \Pr\!\left[ \sum_{t=0}^{r-1}X_t\ge2qr \right]
  \le \exp\!\left(-\frac{qr}{3}\right).
\]
 Since $r=\lceil c\Delta^2\log n\rceil$,
\[
qr =\frac{r}{4c_1\Delta^2} \ge\frac{c\log n}{4c_1},
\]
and hence
\[
\begin{aligned}
  \Pr\!\left[ \sum_{t=0}^{r-1}X_t\ge2qr \right]
  &\le \exp\!\left(-\frac{c\log n}{12c_1}\right)\\
   &= n^{-c/(12c_1)}.
\end{aligned}
\]
Thus, if $c\ge12c_1(a+1)$, then
\[
  \Pr\!\left[ \sum_{t=0}^{r-1}X_t\ge2qr \right]
  \le n^{-(a+1)}.
\]
Since $qr= r/(4c_1\Delta^2) =O(\log n)$, it follows that with probability at least $1-n^{-(a+1)}$,
\[
\begin{aligned}
  \sum_{t=0}^{r-1}|\mathcal{S}_{0,t}|
  &=m\sum_{t=0}^{r-1}X_t\\
  &<2mqr\\
  &=O(m\log n).
\end{aligned}
\]
Let us now bound the total work. Fix a round $t$ and a realization of $\mathcal{S}_{0,t}$, their underlying unordered pairs belong to $\mathcal{M}$ and are pairwise color-disjoint. Hence, by Lemma~\ref{lem:parallel-pairs}, all required bichromatic components can be identified and labeled in $O(m\log n)$ work and $O(\log n)$ span. Since this cost is independent of $|\mathcal{S}_{0,t}|$, we include it in the fixed deterministic per-round costs. By Lemma~\ref{lem:random-sampling-runtime}, once these components have been identified, the work depending on the sampled set for Conflict Removal is $O(m\log n+ \Delta^2|\mathcal{S}_{0,t}|\log n)$ per round, the $O(m \log n)$ term is independent of the $|\mathcal{S}_{0,t}|$ and is included among the fixed per-round costs. The remaining term $O(\Delta^2|\mathcal{S}_{0,t}|\log n)$ work, over all rounds this costs, 
\[
\begin{aligned}
  \sum_{t=0}^{r-1}O(\Delta^2|\mathcal{S}_{0,t}|\log n)
  &= O\!\left( \Delta^2\log n \sum_{t=0}^{r-1}|\mathcal{S}_{0,t}| \right)\\
  &= O(m\Delta^2\log^2 n),
\end{aligned}
\]

The remaining work of each round is $O(m\log n\log\Delta)$. Randomize Missing Colors, including the construction of the power-of-two subpalette, takes $O(m\log n\log\Delta)$ work by Lemma~\ref{lem:randomize-missing-colors-runtime}; constructing and classifying the fans and determining their required pairs takes $O(m\log\Delta)$ work by Lemma~\ref{lem:parallel-fans}; Random Matching takes $O(m\log n)$ work by Lemma~\ref{lem:random-matching-runtime}; activation and sampling take $O(m)$ work; bichromatic-component identification takes $O(m\log n)$ work by Lemma~\ref{lem:parallel-pairs}; the final Vizing operations take $O(m)$ work by Lemma~\ref{lem:parallel-vizing-operations}; and updating the remaining uncolored edges takes $O(m)$ work. Since $r=O(\Delta^2\log n)$, these fixed per-round costs contribute $O(m\Delta^2\log^2 n\log\Delta)$ work in total. Finally combining both costs, with probability at least $1-n^{-(a+1)}$, the total work is $O(m\Delta^2\log^2 n\log\Delta)$.

%%%
\paragraph{Bounding the span.}By Lemma~\ref{lem:randomize-missing-colors-runtime}, constructing the power-of-two subpalette and running Randomize Missing Colors takes $O(\log n\log \Delta)$ span. All remaining steps have $O(\log n)$ span: fan construction and classification take $O(\log\Delta)$ span, Random Matching takes $O(\log n)$ span, activation and sampling take $O(\log n)$ span, bichromatic-component identification and Conflict Removal take $O(\log n)$ span, the parallel Vizing operations take $O(1)$ span by Lemma~\ref{lem:parallel-vizing-operations}, and updating the remaining uncolored edges takes $O(1)$ span. Hence each round has span $O(\log n\log\Delta)$. On the event that the algorithm terminates within  $r=O(\Delta^2\log n)$ rounds, its total span is therefore $O( \Delta^2\log^2 n\log \Delta)$

It remains to bound the probability that either the termination bound or the work bound fails. Recall that $r=\left\lceil c\Delta^2\log n\right\rceil$. Choose $c$ sufficiently large that
\[
    c\ge \max\left\{ \frac{2(a+3)}{\eta}, 12c_1(a+1) \right\}.
\]
Then by the above analysis, the probability that the algorithm has not terminated by round $r$ is at most $n^{-(a+1)}$ and the probability $\sum_{t=0}^{r-1}|\mathcal{S}_{0,t}|$ exceeds the $O(m\log n)$ bound is also at most $n^{-(a+1)}$. Therefore, the probability that either fails is at most $n^{-(a+1)}+n^{-(a+1)} \le n^{-a}$. Thus, with probability at least $1-n^{-a}$, the algorithm terminates within $r=O(\Delta^2\log n)$ rounds, for which the claimed work and span bounds hold with high probability.

\paragraph{Correctness.} At the beginning of every round, the uncolored edges form a matching. The Randomize Missing Colors procedure preserves properness and leaves the set of uncolored edges unchanged. Fan construction, activation, sampling, bichromatic component identification, and Conflict Removal do not modify the coloring. By Theorem~\ref{thm:conflict-free-set}, the surviving set $\mathcal{S}_t$ is conflict-free. Since the required bichromatic components have been identified, Lemma~\ref{lem:parallel-vizing-operations} implies that all Vizing operations corresponding to $\mathcal{S}_t$ can be performed in parallel while preserving properness and coloring its u-edges. No new uncolored edge is created. Hence, after each round, the remaining uncolored edges remain a matching. When the algorithm terminates, no uncolored edge remains, so the resulting coloring is a proper $(\Delta+1)$-edge coloring.

\end{proof}

\subsection{Solving Problem~\ref{prob:delta-plus-one}}\label{subsec:solve_problem_1}

We solve Problem~\ref{prob:delta-plus-one} by plugging our color-matching algorithms into the parallel edge-coloring framework of Elkin and Khuzman~\cite{elkin2026efficient}. We first summarize their framework and then replace its color-reduction subroutine with our color-matching algorithms. Given a graph $G$ with maximum degree $\Delta$, their Procedure \textsc{Edge-Coloring} first partitions $E(G)$ into $\lceil\Delta/2\rceil$ edge-disjoint subgraphs of maximum degree at most two. To obtain this partition, their algorithm first augments $G$ so that it is Eulerian, computes an Eulerian cycle~\cite{atallah1984finding}, constructs the corresponding bipartite in/out graph, and then computes a $\lceil\Delta/2\rceil$-edge coloring of this bipartite graph~\cite{lev1981fast}. The color classes of this coloring define the degree-two subgraphs. This partition can be computed in $O(m\log^2 n\log\Delta)$ work and $O(\log^2 n\log\Delta)$ span after ignoring isolated vertices. Each (at most) degree-two subgraph is then colored with at most three colors using their result for graphs of maximum degree at most two. Let $h=\left \lceil \log \left \lceil \Delta/2 \right \rceil \right \rceil$ and $p=2^h$, where empty subgraphs are added so that there are exactly $p$ initial subgraphs. These subgraphs are merged pairwise through $h$ levels. At every merge, the two child colorings are first combined using disjoint palettes, producing one extra color relative to the palette bound maintained for the parent, and one invocation of their Procedure \textsc{Reduce-Color} removes this extra color. After all $h$ merge levels are completed, the resulting coloring uses either $\Delta+1$ or $\Delta+2$ colors; in the latter case, one additional invocation of Procedure \textsc{Reduce-Color} is needed to yield a proper $(\Delta+1)$-edge coloring.

There are two minor differences between their Procedure \textsc{Reduce-Color} and our color-matching algorithms. First, Procedure \textsc{Reduce-Color} begins with a complete coloring and uncolors one color class, whereas our algorithms assume that the uncolored matching is already given. We therefore first uncolor the edges of one color class, choosing a color class with the fewest edges for our deterministic algorithm. These edges form a matching, which is then passed to our color-matching algorithm. This uncoloring step costs $O(m_H)$ work and $O(\log n)$ span on a graph $H$ with $m_H$ edges, and is dominated by the color-matching call. Second, our color-matching algorithms were stated for a $(\Delta+1)$-color palette, whereas Procedure \textsc{Reduce-Color} may operate with a (slightly) larger palette. The following observation shows that this does not affect our asymptotic bounds.

\begin{observation}\label{obs:larger-palette}
Let $H$ be a graph with $m_H$ edges and maximum degree $\Delta_H$, and suppose that the available palette is $[D]$, where $D\ge\Delta_H+1$. Suppose further that $\Delta_H=O(\hat{\Delta})$ and $D=O(\hat{\Delta})$. Then the randomized and deterministic color-matching algorithms remain valid with the same asymptotic work and span bounds, with $\hat{\Delta}$ instead of $\Delta$. More precisely, the parts of the algorithms and their analyses that depend on the structure of the input graph, such as fan sizes and neighborhood sizes depend on $\Delta_H$, while the parts that depend on the palette size depend on $D$. Since both are $O(\hat{\Delta})$, all such terms are bounded by the corresponding expressions in $\hat{\Delta}$. In particular, a fixed ordered pair is selected by the random matching with probability $\Theta(1/D)=\Omega(1/\hat{\Delta})$. For the deterministic algorithm, it suffices that the uncolored matching has size $O(m_H/\hat{\Delta})$.
\end{observation}

At level $t\in\{0,\ldots,h-1\}$, their framework guarantees that the current graph $H$ has maximum degree at most $2^{t+2}$ and that, before the reduction, its coloring uses at most $2^{t+2}+2$ colors. Consider an invocation of Procedure \textsc{Reduce-Color} on a proper $k$-edge coloring of $H$. After uncoloring one color class, the remaining palette has size $ D=k-1\le 2^{t+2}+1$, and the framework guarantees that $D\ge\Delta_H+1$. We apply Observation~\ref{obs:larger-palette} with $\hat{\Delta}=D$. Since $\Delta_H<D=\hat{\Delta}$ and $D=\hat{\Delta}=O(2^t)$, the randomized and deterministic color-matching algorithms have the same asymptotic bounds with $\hat{\Delta}=O(2^t)$ replacing $\Delta$.

For the deterministic algorithm, we uncolor a color class of minimum cardinality. Since the coloring uses $k=D+1$ colors, this class contains at most $m_H/k = O(m_H/\hat{\Delta})$ edges, as required. We first analyze the final coloring algorithm for Problem~\ref{prob:delta-plus-one} when Algorithm~\ref{alg:rand-color-match} is used. The deterministic analysis is analogous and is given briefly afterwards.

\begin{theorem}\label{thm:rand-delta-plus-one}
There is a randomized parallel algorithm that, given an $n$ vertex, $m$ edge simple graph $G$ of maximum degree $\Delta$, outputs a proper $(\Delta+1)$-edge coloring in $O(m\Delta^2\log^2 n\log\Delta)$ work and $O(\Delta^2\log^2 n\log\Delta)$ span, both with high probability on ARBITRARY CRCW PRAM.
\end{theorem}

\begin{proof}
We use Procedure \textsc{Edge-Coloring} of Elkin and Khuzman~\cite{elkin2026efficient}, replacing every invocation of Procedure \textsc{Reduce-Color} by Algorithm~\ref{alg:rand-color-match} with some modifications described above. For $\Delta\le2$, we use their parallel algorithm for graphs of maximum degree at most two, so assume $\Delta\ge3$.

Consider an invocation of Procedure \textsc{Reduce-Color} on a proper $k$-edge coloring of a graph $H$. We uncolor the edges colored $k$, which form a matching, and run Algorithm~\ref{alg:rand-color-match} with the remaining palette $[k-1]$. By Observation~\ref{obs:larger-palette}, the algorithm remains valid for this palette and, whenever it succeeds, produces a proper $(k-1)$-edge coloring. Thus each invocation has the same input output guarantee as Procedure \textsc{Reduce-Color} but with better runtime bounds. Consequently, conditioned on the success of all color-matching calls, the correctness analysis of~\cite{elkin2026efficient} carries over unchanged. In particular, after the $h$ levels the coloring uses either $\Delta+1$ or $\Delta+2$ colors, and in the latter case one additional reduction yields a proper $(\Delta+1)$-edge coloring.

We now analyze the work and span. At level $t\in\{0,\ldots,h-1\}$, every graph $H$ processed at this level is the union of $2^{t+1}$ initial subgraphs, each of maximum degree at most two. Hence $\Delta(H)\le 2^{t+2}$. Moreover, before the reduction its coloring uses at most $2^{t+2}+2$ colors. Therefore, after uncoloring the last color class, the remaining palette has size at most $2^{t+2}+1$. Therefore, after uncoloring one color, let $D$ denote the size of the remaining palette, so that $D\le2^{t+2}+1$. Setting $\hat{\Delta}=D$, we have $\hat{\Delta}\le2^{t+2}+1=O(2^t)$. Thus, by Observation~\ref{obs:larger-palette} and Theorem~\ref{thm:randomized-color-matching}, a color-matching call on a graph $H$ with $m_H$ edges at level $t$ costs $ O(m_H4^t\log^2 n\log\Delta) $ work and $ O(4^t\log^2 n\log\Delta) $ span, where constant factors are suppressed and $\log\hat{\Delta}=O(\log\Delta)$.

The graphs processed at any fixed level are pairwise edge-disjoint, so their numbers of edges sums to $m$. Since their color-matching calls are performed in parallel, level $t$ costs $O(m4^t\log^2 n\log\Delta)$ work and $O(4^t\log^2 n\log\Delta)$ span. Since  $h=\left\lceil\log\left\lceil \Delta/ 2 \right\rceil\right\rceil$ and $2^h=O(\Delta)$, we have
\[
    \sum_{t=0}^{h-1}4^t = O(4^h) = O(\Delta^2).
\]
Hence all levels together use $O(m\Delta^2\log^2 n\log\Delta)$ work and $O(\Delta^2 \log^2 n \log\Delta)$ span. The possible one additional call to the algorithm on the full graph has the same asymptotic bounds. The initial partition into degree-two subgraphs and their colorings require $O(m\log^2 n\log\Delta)$ work and $O(\log^2 n\log\Delta)$ span, and are therefore dominated.
%%%

Finally, it remains to consider the randomness introduced by the color-matching calls. Since their framework is deterministic the randomness appears only because we replace each invocation of Procedure \textsc{Reduce-Color} by our randomized Algorithm~\ref{alg:rand-color-match}. We want to show that all of these randomized calls succeed simultaneously with high probability. There are initially $2^h$ subgraphs. At level $t\in\{0,\ldots,h-1\}$, the number of merged graphs, and hence the number of color-matching calls, is $2^h/2^{t+1}$. Therefore, the total number of color-matching calls made during the $h$ levels is
\[
\begin{aligned}
    \sum_{t=0}^{h-1}\frac{2^h}{2^{t+1}}
    &= 2^h\sum_{t=0}^{h-1}\frac{1}{2^{t+1}}\\
    &= 2^h\left(1-\frac{1}{2^h}\right)\\
    &= 2^h-1.
\end{aligned}
\]
There is also possibly one additional color-matching call after the final merge, so the total number of calls is at most $2^h$. Since $h=\left\lceil\log\left\lceil \Delta/ 2\right\rceil\right\rceil$, we have $2^h < 2\left\lceil \Delta/2 \right\rceil \le \Delta+1$. And since $G$ is simple, $\Delta<n$, and therefore the entire execution contains less than $n$ randomized color-matching calls.

Fix an arbitrary constant $c>0$. By Theorem~\ref{thm:randomized-color-matching}, the failure probability of each color-matching call can be made at most $ n^{-(c+1)}$ by choosing the constant in the number of rounds sufficiently large. Consider an ordering of the color-matching calls as $1,\ldots,\tau$, where $\tau\le 2^h<n$. For each $j\in[\tau]$, let $\mathcal{E}_j$ be the event that the first $j-1$ color-matching calls succeed and the $j$-th call fails. Conditioned on the success of the first $j-1$ calls, the $j$-th call receives a valid instance of the color-matching problem, then $\Pr[\mathcal{E}_j] \le  n^{-(c+1)}$. The event that some color-matching call fails is the union of the events $\mathcal{E}_1,\ldots,\mathcal{E}_\tau$. Therefore,
\[
\begin{aligned}
    \Pr[\text{a call fails}]
    &= \Pr\left[\bigcup_{j=1}^{\tau}\mathcal{E}_j\right]\\
    &\le \sum_{j=1}^{\tau}\Pr[\mathcal{E}_j]\\
    &\le \tau\,n^{-(c+1)}\\
    &< n\cdot n^{-(c+1)}\\
    &= n^{-c}.
\end{aligned}
\]
Thus all color-matching calls succeed simultaneously with probability at least $1-n^{-c}$ so the algorithm succeeds with high probability.
\end{proof}

\begin{theorem}\label{thm:deter-delta_plus_one}
There is a deterministic parallel algorithm that, given an $n$-vertex, $m$-edge simple graph $G$ of maximum degree $\Delta$, outputs a proper $(\Delta+1)$-edge coloring in $O(m\Delta^3\log^2 n)$ work and $O(\Delta^3\log^4 n)$ span on the ARBITRARY CRCW PRAM.
\end{theorem}

\begin{proof}
We use the same framework and correctness argument as in Theorem~\ref{thm:rand-delta-plus-one}, replacing Algorithm~\ref{alg:rand-color-match} by the deterministic Algorithm~\ref{alg:det-color-matching}. As described above, for each invocation we uncolor a color class of minimum cardinality, so that the resulting uncolored matching satisfies the size requirement of the deterministic color-matching algorithm. By Observation~\ref{obs:larger-palette}, each such call has the same input-output effect as Procedure \textsc{Reduce-Color}, so the correctness argument carries over unchanged and is now deterministic. It remains only to bound the work and span.

At level $t\in\{0,\ldots,h-1\}$, setting $\hat{\Delta}=D$ gives $\hat{\Delta}=O(2^t)$. By Observation~\ref{obs:larger-palette} and Theorem~\ref{thm:deterministic}, a deterministic color-matching call on a graph $H$ with $m_H$ edges therefore requires $O(m_H8^t\log^2 n)$ work and $O(8^t\log^4 n)$ span. Since the graphs at each level are edge-disjoint, their edge counts sum to $m$, and all calls at the same level run in parallel. Hence level $t$ costs $O(m8^t\log^2 n)$ work and $O(8^t\log^4 n)$ span.

Since $2^h=O(\Delta)$, $\sum_{t=0}^{h-1}8^t=O(8^h)=O(\Delta^3)$. Therefore, all levels together cost $O(m\Delta^3\log^2 n)$ work and $O(\Delta^3\log^4 n)$ span. The possible additional color-matching call on the full graph and the initial steps of the framework are within the same bounds.
\end{proof}

\bibliographystyle{plain}
\bibliography{references}

\appendix
\section{Additional Details}\label{sec:appendix}

\begin{proof}[\textbf{Proof of Lemma~\ref{lem:parallel-paths}}]
We use the same connected-components routine and endpoint-identification idea as Elkin and Khuzman~\cite{elkin2026efficient}. We apply the deterministic CRCW PRAM connected-components algorithm of Shiloach and Vishkin~\cite{vishkin1982log}. After $O(\log n)$ parallel rounds, we can produce for every connected component a rooted tree whose vertices all have the same root. We define $\operatorname{comp}(v)$ to be the identifier of the root of the component containing $v$. Thus $ \operatorname{comp}(u)=\operatorname{comp}(v) $ if and only if $u$ and $v$ belong to the same connected component. Similarly we can assign identifier to edges, after the vertex labels are known, every edge $e=(u,v)$  stores $ \operatorname{comp}(e):=\operatorname{comp}(u)=\operatorname{comp}(v) $ independently in $O(1)$ work per edge and $O(1)$ span.

On a graph with $n$ vertices and $m_H$ edges, the connected-components algorithm requires $O((n+m_H)\log n)$ work and $O(\log n)$ span. Since $\Delta(H)\le2$, we have $m_H\le n$. Therefore, the connected-components step costs $O(n\log n)$ work and $O(\log n)$ span. It remains to identify the endpoints of every path component containing at least one edge. We use two arrays $L_1$ and $L_2$, each of size $n$ and indexed by component identifiers, and initialize all their entries to $\bot$ in parallel, which requires $O(n)$ work and $O(1)$ span. Since $H$ is stored by adjacency lists, the degree of each vertex is available from its adjacency-list representation, so every vertex can determine in $O(1)$ work whether its degree is one. 

In the first parallel phase, every degree-one vertex $v$ writes its identifier to $ L_1[\operatorname{comp}(v)]$. A path component containing at least one edge has exactly two degree-one vertices. Thus at most two processors write to any entry of $L_1$, and under ARBITRARY CRCW one of these two identifiers is retained. In the second parallel phase, every degree-one vertex $v$ with $ v\neq L_1[\operatorname{comp}(v)] $ writes its identifier to $ L_2[\operatorname{comp}(v)]$. Consequently, $L_1$ and $L_2$ contain the two endpoints of every path component containing at least one edge. Cycle components have no degree-one vertices, and isolated vertices have degree zero, so neither produces an endpoint entry.

The endpoint-identification procedure consists of initialization and two parallel write phases, and hence costs $O(n)$ work and $O(1)$ span. Together with the connected components computation, the total complexity is $O(n\log n)$ work and $O(\log n)$ span.
\end{proof}

\paragraph{\textbf{Data Structure.}}
We use a data representation similar to that of Elkin and Khuzman~\cite{elkin2026efficient}, with one small modification described later. The graph is stored by adjacency lists, and every vertex and edge has a unique identifier. For every edge $e$, we store its current color $\chi(e)\in[\Delta+1]\cup{\bot}$. We also maintain the tables $\operatorname{Color2Edge}(v,c)$ and $\operatorname{Edge2Color}(v,e)$ as in~\cite{elkin2026efficient}: given a color $c$, $\operatorname{Color2Edge}(v,c)$ determines in $O(1)$ work whether $v$ has an incident edge colored $c$ and, if so, returns that edge, while $\operatorname{Edge2Color}(v, e)$ returns in $O(1)$ work the current color of an edge incident to $v$. In particular, whether a given color is missing at $v$ can be tested in $O(1)$ work.

The data structure can be initialized once at the beginning of a color-matching invocation. Since isolated vertices play no role in edge coloring, we restrict attention to the non-isolated vertices, thus $n=O(m)$. We can initialize the $\operatorname{Color2Edge}$ and $\operatorname{Edge2Color}$ tables by scanning all colored edges in parallel, and then scan the $\Delta+1$ colors at each vertex to identify the required missing colors. For each vertex, all $\Delta+1$ colors are tested in parallel, and a parallel reduction is used to select the required one or two missing colors in $O(\log\Delta)$ span. This takes $O(n\Delta+m)=O(m\Delta)$ work and $O(\log\Delta)$ span. It is important, however, that we do not rebuild this data structure from scratch in every round. An $O(m\Delta)$-work initialization is within the scale of the per-round work in the algorithm of Elkin and Khuzman, but repeating such an initialization in every round of our algorithms would introduce an additional factor of $\Delta$ into our total work. We therefore maintain the data structure throughout the execution and explicitly update it whenever the coloring changes. The only operations that modify colors are path flips and fan rotations together with the final coloring of a u-edge.

For every vertex, we maintain at least one color that is currently missing at that vertex. Our only modification is that every endpoint of a currently uncolored edge stores two distinct missing colors rather than one. Since the uncolored edges form a matching, every endpoint $v$ of an uncolored edge has at most $\Delta-1$ colored incident edges and therefore has at least two distinct missing colors in $[\Delta+1]$. At the center $u$ of a u-edge, one of these colors is distinguished as the designated missing color $\varphi(u)$, while the other is kept as a reserve. This second stored color is used when the u-edge is eventually colored: when the edge is colored, the other color remains missing and can be updated as the designated missing color without scanning the palette for a replacement. Once the u-edge has been colored, its endpoints are no longer incident to an uncolored edge, so only one stored missing color is needed at those vertices. At vertices not incident to an uncolored edge, a single stored missing color suffices.

\begin{proof}[\textbf{Proof of Lemma~\ref{lem:parallel-fans}}]
Lemma~2 of Elkin and Khuzman~\cite{elkin2026efficient} shows that a maximal fan centered at $u$ for a given uncolored edge can be constructed in $O(\deg(u)\log\Delta)$ work and $O(\log\Delta)$ span. Since their construction allows an arbitrary missing color at the center, we use the designated missing color $\varphi(u)$.

We first specify how the constructed fans are represented. For every designated center $u$, we dedicate an array of length $O(\deg(u))$. Since the uncolored edges form a matching, their designated centers are distinct, and hence the total space required by these arrays is in order of $ \sum_{u\in W}\deg(u)\le \sum_{v\in V}\deg(v)=2m$. Let $ F=(u,\varphi(u)),(v_1,c_1),\ldots,(v_k,c_k)$ be the maximal fan constructed at $u$. We keep this indexed sequence in the array associated with $u$. For every $i\in[k]$, the $i$-th entry stores the leaf $v_i$, the identifier of the fan edge $(u,v_i)$, and the designated missing color $ c_i=\varphi(v_i)$. We also store the fan length $k$. Thus, after construction, every fan vertex, fan edge, and associated missing color can be accessed by its index in $O(1)$ work and $O(1)$ span.

We next classify the fan: let $c_k$ be the designated missing color of its last leaf. We can test in $O(1)$ work and $O(1)$ span whether $c_k$ is missing at $u$. If it is, then $F$ is trivial. Otherwise, by maximality of the fan, $c_k=c_j$ for one $j<k$, and hence $F$ is nontrivial with required pair $ (\varphi(u),c_k)$. We store index $j<k$ with $c_j=c_k$, so that the corresponding fan prefix can be identified later. Once $c_k$ is known, all leaves with indices $i<k$ can test in parallel whether $c_i=c_k$. Since one such index exists for every nontrivial fan, finding and storing $j$ requires $O(k)$ work and $O(1)$ span.

All fans can be constructed simultaneously even if they overlap in the original graph. The construction only reads the current coloring, the stored missing colors, and the coloring data structures, while the output of the fan centered at $u$ is written only to the array associated with $u$. Therefore, since $W$ denotes the set of centers, $\sum_{u\in W}\deg(u) \le \sum_{v\in V}\deg(v) = 2m$. Summing the work of the fan-construction routine over all centers requires $ O(m\log\Delta) $ work, while all fan constructions run simultaneously in $O(\log\Delta)$ span.

Moreover, if $k_u$ denotes the length of the fan centered at $u$, then $k_u\le\deg(u)$. Hence the additional work for storing the indexed fan representations, testing whether the terminal color is missing at the center, recording the trivial/nontrivial classification and required pair, and, for every nontrivial fan, finding and storing an index $j$ with $c_j=c_k$, is $O(m)$ in total. These operations consist of a constant number of parallel reads, comparisons, and writes once the fans have been constructed, and therefore add only $O(1)$ span.

Thus all maximal fans can be constructed and stored, classified as trivial or nontrivial, and the required pair of every nontrivial fan can be determined in $O(m\log\Delta)$ work and $O(\log\Delta)$ span.
\end{proof}

\begin{proof}[\textbf{Proof of Lemma~\ref{lem:parallel-fan-rotations}}]
Consider a trivial fan $ F=(u,\alpha),(v_1,c_1), \ldots,(v_k,c_k)\in\mathcal{F}$, by the indexed representation of the fan, each leaf $v_i$, fan edge $e_i=(u,v_i)$, and color $c_i$ can be accessed in $O(1)$ work and $O(1)$ span. Before the rotation, $e_1$ is uncolored and $\chi(e_i)=c_{i-1}$ for every $i\ge2$, while $c_k$ is missing at $u$. Hence the rotation and final coloring can be performed directly by setting $\chi(e_i)\leftarrow c_i$ for all $i\in[k]$ in parallel. This recoloring is proper: $c_1,\ldots,c_{k-1}$ are the distinct old colors of the fan edges incident to $u$, $c_k$ is missing at $u$, and each $c_i$ is missing at $v_i$. In particular, $(u,v_k)$ is colored with $c_k$.

The data structures can be updated in a constant number of parallel phases. For every recolored edge, its color and $\operatorname{Edge2Color}$ entries are updated directly. For $\operatorname{Color2Edge}$, we first delete all entries corresponding to old fan-edge colors and then insert all entries corresponding to the new colors. Within each phase, the relevant entries are distinct because the fan-edge colors at the center are distinct and the leaves are distinct. The stored missing colors are updated locally: for $i\ge2$, if the stored missing color $c_i$ at $v_i$ is used by the rotation, it is replaced by the newly missing color $c_{i-1}$; at $v_1$ and $u$, one of the two stored missing colors remains missing after the uncolored edge is colored and is retained. Thus all data-structure updates require $O(1)$ work per involved edge or vertex and $O(1)$ span.

Since the fans in $\mathcal{F}$ are pairwise vertex-disjoint, updates belonging to different fans do not interfere. The total number of fan edges involved is at most $m$. Therefore, all fans can be rotated and their corresponding uncolored edges colored in $O(m)$ work and $O(1)$ span.
\end{proof}

\begin{proof}[\textbf{Proof of Theorem~\ref{thm:conc}}]
Recall that the set $W$ and the subpalette $\Gamma$ are fixed before any of the phase coins are sampled. For $0\le i\le b$, define $S_i$ and $m_i$ as in the proof of Theorem~\ref{thm:expectation-same-color}. Thus $m_0=|C|$ and $D=m_b$. We first show that, for each fixed $i\in[b]$, the following holds with probability at least $1-e^{-\lambda}$,
\begin{equation}\label{eq:rec}
  m_i \le \frac{1}{2} \,m_{i-1}+\sqrt{2\lambda m_{i-1}}+1.
\end{equation}

Fix $i\in[b]$. Let $\mathcal{G}_{i-1}$ be the $\sigma$-field generated by $\mathcal{F}_{i-1}$ with $u$'s phase-$i$ coin toss $\xi_{i,p_i(u),K_i(u)}$; then conditioning on $\mathcal{G}_{i-1}$ fixes the entire outcome of the first $i-1$ phases as well as the value of $u$'s coin in phase $i$, while leaving every other phase-$i$ coin fair and independent. Since $\mathcal{F}_{i-1} \subseteq \mathcal{G}_{i-1}$, every $\mathcal{F}_{i-1}$-measurable quantity is also determined by $\mathcal{G}_{i-1}$. In particular, by Lemma~\ref{lem:agree-and-measurable}, the set $S_{i-1}$ and the value $m_{i-1}$ are $\mathcal{F}_{i-1}$-measurable, and hence also fixed once we condition on $\mathcal{G}_{i-1}$.

If $m_{i-1}=0$, then $S_i\subseteq S_{i-1}=\emptyset$ and hence $m_i=0$, so Inequality~\eqref{eq:rec} holds trivially; we therefore assume $m_{i-1}\ge1$. Recall that at most one center of $S_{i-1}$ is tied to $u$ in phase $i$. If such a center exists, denote it by $w^*$ and set $T_{i-1}:=S_{i-1}\setminus\{w^*\}$; otherwise set $T_{i-1}:=S_{i-1}$. As before, in either case $w^*$ contributes at most $1$ to $m_i$, so
\begin{equation}\label{eq:mi-split}
  m_i\,\le\,1+\bigl|\{w\in T_{i-1}:w\in S_i\}\bigr|.
\end{equation}

If $T_{i-1}=\emptyset$, then the right-hand side of the above inequality is at most $1$, and therefore inequality~\eqref{eq:rec} again holds trivially; so we assume $T_{i-1}\neq\emptyset$.

Partition $T_{i-1}$ into classes according to the relation of being tied in phase $i$, with each untied center forming a singleton class. Thus two centers belong to the same class if and only if they are tied in phase $i$, and every class has size at most two. By Lemma~\ref{lem:agree-and-measurable} the partition is $\mathcal{F}_{i-1}$-measurable as well, hence determined under $\mathcal{G}_{i-1}$. Let the classes be indexed by $j\in[r]$, with sizes $s_j\in\{1,2\}$, so that
\[
  \sum_{j=1}^r s_j=|T_{i-1}|\le m_{i-1}.
\]

Since no center of $T_{i-1}$ is tied to $u$, none of the classes uses the same coin as $u$. Moreover, different classes use distinct phase-$i$ coins. Hence the class coins are mutually independent and are also independent of $\mathcal{G}_{i-1}$. For each $j$, let $Y_j:=\bigl|\{w\in\text{class }j:w\in S_i\}\bigr|\in\{0,1,2\}$. Given $\mathcal{G}_{i-1}$, the value of $Y_j$ is determined by the single coin that the members of class $j$ read; as these class coins are independent, so are $Y_1,\dots,Y_r$ conditioned on $\mathcal{G}_{i-1}$. Since the classes partition $T_{i-1}$ and from \eqref{eq:mi-split} we have
\begin{equation}\label{eq:mi-sum}
  m_i\,\le 1+\bigl|\{w\in T_{i-1}:w\in S_i\}\bigr| = \,1+\sum_{j=1}^r Y_j .
\end{equation}

Fix $w\in T_{i-1}$. Since $w$ is not tied to $u$, its phase-$i$ coin $\xi_{i,p_i(w),K_i(w)}$ is distinct from $u$'s coin, hence fair and independent of $\mathcal{G}_{i-1}$. Under $\mathcal{G}_{i-1}$ the bits $(\varphi_{i-1}(w))^{(i)}$ and $(\varphi_{i-1}(u))^{(i)}$ and the coin $\xi_{i,p_i(u),K_i(u)}$ are all fixed, so the survival condition, $(\varphi_{i-1}(w))^{(i)}\oplus\xi_{i,p_i(w),K_i(w)} =(\varphi_{i-1}(u))^{(i)}\oplus\xi_{i,p_i(u),K_i(u)}$ holds if and only if $w$'s coin $\xi_{i,p_i(w),K_i(w)}$ takes one particular value. Since this coin is fair and independent of $\mathcal{G}_{i-1}$, this event occurs with probability $\frac{1}{2}$. Hence $\Pr[w\in S_i\mid\mathcal{G}_{i-1}]=\tfrac12$ for every $w\in T_{i-1}$. We can now bound the conditional expectation of the number of centers of $T_{i-1}$ that survive phase $i$ (i.e., that still agree with $u$ on the first $i$ bits):
\begin{equation}\label{eq:mean-half}
  \mathbb{E}\,\left[\sum_{j=1}^r Y_j\ \middle|\ \mathcal{G}_{i-1} \right] =\sum_{w\in T_{i-1}} \Pr[w\in S_i\mid \mathcal{G}_{i-1}] =\frac{|T_{i-1}|}{2} \le\frac{m_{i-1}}{2}.
\end{equation}

The variables $Y_1,\dots,Y_r$ are independent given $\mathcal{G}_{i-1}$, and each $Y_j$ takes values in an interval of length $s_j$, the size of class $j$. Before applying Hoeffding's inequality we first need the sum of the squared interval lengths: since $s_j\le2$ for every $j$, and $\sum_j s_j=|T_{i-1}|\le m_{i-1}$,
\[
  \sum_{j=1}^r s_j^{\,2} \,\le\,2\sum_{j=1}^r s_j  \,\le\,2\,m_{i-1}.
\]

Recall that, conditioned on $\mathcal{G}_{i-1}$, the variables $Y_1,\dots,Y_r$ are independent, and each $Y_j$ takes values in $[0,s_j]$. Let $\mu:= \mathbb{E}[\sum_{j}Y_j \mid \mathcal{G}_{i-1}]$. From Hoeffding's inequality we get for every $t\ge0$,
\[
  \Pr \, \left[ \sum_{j=1}^r Y_j\ge\mu+t \, \middle|\, \mathcal{G}_{i-1} \right]
  \le \exp\,\left(-\frac{2t^{2}}{\sum_{j}s_j^{2}}\right).
\]
By \eqref{eq:mean-half} we have $\mu\le m_{i-1}/2$, so the threshold $m_{i-1}/2+t$ is at least $\mu+t$, and the tail probability only decreases. Combining with the fact that $\sum_j s_j^{2} \le 2\,m_{i-1}$,
\begin{equation}
  \Pr\,\left[ \sum_{j=1}^r Y_j \ge \frac{m_{i-1}}{2}+t \, \middle|\, \mathcal{G}_{i-1} \right] \le
  \exp\,\left( -\frac{2t^{2}}{\sum_j s_j^{2}} \right) \le
  \exp\,\left( -\frac{t^{2}}{m_{i-1}} \right).
\end{equation}

Setting $t:=\sqrt{2\,\lambda\,m_{i-1}}$, the right-hand side becomes
\[
  \exp\, \left( -\frac{t^{2}}{m_{i-1}} \right) =\exp\, \left(- \frac{2\, \lambda\,m_{i-1}}{m_{i-1}}\right)
  =e^{-2\lambda} \le e^{-\lambda},
\]
so with probability at least $1-e^{-\lambda}$, conditioned on $\mathcal{G}_{i-1}$, we have
\[
  \sum_{j=1}^r Y_j \,<\,\frac{m_{i-1}}{2}+\sqrt{2\lambda\,m_{i-1}} .
\]

Combining this with \eqref{eq:mi-sum} gives us
\[
  m_i\,\le\,1+\sum_{j=1}^r Y_j \, < \,\frac{m_{i-1}}{2}+\sqrt{2\lambda\,m_{i-1}}+1,
\]
which is what we wanted to show in inequality~\eqref{eq:rec}. Let $E_i$ denote the event that inequality~\eqref{eq:rec} holds. We have shown $\Pr[E_i\mid\mathcal{G}_{i-1}]\ge 1-e^{-\lambda}$. Since $\mathcal{G}_{i-1}$ adds to $\mathcal{F}_{i-1}$ only the value of $u$'s coin, and this bound holds for either of its two values, it also holds conditioned on $\mathcal{F}_{i-1}$ alone. Since the unconditional probability $\Pr[E_i]$ is the average of $\Pr[E_i\mid\mathcal{F}_{i-1}]$ over the outcomes of the first $i-1$ phases, and each such conditional probability is at least $1-e^{-\lambda}$, we conclude $\Pr[E_i]\ge 1-e^{-\lambda}$. Finally, by a union bound over the $b$ phases, inequality~\eqref{eq:rec} holds for every $i\in[b]$ simultaneously with probability at least $1-b\,e^{-\lambda}$:
\[
  \Pr \, \left[ \, \bigcap_{i=1}^b E_i\,\right] \,\ge\,
  1-\sum_{i=1}^b\Pr[ \overline{E_i}] \,\ge\,
  1-\sum_{i=1}^b e^{-\lambda} \,=\, 1-b\,e^{-\lambda}.
\]

From the above, inequality~\eqref{eq:rec} holds simultaneously for every phase $i\in[b]$ with probability at least $1-b\,e^{-\lambda}$. Assume from now on that this is the case; we want to show it implies $D\le \kappa (\Delta +\lambda)$, which then holds with the same probability. Let us first convert inequality~\eqref{eq:rec} into a linear recursion for $\sqrt{m_i}$. Since $\lambda \ge1$,
\begin{align}
  \left(\sqrt{\frac{m_{i-1}}{2}}+\sqrt{\lambda} \right)^{\,2}
  &=\frac{m_{i-1}}{2}
    +2\sqrt{\frac{m_{i-1}}{2}}\,\sqrt{\lambda}
    +\lambda \notag\\
  &=\frac{m_{i-1}}{2}+\sqrt{2\lambda\,m_{i-1}}+\lambda \notag\\
  &\ge\frac{m_{i-1}}{2}+\sqrt{2\lambda\,m_{i-1}}+1
  \,\ge\,m_i, \label{eq:sqrt-step}
\end{align}
where the first inequality uses $\lambda \ge1$, and the last inequality is from \eqref{eq:rec}. Then taking the square roots of \eqref{eq:sqrt-step} we have
\begin{equation}\label{eq:sqrt-rec}
  \sqrt{m_i} \, \le \,\frac{1}{\sqrt{2}} \sqrt{m_{i-1}}+ \sqrt{\lambda}.
\end{equation}

We now bound $\sqrt{m_\ell}$ in terms of its initial value $\sqrt{m_0}$. We claim that for every $0\le \ell\le b$ the following holds,
\begin{equation}\label{eq:sqrt-unroll}
  \sqrt{m_\ell} \,\le\, 2^{-\ell/2}\sqrt{m_0}+\sqrt{\lambda}\,\sum_{j=0}^{\ell-1}2^{-j/2}.
\end{equation}

We prove this by induction on $\ell$. For $\ell=0$ both sides equal $\sqrt{m_0}$, since the sum is empty. Assuming \eqref{eq:sqrt-unroll} for $\ell-1$ and applying \eqref{eq:sqrt-rec},
\begin{align*}
  \sqrt{m_\ell}
  &\,\le\, \frac{1}{\sqrt{2}} \sqrt{m_{\ell-1}}+\sqrt {\lambda}\\
  &\,\le\, \frac{1}{\sqrt{2}} \left( 2^{-(\ell-1)/2} \sqrt{m_0} +\sqrt{ \lambda} \sum_{j=0}^{\ell-2}2^{-j/2} \right)+\sqrt{ \lambda}\\
  &\,=\, \frac{2^{-(\ell-1)/2}}{\sqrt{2}} \, \sqrt{m_0} +\sqrt{ \lambda} \left( \frac{1}{\sqrt{2}} \sum_{j=0}^{\ell-2}2^{-j/2}+1 \right)\\
  &\,=\, 2^{-\ell/2} \sqrt{m_0} +\sqrt{ \lambda} \left (\sum_{j=1}^{\ell-1}2^{-j/2}+1 \right)\\
  &\,=\, 2^{-\ell/2} \sqrt{m_0} +\sqrt{ \lambda} \sum_{j=0}^{\ell-1}2^{-j/2},
\end{align*}
This completes the induction. Set $\ell=b$ to bound the sum on the right hand side of the inequality~\eqref{eq:sqrt-unroll}:
\[
  \sum_{j=0}^{b-1}2^{-j/2}
  \le\sum_{j=0}^{\infty}2^{-j/2}
  =\frac{1}{1-1/\sqrt{2}}
  =2+\sqrt{2},
\]

Recall that $m_0=|C|$, $|C|\le\Delta^2$, and $|\Gamma|=k=2^b>(\Delta+1)/2$. Then
\begin{align*}
  \sqrt{m_b}
  &\le 2^{-b/2}\sqrt{|C|}+(2+\sqrt{2})\sqrt{\lambda}\\
  &= \sqrt{\frac{|C|}{2^b}}+(2+\sqrt{2})\sqrt{\lambda}\\
  &\le \sqrt{\frac{\Delta^2}{k}}+(2+\sqrt{2})\sqrt{\lambda}\\
  &< \sqrt{\frac{2\Delta^2}{\Delta+1}}+(2+\sqrt{2})\sqrt{\lambda}\\
  &\le \sqrt{2\Delta}+(2+\sqrt{2})\sqrt{\lambda}
\end{align*}
and therefore
\begin{align*}
  D=m_b
  &< \left(\sqrt{2\Delta}+(2+\sqrt{2})\sqrt{\lambda}\right)^2\\
  &=2\Delta+2(2+\sqrt{2})\sqrt{2\Delta\lambda} +(2+\sqrt{2})^2\lambda\\
  &\le 2\Delta+(2+\sqrt{2})(2\Delta+\lambda) +(2+\sqrt{2})^2\lambda\\
  &<16(\Delta+\lambda).
\end{align*}
Thus, taking $\kappa:=16$, whenever inequality~\eqref{eq:rec} holds for all phases, we have $D<\kappa(\Delta+\lambda)$.
Since this happens with probability at least $1-b\,e^{-\lambda}$, we conclude
\[
  \Pr \bigl[D \ge \kappa(\Delta+\lambda) \bigr] \le b\,e^{-\lambda}.
\]

Set $\lambda:=c\log n$ for a constant $c>1$. Since $\Delta+1\le n$ and $b=\lfloor\log_2(\Delta+1)\rfloor$, we have $b\le\log_2 n$, so the failure probability is
\[
  b\, e^{-\lambda} \, \le \, (\log_2 n)\,e^{-c\log n}
  \, = \, \frac{\log_2 n}{n^{c}}
\]
since $c>1$ this is at most $1/\mathrm{poly}(n)$. Hence, with probability at least $1-1/\mathrm{poly}(n)$,
\[
    D\,\le\,\kappa(\Delta+\lambda)\,=\,\kappa\bigl(\Delta+c\log n\bigr)\,=\,O(\Delta+\log n)
\]
Finally, if $\Delta=\Omega(\log n)$, then $D=O(\Delta)$, again with probability at least $1-1/\mathrm{poly}(n)$.
\end{proof}

\end{document}